\documentclass[a4paper,11pt,oneside,reqno,final]{amsart}
\numberwithin{equation}{section}

\usepackage{graphicx}
\usepackage[latin1]{inputenc}
\usepackage{geometry}
\usepackage{latexsym}
\usepackage{amstext} \usepackage{amsmath} \usepackage{amssymb}
\usepackage{psfrag}
\usepackage{graphics}
\usepackage{color}

\renewcommand{\phi}{\varphi}

\newtheorem{remark}{Remark}
\newtheorem{theorem}{Theorem}
\newtheorem{lemma}{Lemma}
\newtheorem{proposition}{Proposition}

 2
 3
 4
 2 

\begin{document}

\title[Asymptotics of a PII Fredholm determinant]{Asymptotics of a Fredholm determinant involving the second Painlev\'e
transcendent}

\author{Thomas Bothner}
\address{Department of Mathematical Sciences,
Indiana University-Purdue University Indianapolis,
402 N. Blackford St., Indianapolis, IN 46202, U.S.A.}
\email{tbothner@iupui.edu}

\author{Alexander Its}
\address{Department of Mathematical Sciences,
Indiana University-Purdue University Indianapolis,
402 N. Blackford St., Indianapolis, IN 46202, U.S.A.}
\email{itsa@math.iupui.edu}

\thanks{This work was supported in part by the National Science Foundation (NSF) Grant DMS-1001777.}

\date{\today}

\begin{abstract}
We study the determinant $\det(I-K_{\textnormal{PII}})$ of an integrable Fredholm operator $K_{\textnormal{PII}}$ acting on the interval $(-s,s)$ whose kernel is constructed out of the $\Psi$-function associated with the Hastings-McLeod solution of the second Painlev\'e equation. This Fredholm determinant describes the critical behavior of the eigenvalue gap probabilities of a random Hermitian matrix chosen from the Unitary Ensemble in the bulk double scaling limit near a quadratic zero of the limiting mean eigenvalue density. Using the Riemann-Hilbert method, we evaluate the large $s$-asymptotics of $\det(I-K_{\textnormal{PII}})$.

\end{abstract}

\maketitle

%---------------------------------------------------------------------------------------------------------------------------------------------------------------------------------------------------------
\section{Introduction and statement of the results}\label{sec1}
Let $\mathcal{M}(n)$ be a unitary ensembles of random $n\times n$ Hermitian matrices $M=(M_{ij})=\overline{M}^t$ equipped with the probability measure,
\begin{equation}\label{UEmodel}
	P^{(n,N)}(M)dM = ce^{-N\textnormal{tr}V(M)}dM,\hspace{1cm} c\int\limits_{\mathcal{M}(n)}e^{-N\textnormal{tr}V(M)}dM=1.
\end{equation}
Here $dM$ denotes the Haar measure on $\mathcal{M}(n)\simeq\mathbb{R}^{n^2}$, $N$ is a fixed integer and the potential $V:\mathbb{R}\rightarrow\mathbb{R}$ is assumed to be real analytic satisfying the growth condition
\begin{equation}\label{potregul}
	\frac{V(x)}{\ln(x^2+1)}\rightarrow\infty\hspace{0.5cm}\ \textnormal{as}\ |x|\rightarrow\infty.
\end{equation}
The principal object of the analysis of the model  is the  statistics of eigenvalues of the  matrices from $\mathcal{M}(n)$. A classical  fact  \cite{D}, \cite{M} is that the eigenvalues form a  determinantal random point process 
with the  kernel
\begin{equation}\label{correlkernel}
	K_{n,N}(x,y)=e^{-\frac{N}{2}V(x)}e^{-\frac{N}{2}V(y)}\sum_{i=0}^{n-1}p_i(x)p_i(y),
\end{equation}
where $p_j(x)$ are polynomials orthonormal with respect to the weight $e^{-NV(x)}$,
\begin{equation}\label{ortpol}
\int\limits_{\mathbb{R}}p_i(x)p_j(x)e^{-NV(x)}dx=\delta_{ij}, \quad p_j(x) = \kappa_jx^j + ...
\end{equation}
In particular, one of the basic statistical characteristics, the  {\it gap probability}, 
\begin{equation*}
	E_{n,N}(s) = \textnormal{Prob}\big(M\in\mathcal{M}(n)\ \textnormal{has no eigenvalues in the interval}\ (-s,s),\ s>0\big)
\end{equation*}
is given by the formula,
\begin{eqnarray*}
	E_{n,N}(s) &=& \sum_{j=0}^n\frac{(-1)^j}{j!}\int\limits_{-s}^s\cdots\int\limits_{-s}^s\det\big(K_{n,N}(x_k,x_l)_{k,l=1}^j\big)dx_1\cdots dx_j\\
	&\equiv& \det(I-K_{n,N}),
\end{eqnarray*}
where $K_{n.N}$ is the trace class operator 
acting on $L^2\big((-s,s),dx\big)$ with kernel $K_{n.N}(x,y)$. 
%A key issue is the large $n$, $N$  limit of the  kernel $K_{n.N}(x,y)$ and of  the Fredholm 
%determinant $\det(I-K_{n,N})$

Assumptions \eqref{potregul} on the potential $V(x)$ ensure \cite {DKM} (see also \cite{D}
for more on the history of the subject) that the mean eigenvalue density $\frac{1}{n}K_{n,N}(x,x)$ has a limit,
\begin{equation}\label{density}
\lim_{\substack{n,N\rightarrow\infty\\
	\frac{n}{N}\rightarrow 1}}\frac{1}{n}K_{n,N}(x,x) = \rho_V(x) \geq 0,
\end{equation}
whose support, $\Sigma_V\equiv \overline{\left\{ x\in \mathbb{R}: \rho_V(x) > 0\right\}}$, is a finite
union of intervals (simultaneously,  $\rho_V(x)$ defines the density of the  equilibrium measure
for the logarithmic potentials in the presence of the external  potential $V$). The limiting density 
$\rho_V(x)$ is determined by the potential $V(x)$. At the same time, the local statistics 
of eigenvalues in the large $n, N$ limit  satisfies the so-called {\it universality property}, i.e.
it is determined only by the local characteristics of the eigenvalue density $\rho_V$ (\cite{PS}, \cite{BI1},\cite{DKMVZ}). For instance, let us choose a regular point $x^{\ast}\in\Sigma_V$, i.e. $\rho_V(x^{\ast})>0$. Then the \emph{bulk universality} states that
\begin{equation}\label{bulkregular}
	\lim_{n\rightarrow\infty}\frac{1}{n\rho_V(x^{\ast})}K_{n,n}\bigg(x^{\ast}+\frac{\lambda}{n\rho_V(x^{\ast})},x^{\ast}+\frac{\mu}{n\rho_V(x^{\ast})}\bigg) = K_{\sin}(\lambda,\mu)\equiv \frac{\sin\pi(\lambda-\mu)}{\pi(\lambda-\mu)}
\end{equation}
uniformly on compact subsets of $\mathbb{R}$, which in turn implies \cite{DKMVZ} that for a regular point 
$x^{\ast}$,
\begin{equation*}
	\lim_{\substack{n,N\rightarrow\infty\\
	\frac{n}{N}\rightarrow 1}}\textnormal{Prob}\bigg(M\in\mathcal{M}(n)\ \textnormal{has no eigenvalues}\in \Big(x^{\ast}-\frac{s}{n\rho_V(x^{\ast})},x^{\ast}+\frac{s}{n\rho_V(x^{\ast})}\Big)\bigg)
\end{equation*}
\begin{equation}\label{empty0}
=\det(I-K_{\sin}),
\end{equation}
where $K_{\sin}$ is the trace class operator on $L^2\big((-s,s);d\lambda\big)$ with kernel $K_{\sin}(\lambda,\mu)$ given in \eqref{bulkregular}.  (This result was first obtained for GUE in the classical works of Gaudin and Dyson.)
The Fredholm determinant in the right hand side of \eqref{empty0} admits the following asymptotic
representation \cite{dyson}, 
\begin{equation}\label{dyson0}
	\ln\det(I-K_{\sin}) = -\frac{(\pi s)^2}{2}-\frac{1}{4}\ln (\pi s)+\frac{1}{12}\ln 2+3\zeta'(-1)+O\big(s^{-1}\big),\ \ s\rightarrow\infty,
\end{equation}
where $\zeta'(z)$ is the derivative of the Riemann zeta-function
(a rigorous proof without the constant term was obtained independently  by Widom  and Suleimanov - 
see \cite{DIZ} for
more historical details; a rigorous proof including the constant terms was obtained
independently in  \cite{K}, \cite{E} - see also \cite{DIKZ}). This remarkable formula yields one
of the most important results in random matrix theory, i.e. an explicit evaluation of {\it the large gap
probability}.

Equation \eqref{empty0} shows that in double scaling limits the basic statistical properties of hermitian random matrices are still expressible in terms of Fredholm determinants. This is also true for the first critical case, when $\rho_V(x)$ vanishes quadratically at an interior point $x^{\ast}\in\Sigma_V$. However, in this situation the scaling limit is more complicated \cite{BI2}, \cite{CK}. Let $\rho_V(x^{\ast})=\rho'_V(x^{\ast})=0,\rho''_V(x^{\ast})>0$ and $n,N\rightarrow\infty$ such that
\begin{equation*}
	\lim_{n,N\rightarrow\infty}n^{2/3}\bigg(\frac{n}{N}-1\bigg)=C
\end{equation*}
exists with $C\in\mathbb{R}$. Then the \emph{critical bulk universality} guarantees existence of positive constants $c$ and $c_1$ such that
\begin{equation}\label{bulkcritical}
	\lim_{n,N\rightarrow\infty}\frac{1}{cn^{1/3}}K_{n,N}\bigg(x^{\ast}+\frac{\lambda}{cn^{1/3}},x^{\ast}+\frac{\mu}{cn^{1/3}}\bigg)=K_{\textnormal{PII}}(\lambda,\mu;x)
\end{equation}
uniformly on compact subsets of $\mathbb{R}$ where the variable $x$ is the scaling parameter defined by the relation
\begin{equation*}
	\lim_{n,N\rightarrow\infty}n^{2/3}\bigg(\frac{n}{N}-1\bigg)=xc_1.
\end{equation*}
Here the limiting kernel $K_{\textnormal{PII}}(\lambda,\mu;x)$ is constructed out of the $\Psi$-function associated with a special solution of the second Painlev\'e equation. The precise description of the kernel $K_{\textnormal{PII}}(\lambda,\mu;x)$
is as follows.

Let $u(x)$ be the  Hastings-McLeod solution of the Painlev\'e II equation \cite{HM}, i.e. the unique real-valued solution to the boundary value problem
\begin{equation*}
	u_{xx}=xu+2u^3,\hspace{0.5cm} u(x)\sim\left\{
                                   \begin{array}{ll}
                                     \textnormal{Ai}(x), & \hbox{$x\rightarrow+\infty$;} \\\\
                                     \sqrt{-\frac{x}{2}}, & \hbox{$x\rightarrow-\infty$,} 
                                   \end{array}
                                 \right.
\end{equation*}
where $\textnormal{Ai}(x)$ is the Airy-function (the solution $u(x)$ is in fact uniquely determined 
by its Airy-asymptotics at $x = +\infty$). Viewing $x$, $u\equiv u(x)$ and $u_x \equiv du(x)/dx$
as real parameters, consider the $2\times 2$ system of linear ordinary differential
equations,
\begin{equation}\label{lax1}
   \frac{d\Psi}{d\lambda}=A(\lambda,x)\Psi,\
    A(\lambda,x)=-4i\lambda^2\sigma_3+4i\lambda\begin{pmatrix}
                                               0 & u \\
                                               -u & 0 \\
                                            \end{pmatrix}+\begin{pmatrix}
                                                             -ix-2iu^2 & -2u_x \\
                                                             -2u_x & ix+2iu^2 \\
                                                           \end{pmatrix}.
\end{equation}
Let $\Psi(\lambda) \equiv \Psi(\lambda,x)$
be the fundamental solution of system (\ref{lax1}) which is uniquely fixed by the asymptotic condition,
\begin{equation*}
	\Psi(\lambda,x)=\Big(I+O\big(\lambda^{-1}\big)\Big)e^{-i(\frac{4}{3}\lambda^3+x\lambda)\sigma_3},\hspace{0.5cm}\lambda\rightarrow\infty,\ 0<\textnormal{arg}\ \lambda<\pi.
\end{equation*}
Then, the kernel   $K_{\textnormal{PII}}(\lambda,\mu;x)$ is given by the formula,
\begin{equation}\label{PIIkernel}
    K_{\textnormal{PII}}(\lambda,\mu;x)\equiv K_{\textnormal{PII}}(\lambda,\mu)=\frac{1}{2\pi}\bigg(\frac{\psi_{21}(\lambda,x)\psi_{11}(\mu,x)-\psi_{21}(\mu,x)\psi_{11}(\lambda,x)}{\lambda-\mu}\bigg), %x\in\mathbb{R},
\end{equation}
where $\psi_{11}(\lambda,x)$ and $\psi_{21}(\lambda,x)$ are the entries of the matrix 
valued function $\Psi(\lambda,x)
\equiv \{\psi_{jk}(\lambda,x)\}_{j,k = 1,2}$. 
\begin{remark} The function $\Psi(\lambda,x)$ can be alternatively defined as a solution of a
certain matrix oscillatory Riemann-Hilbert problem. The  exact formulation of this Riemann-Hilbert
problem is given in the next section.
\end{remark}

In this paper we study  the Fredholm determinant
\begin{equation}\label{Pdet}
	\det(I-K_{\textnormal{PII}}),
\end{equation}
where $K_{\textnormal{PII}}$ is the trace class operator on $L^2\big((-s,s);d\lambda\big)$ with kernel \eqref{PIIkernel}.
In virtue of (\ref{bulkcritical}), this determinant replaces the sine - kernel determinant in the description of the gap-probability near the critical point $x^{\ast}$, i.e. instead of \eqref{empty0} one has that
\begin{equation*}
	\lim\textnormal{Prob}\bigg(M\in\mathcal{M}(n)\ \textnormal{has no eigenvalues}\in \Big(x^{\ast}-\frac{s}{cn^{1/3}},x^{\ast}+\frac{s}{cn^{1/3}}\Big)\bigg)
\end{equation*}	
\begin{equation}\label{emp1}
	=\det(I-K_{\textnormal{PII}}),
\end{equation}
as $n, N \to \infty$ and
\begin{equation*}
	\lim_{n,N\rightarrow\infty}n^{2/3}\bigg(\frac{n}{N}-1\bigg)=xc_1.
\end{equation*}
(The proof can  obtained in a same manner as the  proof of the similar equation $(21)$ in  \cite{DIKa} with
the help of the proper estimates from \cite{BI2}.) The main result of the present paper is the following analogue
of the Dyson formula \eqref{dyson0} for the Painlev\'e II determinant \eqref{Pdet}. 
\begin{theorem}\label{t1} Let $K_{\textnormal{PII}}$ denote the trace class operator on $L^2\big((-s,s);d\lambda\big)$ with kernel \eqref{PIIkernel}. Then as $s\rightarrow\infty$ the Fredholm determinant $\det(I-K_{\textnormal{PII}})$ behaves as 
\begin{equation*}
	\ln \det(I-K_{\textnormal{PII}})=-\frac{2}{3}s^6-s^4x-\frac{1}{2}(sx)^2 - \frac{3}{4}\ln s + \int\limits_x^{\infty}(y-x)u^2(y)dy
\end{equation*}
\begin{equation}\label{theo1}
	 -\frac{1}{6}\ln 2+3\zeta'(-1)+O\big(s^{-1}\big),
\end{equation}
and the error term in \eqref{theo1} is uniform on any compact subset of the set
\begin{equation}\label{excset1}
	\big\{ x\in\mathbb{R}:\ -\infty<x<\infty\big\}.
\end{equation}
\end{theorem}
We bring the reader's attention to the following two interesting aspects of formula (\ref{theo1}).
One is the appearance in the asymptotics  of the Tracy-Widom function,
\begin{equation*}
	F_{TW}(x) = e^{-\int_x^{\infty}(y-x)u^2(y)dy}.
\end{equation*}
The second important feature of the estimate (\ref{theo1}) is related to the Forrester-Chen-Eriksen-Tracy
conjecture (\cite{For}, \cite{CET}; see also \cite{BH}) concerning the behavior of the large gap  probabilities. The conjecture states that  
the probability  $E(s)$ of emptiness of the (properly scaled) interval $(x^{\ast}-s,x^{\ast}+s)$ around the
point $x^{\ast}$  satisfies the estimate,
\begin{equation}\label{BHF}
E(s) \sim \exp\Bigl(-Cs^{2\kappa +2}\Bigr),
\end{equation}
if the mean density $\rho(x)$ behaves as $\rho \sim (x-x^{\ast})^{\kappa}$. This conjecture
is supported by the classical results concerning the regular bulk point ($\kappa = 0$, the sine-kernel
determinant) and
regular edge point ($\kappa = 1/2$, the Airy-kernel determinant). For higher order critical edge points 
($\kappa = 2l + 1/2$, the higher Painlev\'e I - kernel determinants), estimate (\ref{BHF}) follows from the asymptotic results of \cite{CIK}. Our asymptotic equation (\ref{theo1}) supports the Forrester-Chen-Eriksen-Tracy
conjecture for the first critical case in the bulk, when $\kappa = 2$.
\bigskip

The proof of Theorem \ref{t1} is based on a Riemann-Hilbert approach. 
This approach (compare \cite{IIKS}, \cite{DIZ}) uses the given integrable form of the Fredholm operator, allowing us to connect the resolvent kernel to the solution of a Riemann-Hilbert problem. The latter  can be analysed rigorously via the Deift-Zhou nonlinear steepest descent method. 

Let us finish this introduction with a brief outline for the rest of the paper. Section \ref{sec2} starts with the more
details concerning the  definition of the kernel \eqref{PIIkernel} followed by a short review of the Riemann-Hilbert approach for the asymptotics of integrable Fredholm operators. We then apply the general framework to 
the Fredholm determinant $\det(I-K_{\textnormal{PII}})$ and formulate the associated  ``master'' Riemann-Hilbert problem (RHP).
We will also evaluate logarithmic $s$ and $x$ derivatives of the determinant $\det(I-K_s)$ in terms
of the solution of this RHP and outline a derivation
of an integrable system whose tau-function is represented by $\det(I-K_{\textnormal{PII}})$. Since the main objective
of this paper is the large $s$ asymptotics of $\det(I-K_{\textnormal{PII}})$, we shall  postpone 
a thorough  discussion of the differential equations describing this determinant until the next publication. 
In sections $\ref{sec4},\ref{sec7}-\ref{sec12}$,  following the Deift-Zhou scheme, we construct the asymptotic solution of the
master Riemann-Hilbert problem. Comparing to the more usual cases, an extra ``undressing''
step is needed to overcome the transcendentality of the kernel $K_{\textnormal{PII}}(\lambda,\mu;x)$.
Here, a crucial role is played by the aforementioned alternative Riemann-Hilbert definition of the function
$\Psi(\lambda,x)$.  The situation is similar to the one dealt with in \cite{CIK}.
The calculations of sections \ref{sec13} and \ref{sec14}  provide us with the asymptotics of  $\ln \det(I-K_{\textnormal{PII}})$
announced in \eqref{theo1} but up to the constant term. In order to determine the latter, 
we will, in section \ref{sec15}, go back to equation \eqref{PIIkernel} and look at the behavior of the kernel 
$K_{\textnormal{PII}}(\lambda,\mu;x)$ as $x\rightarrow\infty$. We will see that in the large $x$ limit,
the kernel $K_{\textnormal{PII}}(\lambda,\mu)$ is replaced by the following cubic generalization
of the sine kernel, 
\begin{equation}\label{Kcheck}
	K_{\textnormal{PII}}(\lambda,\mu) \mapsto \check{K}_{\textnormal{csin}}(\lambda,\mu)=\frac{\sin\big(\frac{4}{3}(\lambda^3-\mu^3)+x(\lambda-\mu)\big)}{\pi(\lambda-\mu)},
\end{equation}
We will then introduce a parameter $t\in[0,1]$
\begin{equation*}
	\check{K}_{\textnormal{csin}}(\lambda,\mu)\mapsto K_{\textnormal{csin}}(\lambda,\mu)=\frac{\sin\big(\frac{4}{3}t(\lambda^3-\mu^3)+x(\lambda-\mu)\big)}{\pi(\lambda-\mu)}
\end{equation*}
and compute the large $s$ behavior of $\det(I-K_{\textnormal{csin}})$ using again the Riemann-Hilbert approach.
This will be done in sections \ref{sec17}-\ref{sec24}. This analysis will indeed produce the constant term  in \eqref{theo1} since $\det(I-K_{\textnormal{csin}})\big|_{t=0}$ reduces to the sine kernel with known asymptotics see (\ref{dyson0}),
\begin{equation*}
	\ln\det(I-K_{\textnormal{csin}})\big|_{t=0} = -\frac{(sx)^2}{2}-\frac{1}{4}\ln(sx)+\frac{1}{12}\ln 2+3\zeta'(-1)+O\big(s^{-1}\big),\ \ s\rightarrow\infty
\end{equation*}
uniformly on any compact subset of \eqref{excset1}.
\begin{remark} We do not address in this paper the question of the higher corrections to (\ref{theo1}).
After the leading and constant terms are determined, the higher corrections can be in principal
obtained by iterating the final $R$-Riemann-Hilbert problem (see section \ref{sec14}). Alternatively,
one can use the related to the determinant $\det(I-K_{\textnormal{PII}})$ differential system which we have mentioned 
above, and which we intend to discuss in detail in our next publication. 
\end{remark}
The asymptotic of the determinant $\det(I -\check{K}_{\textnormal{csin}})$ is of interest
on its own. We show, that it is given by a formula ignoring the Tracy-Widom term in \eqref{theo1}.
That is, our second asymptotic result is the following.
\begin{theorem}\label{theo2} Let $\check{K}_{\textnormal{csin}}$ denote the trace class operator on $L^2\big((-s,s);d\lambda\big)$ with kernel \eqref{Kcheck}. Then as $s\rightarrow\infty$ the Fredholm determinant $\det(I-\check{K}_{\textnormal{csin}})$ behaves as 
\begin{equation}\label{theo10}
\ln \det(I-\check{K}_{\textnormal{csin}})=-\frac{2}{3}s^6-s^4x-\frac{1}{2}(sx)^2 - \frac{3}{4}\ln s 
 -\frac{1}{6}\ln 2+3\zeta'(-1)+O\big(s^{-1}\big).
\end{equation}
and the error term in \eqref{theo10} is uniform on any compact subset of the set \eqref{excset1}.
\end{theorem}

%---------------------------------------------------------------------------------------------------------------------------------------------------------------------------------------------------------
\section{Riemann-Hilbert approach - setup and review}\label{sec2}
The classical theory of ordinary differential equations in the complex plane implies that system \eqref{lax1} has precisely one irregular singular point of Poincar\'e rank $3$ at infinity. This observation leads to the existence of seven canonical solutions $\Psi_n(\lambda)$ which are fixed uniquely by their asymptotics (for more detail see \cite{FIKN})
\begin{equation*}
	\Psi_n(\lambda)\sim\Big(I+O\big(\lambda^{-1}\big)\Big)e^{-i(\frac{4}{3}\lambda^3+x\lambda)\sigma_3},\ \ \lambda\rightarrow\infty,\ \lambda\in\Omega_n
\end{equation*}
where the canonical sectors $\Omega_n$ (compare Figure \ref{fig1}) are defined by
\begin{equation*}
    \Omega_n=\Big\{\lambda\in\mathbb{C}\ |\ \textnormal{arg}\
    \lambda\in\Big(\frac{\pi}{3}(n-2),\frac{\pi}{3}n\Big),
    n=1,\ldots,7\Big\}.
\end{equation*}
\begin{figure}[tbh]
  \begin{center}
  \psfragscanon
  \psfrag{1}{\footnotesize{$\Omega_2$}}
  \psfrag{2}{\footnotesize{$0<\textnormal{arg}\ \lambda<\frac{2\pi}{3}$}}
  \psfrag{3}{\footnotesize{$\Omega_3$}}
  \psfrag{4}{\footnotesize{$\frac{\pi}{3}<\textnormal{arg}\ \lambda<\pi$}}
  \psfrag{5}{\footnotesize{$\Omega_4$}}
  \psfrag{6}{\footnotesize{$\frac{2\pi}{3}<\textnormal{arg}\ \lambda<\frac{4\pi}{3}$}}
  \psfrag{7}{\footnotesize{$\pi<\textnormal{arg}\ \lambda<\frac{5\pi}{3}$}}
  \psfrag{8}{\footnotesize{$\Omega_5$}}
  \psfrag{9}{\footnotesize{$\Omega_6$}}
  \psfrag{10}{\footnotesize{$\frac{4\pi}{3}<\textnormal{arg}\
  \lambda<2\pi$}}
  \psfrag{11}{\footnotesize{$\Omega_7$}}
  \psfrag{12}{\footnotesize{$\frac{5\pi}{3}<\textnormal{arg}\ \lambda<\frac{7\pi}{3}$}}
  \psfrag{13}{\footnotesize{$\Omega_1$}}
  \psfrag{14}{\footnotesize{$-\frac{\pi}{3}<\textnormal{arg}\ \lambda<\frac{\pi}{3}$}}
  \includegraphics[width=11cm,height=6cm]{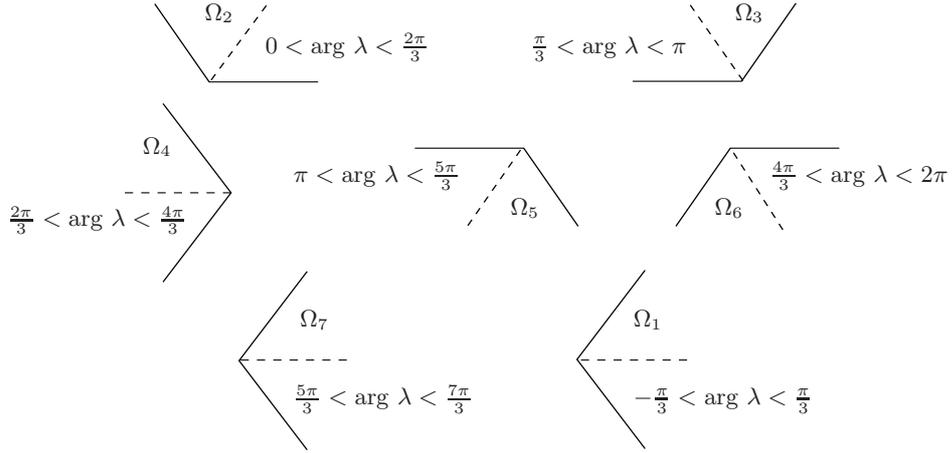}
  \end{center}
  \caption{Canonical sectors of system \eqref{lax1} with the dashed lines indicating where Re $\lambda^3=0$}
  \label{fig1}
\end{figure}

Moreover the presence of an irregular singularity gives us a non trivial Stokes phenomenon described by the Stokes matrices $S_n$:
\begin{equation*}
	S_n=\big(\Psi_n(\lambda)\big)^{-1}\Psi_{n+1}(\lambda)
\end{equation*}
In the given situation \eqref{PIIkernel} (see again \cite{FIKN}) these multipliers are 
\begin{equation}\label{stokesmat}
    S_1=\begin{pmatrix}
          1 & 0 \\
          -i & 1 \\
        \end{pmatrix},S_2=\begin{pmatrix}
                            1 & 0 \\
                            0 & 1 \\
                          \end{pmatrix},
        S_4=\begin{pmatrix}
                            1 & i \\
                            0 & 1 \\
                          \end{pmatrix},\
                          S_3=\bar{S}_1,S_5=\bar{S}_2,
                          S_6=\bar{S}_4,
\end{equation}
hence the required solution in \eqref{lax1} is the second and third canonical solution $\Psi(\lambda,x)\equiv \Psi_2(\lambda,x)=\Psi_3(\lambda,x)$ with asymptotics
\begin{equation}\label{initialasy}
	\Psi(\lambda,x)\sim\Big(I+O\big(\lambda^{-1}\big)\Big)e^{-i(\frac{4}{3}\lambda^3+x\lambda)\sigma_3},\hspace{0.5cm}\lambda\rightarrow\infty,\ 0<\textnormal{arg}\ \lambda<\pi
\end{equation}
and Stokes matrices as in \eqref{stokesmat}. Now that we have defined the integral kernel \eqref{PIIkernel} let us locate its structure within the algebra of integrable Fredholm operators and discuss the underlying Riemann-Hilbert approach. The given kernel belongs to an algebra of integrable operators first introduced in \cite{IIKS}: Let $\Gamma$ be an oriented contour in the complex plane $\mathbb{C}$ such as a Jordan curve. We are interested in operators of the form $\lambda I+K$ on $L^2(\Gamma)$, where $K$ denotes an integral operator with kernel
\begin{equation}\label{IIKStheo1}
	K(\lambda,\mu) = \frac{\sum_{i=1}^Mf_i(\lambda)h_i(\mu)}{\lambda-\mu},\hspace{0.5cm}\sum_{i=1}^Mf_i(\lambda)h_i(\lambda) = 0,\ \ M\in\mathbb{Z}_{\geq 1}
\end{equation}
with functions $f_i,h_i$ which are smooth up to the boundary of $\Gamma$. Given two operators $\lambda I+K,\check{\lambda}I+\check{K}$ of this type, the composition $(\lambda I+K)(\check{\lambda}I+\check{K})$ is again of the same form, hence we have a ring. Moreover let $K^t$ denote the real adjoint of $K$, i.e.
\begin{equation*}
	K^t(\lambda,\mu) = -\frac{\sum_{i=1}^Nh_i(\lambda)f_i(\mu)}{\lambda-\mu}.
\end{equation*}
Our results are based on the following facts of the theory of integrable operators (see e.g. \cite{DIZ}). First an algebraic Lemma, showing that the resolvent of $I-K$ is again integrable.
\begin{lemma}
Given an operator $I-K$ on $L^2(\Gamma)$ in the previous algebra with kernel \eqref{IIKStheo1}. Suppose the inverse $(I-K)^{-1}$ exists, then $I+R = (I-K)^{-1}$ lies again in the same algebra with
\begin{equation}\label{IIKStheo2}
	R(\lambda,\mu)=\frac{\sum_{i=1}^MF_i(\lambda)H_i(\mu)}{\lambda-\mu},\hspace{0.5cm}\sum_{i=1}^MF_i(\lambda)H_i(\lambda)=0
\end{equation}
and the functions $F_i,H_i$ are given by
\begin{equation}\label{IIKStheo3}
	F_i(\lambda)=\Big((I-K)^{-1}f_i\Big)(\lambda),\hspace{0.5cm} H_i(\lambda)=\Big((I-K^t)^{-1}h_i\Big)(\lambda).
\end{equation}
\end{lemma}
Secondly an analytical Lemma, which connects integrable operators to a Riemann-Hilbert problem.
\begin{lemma}
Let $K$ be of integrable type such that $(I-K)^{-1}$ exists and let $Y=Y(z)$ denote the unique solution of the following $N\times N$ Riemann-Hilbert problem (RHP)
\begin{itemize}
	\item $Y(z)$ is analytic for $z\in\mathbb{C}\backslash \Gamma$
	\item On the contour $\Gamma$, the boundary values of the function $Y(z)$ satisfy the jump relation
	\begin{equation*}
		Y_+(z)=Y_-(z)\big(I-2\pi i f(z)h^t(z)\big),\hspace{0.5cm} z\in\Gamma
	\end{equation*}
	where $f(z)=\big(f_1(z),\ldots,f_N(z)\big)^t$ and similarly $h(z)=\big(h_1(z),\ldots,h_N(z)\big)^t$
	\item At an endpoint of the contour $\Gamma$, $Y(z)$ has no more than a logarithmic singularity
	\item As $z\rightarrow\infty$
	\begin{equation*}
		Y(z)=I+O\big(z^{-1}\big)
	\end{equation*}
\end{itemize}
Then $Y(z)$ determines the resolvent kernel via
\begin{equation}\label{IIKStheo4}
	F(z)=Y(z)f(z),\hspace{0.5cm}H(z)=\big(Y^t(z)\big)^{-1}h(z)
\end{equation}
and conversely the solution of the above RHP is expressible in terms of the function $F(z)$ using the Cauchy integral
\begin{equation}\label{IIKStheo5}
	Y(z)=I-\int\limits_{\Gamma}F(w)h^t(w)\frac{dw}{w-z}.
\end{equation}
\end{lemma}
Let us use this general setup in the given situation \eqref{PIIkernel}. We have
\begin{equation}\label{spec}
	K_{\textnormal{PII}}(\lambda,\mu) = \frac{f^t(\lambda)h(\mu)}{\lambda-\mu},\hspace{0.2cm} f(\lambda)=\frac{i}{\sqrt{2\pi}}\binom{\psi_{11}(\lambda,x)}{\psi_{21}(\lambda,x)},\
    \ 
    h(\mu)=\frac{i}{\sqrt{2\pi}}\binom{\psi_{21}(\mu,x)}{-\psi_{11}(\mu,x)}
\end{equation}
hence the $Y$-RHP of Lemma $2$ reads as
\begin{itemize}
	\item $Y(\lambda)$ is analytic for $\lambda\in\mathbb{C}\backslash[-s,s]$
	\item Orienting the line segment $[-s,s]$ from left to right, the following jump holds
	\begin{equation*}
		Y_+(\lambda)=Y_-(\lambda)\begin{pmatrix}
                                   1+i\psi_{11}(\lambda)\psi_{21}(\lambda) & -i\psi_{11}^2(\lambda) \\
                                   i\psi_{21}^2(\lambda) & 1-i\psi_{11}(\lambda)\psi_{21}(\lambda) \\
                                 \end{pmatrix},\ \ \
                                 \lambda\in[-s,s]
    \end{equation*}
  \item At the endpoints $\lambda=\pm s$, $Y(\lambda)$ has logarithmic singularities, i.e.
   \begin{equation*}
   	Y(\lambda)=O\big(\ln(\lambda\mp s)\big),\ \ \lambda\rightarrow\pm s
   \end{equation*}
  \item As $\lambda\rightarrow\infty$ we have 
  \begin{equation*}
  	Y(\lambda)=I+\frac{m_1}{\lambda}+O\big(\lambda^{-2}\big).
  \end{equation*}
\end{itemize} 
The given jump matrix on the segment $[-s,s]$ can be factorized using the unimodular fundamental solution $\Psi(\lambda)$ of \eqref{lax1} corresponding to the choices \eqref{stokesmat} and \eqref{initialasy}
\begin{eqnarray*}
	G(\lambda) &=& \begin{pmatrix}
                                   1+i\psi_{11}(\lambda)\psi_{21}(\lambda) & -i\psi_{11}^2(\lambda) \\
                                   i\psi_{21}^2(\lambda) & 1-i\psi_{11}(\lambda)\psi_{21}(\lambda) \\
                                 \end{pmatrix}\\
                                 &=&\begin{pmatrix}
                                                 \psi_{11}(\lambda) & \psi_{12}(\lambda) \\
                                                 \psi_{21}(\lambda) & \psi_{22}(\lambda) \\
                                               \end{pmatrix}\begin{pmatrix}
                                                              1 & -i \\
                                                              0 & 1 \\
                                                            \end{pmatrix}\begin{pmatrix}
                                                                           \psi_{22}(\lambda) & -\psi_{12}(\lambda) \\
                                                                           -\psi_{21}(\lambda) & \psi_{11}(\lambda) \\
                                                                         \end{pmatrix}\\
                                 &=&\Psi(\lambda)\begin{pmatrix}
                                                   1 & -i \\
                                                   0 & 1 \\
                                                 \end{pmatrix}\big(\Psi(\lambda)\big)^{-1}.
\end{eqnarray*}
This motivates the first transformation of the RHP.

%----------------------------------------------------------------------------------------------------------------------------------------

\section{First transformation of the RHP}\label{sec3}
We make the following substitution in the original $Y$-RHP
\begin{equation}\label{firsttrafo}
	\tilde{X}(\lambda)=Y(\lambda)\Psi(\lambda),\ \ \lambda\in\mathbb{C}\backslash[-s,s].
\end{equation}
This leads to a RHP for the function $\tilde{X}(\lambda)$
\begin{itemize}
	\item $\tilde{X}(\lambda)$ is analytic for $\lambda\in\mathbb{C}\backslash[-s,s]$
	\item The following jump holds
	\begin{equation}\label{Xtildejump}
		\tilde{X}_+(\lambda) = \tilde{X}_-(\lambda)\begin{pmatrix}
		1 & -i\\
		0 & 1\\
		\end{pmatrix},\ \ \lambda\in[-s,s]
	\end{equation}
	\item As $\lambda\rightarrow \pm s$, we have
	\begin{equation*}
		\tilde{X}(\lambda) = O\big(\ln(\lambda\mp s)\big)
	\end{equation*}
	\item At infinity,
	\begin{equation*}
		\tilde{X}(\lambda) = \Big(I+O\big(\lambda^{-1}\big)\Big)\Psi(\lambda),\ \ \lambda\rightarrow\infty
	\end{equation*}
\end{itemize}
We will now use the Stokes phenomenon \eqref{stokesmat} of $\Psi(\lambda)$ and introduce more cuts to the Riemann-Hilbert problem in order to uniformize its behavior at infinity.

%------------------------------------------------------------------------------------------------------------------------------------------------------------------------------------------------------------------------------------

\section{Second transformation of the RHP - uniformization}\label{sec4}

Let 
\begin{equation*}
    X(\lambda)=\tilde{X}(\lambda)\left\{
                                   \begin{array}{ll}
                                     I, & \hbox{$\lambda\in\hat{\Omega}_1$,} \\
                                     S_3, & \hbox{$\lambda\in\hat{\Omega}_2$,} \\
                                     S_3S_4, & \hbox{$\lambda\in\hat{\Omega}_3$,} \\
                                     S_3S_4S_6, & \hbox{$\lambda\in\hat{\Omega}_4$,}
                                   \end{array}
                                 \right.
\end{equation*}
with
\begin{eqnarray*}
	\Gamma_1 = \Big\{\lambda\in\mathbb{C}:\ \textnormal{arg}(\lambda-s) = \frac{\pi}{6}\Big\},&& \Gamma_3 = \Big\{\lambda\in\mathbb{C}:\ \textnormal{arg}(\lambda+s) = \frac{5\pi}{6}\Big\},\\
	\Gamma_4 = \Big\{\lambda\in\mathbb{C}:\ \textnormal{arg}(\lambda+s) = -\frac{5\pi}{6}\Big\},&& \Gamma_6 = \Big\{\lambda\in\mathbb{C}:\ \textnormal{arg}(\lambda-s) = -\frac{\pi}{6}\Big\},
\end{eqnarray*}
then $X(\lambda)$ satisfies the following RHP, depicted in Figure \ref{fig2}
\begin{figure}[tbh]
  \begin{center}
  \psfragscanon
  \psfrag{1}{\footnotesize{$\Gamma_1$}}
  \psfrag{2}{\footnotesize{$\Gamma_3$}}
  \psfrag{3}{\footnotesize{$\Gamma_4$}}
  \psfrag{4}{\footnotesize{$\Gamma_6$}}
  \psfrag{5}{\footnotesize{$-s$}}
  \psfrag{6}{\footnotesize{$s$}}
  \psfrag{7}{\footnotesize{$\hat{\Omega}_1$}}
  \psfrag{8}{\footnotesize{$\hat{\Omega}_2$}}
  \psfrag{9}{\footnotesize{$\hat{\Omega}_3$}}
  \psfrag{10}{\footnotesize{$\hat{\Omega}_4$}}
  \includegraphics[width=6cm,height=4cm]{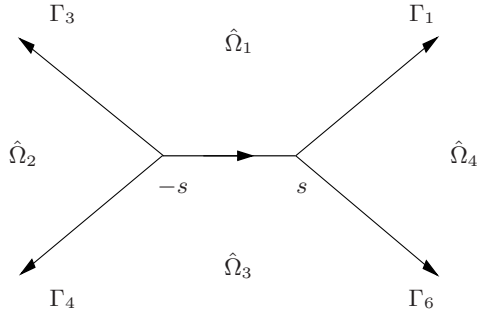}
  \end{center}
  \caption{Jump contours of the $X$-RHP}
  \label{fig2}
\end{figure}
\begin{itemize}
    \item $X(\lambda)$ is analytic for
    $\lambda\in\mathbb{C}\backslash\big( [-s,s]\cup\bigcup_k\Gamma_k\big)$
    \item Along the infinite rays $\Gamma_k$, $X(\lambda)$ has jumps
    described by the Stokes matrices
    \begin{equation*}
        X_+(\lambda)=X_-(\lambda)S_k,\ \ \lambda\in\Gamma_k,
    \end{equation*}
    whereas on the line segment $[-s,s]$ we have the following jump
    \begin{equation*}
        X_+(\lambda)=X_-(\lambda)\begin{pmatrix}
                                   0 & -i \\
                                   -i & 0 \\
                                 \end{pmatrix},\ \ \lambda\in[-s,s].
    \end{equation*}
    \item In a neighborhood of the endpoints $\lambda=\pm s$,
    \begin{equation}\label{Xlocal}
        X(\lambda)=\check{X}(\lambda)\begin{pmatrix}
                                     1 & -\frac{1}{2\pi}\ln\frac{\lambda-s}{\lambda+s} \\
                                     0 & 1 \\
                                   \end{pmatrix}\left\{
                                   \begin{array}{ll}
                                     I, & \hbox{$\lambda\in\hat{\Omega}_1$,} \\
                                     S_3, & \hbox{$\lambda\in\hat{\Omega}_2$,} \\
                                     S_3S_4, & \hbox{$\lambda\in\hat{\Omega}_3$,} \\
                                     S_3S_4S_6, & \hbox{$\lambda\in\hat{\Omega}_4$,}
                                   \end{array}
                                 \right.
    \end{equation}
    where $\check{X}(\lambda)$ is analytic at $\lambda=\pm s$ and the
    branch of the logarithm is fixed by the condition
    $-\pi<\textnormal{arg}\
    \frac{\lambda-s}{\lambda+s}<\pi$.\footnote{The local behavior \eqref{Xlocal} of $X(\lambda)$ at the endpoints
		$\pm s$ can be derived directly from the a-priori information $\tilde{X}(\lambda)=O(\ln(\lambda\mp s)),\lambda\rightarrow\pm s$ and the jump condition \eqref{Xtildejump}.}
    \item As $\lambda\rightarrow\infty$ the following
    asymptotical behavior holds
    \begin{equation}\label{Xasyinfinity}
    	X(\lambda) = \Big(I+\frac{m_1}{\lambda}+O\big(\lambda^{-2}\big)\Big)\Big(I+\frac{m_1^{HM}}{\lambda}+O\big(\lambda^{-2}\big)\Big)e^{-i(\frac{4}{3}\lambda^3+x\lambda)\sigma_3}
    \end{equation}
    with
    \begin{equation*}
    	m_1^{HM}=\frac{1}{2}\begin{pmatrix}
                          -iv & u \\
                          u & iv \\
                        \end{pmatrix},\hspace{1cm}v=(u_x)^2-xu^2-u^4,\hspace{0.5cm} v_x=-u^2.
		\end{equation*}
\end{itemize}
As we are going to see in the next sections the latter $X$-RHP can be solved asymptotically by approximating its global solution with local model functions. Before we start this analysis in detail, we first connect the solution of the $X$-RHP to the Fredholm determinant $\det(I-K_{\textnormal{PII}})$.

%--------------------------------------------------------------------------------------------------------------------------------------------------------------------------------------------

\section{Logarithmic derivatives - connection to $X$-RHP}\label{sec5}
We wish to express logarithmic derivatives of the determinant $\det(I-K_{\textnormal{PII}})$ in terms of the solution of the $X$-RHP. To this end recall the following classical identity
%using the Deift-Zhou nonlinear steepest descent method, therefore tracing back the transformations
%\begin{equation*}
%	Y(\lambda)\mapsto \tilde{X}(\lambda)\mapsto X(\lambda),
%\end{equation*}
%we can asymptotically solve the original $Y$-RHP. Its solution will determine the large $s$ asymptotics of logarithmic derivatives of the given Fredholm determinant $\det(I-K_s)$. Indeed recall the following classical identity 
\begin{equation}\label{traceform}
    \frac{d}{ds}\ln\det(I-K_{\textnormal{PII}})=-\textnormal{trace}\bigg((I-K_{\textnormal{PII}})^{-1}\frac{dK_{\textnormal{PII}}}{ds}\bigg).
\end{equation}
In our situation
\begin{equation*}
    \frac{dK_{\textnormal{PII}}}{ds}(\lambda,\mu)=K_{\textnormal{PII}}(\lambda,\mu)\big(\delta(\mu-s)+\delta(\mu+s)\big),
\end{equation*}
where, by definition
\begin{equation*}
    \int\limits_{-s}^s\delta(w\mp s)f(w)dw = f(\pm s)
\end{equation*}
and therefore
\begin{equation*}
    -\textnormal{trace}\ \Big((I-K_{\textnormal{PII}})^{-1}\frac{dK_{\textnormal{PII}}}{ds}\Big) =
    -R(s,s)-R(-s,-s)
\end{equation*}
with $R(\lambda,\mu)$ denoting the kernel (see \eqref{IIKStheo2}) of the resolvent
$R=(I-K_{\textnormal{PII}})^{-1}K_{\textnormal{PII}}$. The latter derivative can be simplified using the equations (see \eqref{spec})
\begin{equation*}
	f_1(\lambda)=-h_2(\lambda),\hspace{0.5cm} f_2(\lambda)=h_1(\lambda)
\end{equation*} 
as well as the identity $\det Y(\lambda)\equiv 1$, which is a direct consequence of the unimodularity of the jump matrix $G(\lambda)$ and Liouville's theorem. We have,
\begin{equation*}
	R(\lambda,\mu)=\frac{F_1(\lambda)H_1(\mu)+F_2(\lambda)H_2(\mu)}{\lambda-\mu}=\frac{F_1(\lambda)F_2(\mu)-F_2(\lambda)F_1(\mu)}{\lambda-\mu}.
\end{equation*}
Since $R(\lambda,\mu)$ is continuous along the diagonal $\lambda=\mu$ we obtain further
\begin{equation}\label{resolventRHP}
	R(s,s)=F_1'(s)F_2(s)-F_2'(s)F_1(s),\ \ R(-s,-s)=F_1'(-s)F_2(-s)-F_2'(-s)F_1(-s)
\end{equation}
provided $F_i$ is analytic at $\lambda=\pm s$. One way to see this is a follows. Use the connection
\begin{equation*}
    X(\lambda)= Y(\lambda)\Psi(\lambda)\left\{
                                   \begin{array}{ll}
                                     I, & \hbox{$\lambda\in\hat{\Omega}_1$,} \\
                                     S_3, & \hbox{$\lambda\in\hat{\Omega}_2$,} \\
                                     S_3S_4, & \hbox{$\lambda\in\hat{\Omega}_3$,} \\
                                     S_3S_4S_6, & \hbox{$\lambda\in\hat{\Omega}_4$,}
                                   \end{array}
                                 \right.\equiv
                                 Y(\lambda)\Psi(\lambda)S(\lambda)
\end{equation*}
and \eqref{IIKStheo4}
\begin{equation*}
	F(\lambda)=X(\lambda)\big(S(\lambda)\big)^{-1}\big(\Psi(\lambda)\big)^{-1}f(\lambda)= X(\lambda)\big(S(\lambda)\big)^{-1}\frac{i}{\sqrt{2\pi}}\binom{1}{0}
\end{equation*}
as well as \eqref{Xlocal} to derive the following local identity
\begin{equation}\label{localF}
	F(\lambda)=\check{X}(\lambda)\begin{pmatrix}
                                 1 & -\frac{1}{2\pi}\ln\frac{\lambda-s}{\lambda+s} \\
                                 0 & 1 \\
                               \end{pmatrix}S(\lambda)\big(S(\lambda)\big)^{-1}
                               \frac{i}{\sqrt{2\pi}}\binom{1}{0}= \check{X}(\lambda)\frac{i}{\sqrt{2\pi}}\binom{1}{0},
\end{equation}
valid in a vicinity of $\lambda=\pm s$. But this proves analyticity of $F(\lambda)$ at the endpoints and as we shall see later on, \eqref{localF} is all we need to connect \eqref{traceform} via \eqref{resolventRHP} to the solution of the $X$-RHP. We summarize
\begin{proposition}\label{prop1} The logarithmic $s$-derivative of the Fredholm determinant \eqref{Pdet} can be expressed as
\begin{eqnarray}\label{sidentity}
	\frac{d}{ds}\ln\det(I-K_{\textnormal{PII}})&=&-R(s,s)-R(-s,-s),\\
	&& \ R(\pm s,\pm s)=F_1'(\pm s)F_2(\pm s)-F_2'(\pm s)F_1(\pm s)\nonumber
\end{eqnarray}
and the connection to the $X$-RHP is established through
\begin{equation*}
	F(\lambda)=\check{X}(\lambda)\frac{i}{\sqrt{2\pi}}\binom{1}{0},
\end{equation*}
where $\check{X}(\lambda)$ is analytic in a neighborhood of $\lambda=\pm s$.
\end{proposition}

Besides the logarithmic $s$-derivative we also differentiate with respect to $x$
\begin{equation*}
    \frac{d}{dx}\ln\det(I-K_{\textnormal{PII}})=-\textnormal{trace}\
    \bigg((I-K_{\textnormal{PII}})^{-1}\frac{dK_{\textnormal{PII}}}{dx}\bigg).
\end{equation*}
In our situation the kernel depends explicitly on $x$, since (see e.g. \cite{FIKN})
\begin{equation*}
	\frac{d\Psi}{dx} = U(\lambda,x)\Psi,\hspace{0.5cm} U(\lambda,x) = -i\lambda\sigma_3+i\begin{pmatrix}
	0 & u\\
	-u & 0\\
	\end{pmatrix},
\end{equation*}
we have
\begin{eqnarray*}
    \frac{dK_{\textnormal{PII}}}{dx}(\lambda,\mu) &=&
    \frac{i}{2\pi}\big(\psi_{21}(\lambda,x)\psi_{11}(\mu,x)+\psi_{21}(\mu,x)\psi_{11}(\lambda,x)\big)\\
    &=&i\big(f_2(\lambda)h_2(\mu)-f_1(\lambda)h_1(\mu)\big)
\end{eqnarray*}
and with \eqref{IIKStheo3}
\begin{equation*}
    -\textnormal{trace}\ \bigg((I-K_{\textnormal{PII}})^{-1}\frac{dK_{\textnormal{PII}}}{dx}\bigg)
    =
    -i\int\limits_{-s}^s\big(F_2(\lambda)h_2(\lambda)-F_1(\lambda)h_1(\lambda)\big)\
    d\lambda.
\end{equation*}
On the other hand the Cauchy integral \eqref{IIKStheo4} implies
\begin{equation*}
	Y(\lambda)=I+\frac{m_1}{\lambda}+O\big(\lambda^{-2}\big),\ \ \lambda\rightarrow\infty;\ m_1=\int\limits_{-s}^sF(w)h^t(w)dw
\end{equation*}
so
\begin{equation*}
	\frac{d}{dx}\ln\det(I-K_{\textnormal{PII}})=i\big(m_1^{11}-m_1^{22}\big),\ \ \ m_1=\big(m_1^{ij}\big)
\end{equation*}
and the connection to the $X$-RHP is established via \eqref{Xasyinfinity}. Again we summarize
\begin{proposition}\label{prop2}
The logarithmic $x$-derivative of the given Fredholm determinant can be expressed as
\begin{equation}\label{xidentity}
	\frac{d}{dx}\ln\det(I-K_{\textnormal{PII}})=i\big(X_1^{11}-X_1^{22}\big)-v
\end{equation}
with
\begin{equation*}
	X(\lambda)\sim \Big(I+\frac{X_1}{\lambda}+O\big(\lambda^{-2}\big)\Big)e^{-i(\frac{4}{3}\lambda^3+x\lambda)\sigma_3},\ \ \lambda\rightarrow\infty; \ X_1=\big(X_1^{ij}\big).
\end{equation*}
\end{proposition}

As already indicated in the introduction, Proposition \ref{prop1} and \ref{prop2} are enough to determine the large $s$-asymptotics of $\ln\det(I-K_{\textnormal{PII}})$ up to the constant term.

%-------------------------------------------------------------------------------------------------------------------------------------------------------------------------

\section{Differential equations associated with the determinant $\det(I-K_{\textnormal{PII}})$}\label{sec6}
Our consideration relies only on the underlying Riemann-Hilbert problem. Nevertheless, before we go further ahead in the asymptotical analysis, we would like to take a short look into the differential equations associated with the $X$-RHP.
\bigskip

To this end we notice that the $X$-RHP has unimodular constant jump matrices, thus the well-defined logarithmic derivatives $X_{\lambda}X^{-1}(\lambda),X_sX^{-1}(\lambda)$ and $X_xX^{-1}(\lambda)$ are rational functions. Indeed using \eqref{Xlocal} as well as \eqref{Xasyinfinity} we have
\begin{equation}\label{diffeq1}
    \frac{\partial X}{\partial\lambda} =
    \bigg[-4i\lambda^2\sigma_3+4i\lambda\begin{pmatrix}
                                          0 & b \\
                                          -c & 0 \\
                                        \end{pmatrix}+\begin{pmatrix}
                                                        d & e \\
                                                        f & -d \\
                                                      \end{pmatrix}+\frac{A}{\lambda-s}-\frac{B}{\lambda+s}\bigg]X\equiv
                                                      \mathcal{A}(\lambda,s,x)X
\end{equation}
where
\begin{equation*}
    A=-\frac{1}{2\pi}\check{X}(s)\begin{pmatrix}
                                                       0 & 1 \\
                                                       0 & 0 \\
                                                     \end{pmatrix}\big(\check{X}(s)\big)^{-1};\
                                                     \ \
                                                     B=-\frac{1}{2\pi}\check{X}(-s)\begin{pmatrix}
                                                                             0 & 1 \\
                                                                             0 & 0 \\
                                                                           \end{pmatrix}
                                                                           \big(\check{X}(-s)\big)^{-1}
\end{equation*}
and with parameters $b,c,d,e,f$ which can be expressed in terms of the entries of $m_1$ and $m_1^{HM}$ (see \eqref{Xasyinfinity}). Moreover
\begin{equation*}
	\frac{\partial X}{\partial
    s}=\bigg[-\frac{A}{\lambda-s}-\frac{B}{\lambda+s}\bigg]X\equiv
    \Theta(\lambda,s,x)X
\end{equation*}
and also
\begin{equation*}
	\frac{\partial X}{\partial x}=\bigg[-i\lambda\sigma_3+i\begin{pmatrix}
                                                             0 & b \\
                                                             -c & 0 \\
                                                           \end{pmatrix}\bigg]X\equiv\Pi(\lambda,s,x)X.
\end{equation*}
Hence we arrive at the Lax-system for the function $X$,
\begin{equation*}
    \left\{
      \begin{array}{c}
        \frac{\partial X}{\partial \lambda}=\mathcal{A}(\lambda,s,x)X \\
        \\
        \frac{\partial X}{\partial s}=\Theta(\lambda,s,x)X,  \\
        \\
        \frac{\partial X}{\partial x}=\Pi(\lambda,s,x)X.
      \end{array}
    \right.
\end{equation*}
Considering the compatibility conditions of the system,
 \begin{equation}\label{compa}
    \mathcal{A}_s-\Theta_{\lambda}=[\Theta,\mathcal{A}],\hspace{1cm}
\mathcal{A}_x-\Pi_{\lambda}=[\Pi,\mathcal{A}],\hspace{1cm}
\Theta_x-\Pi_s=[\Pi,\Theta]
\end{equation}
we are lead to a system of eighteen nonlinear ordinary differential equations for the unknown quantities $b,c,d,e,f$ and the entries of $A$ and $B$. Since it is possible to express the previous derivatives of $\ln\det(I-K_{\textnormal{PII}})$ solely in terms of the unknowns $b,c,d,e,f,A$ and $B$, one could then try to derive a differential equation for the Fredholm determinant \eqref{Pdet} using \eqref{compa}. We shall devote to these issues our next publication.

%---------------------------------------------------------------------------------------------------------------------------------------------------------------------------------------------------------
\section{Second transformation of the RHP - rescaling and $g$-function}\label{sec7}
%\section{Riemann-Hilbert approach - Asymptotic analysis}
%The main goal of this section is to construct an asymptotic solution of the $X$-RHP. In the framework of the Deift-Zhou nonlinear steepest descent method this %can be done by approximating the global solution with local model functions and passing from the given RHP to an error RHP which has jumps close to the identity. %Doing so, we will be able to solve the corresponding integral equations for the error RHP asymptotically via iteration and thus deduce an asymptotic solution of %the $X$-RHP. Several steps are needed.
Along the line segment $[-s,s]$ (see Figure \ref{fig2}) we have that
\begin{equation}\label{permu}
	X_+(\lambda)=X_-(\lambda)\bigl(\begin{smallmatrix}
                                   0 & -i \\
                                   -i & 0 \\
                                 \end{smallmatrix}\bigr),
\end{equation}
i.e. we face a permutation jump matrix. This behavior (cf. \cite{DIZ}) motivates the introduction of the following g-function,
\begin{equation}\label{gfunction}
	g(z)=\frac{4i}{3}\sqrt{z^2-1}\bigg(z^2+\frac{1}{2}+\frac{3x}{4s^2}\bigg),\ \ \sqrt{z^2-1} \sim z,\ z\rightarrow\infty.
\end{equation}
This function is analytic outside the segment $[-1,1]$ and as $z\rightarrow\infty$
\begin{equation*}
	g(z)=\vartheta(z)+O\big(z^{-1}\big),\ \ \vartheta(z)=i\bigg(\frac{4}{3}z^3+\frac{xz}{s^2}\bigg).
\end{equation*}
Also,
\begin{equation}\label{gjump}
	g_+(z)+g_-(z)=0,\hspace{0.5cm} z\in[-1,1].
\end{equation}
We put
\begin{equation}\label{fromXtoT}
	T(z) = X(zs)e^{s^3g(z)\sigma_3},\ \ z\in\mathbb{C}\backslash\Big([-1,1]\cup\bigcup_k\Gamma_k\Big)
\end{equation}
and, taking into account \eqref{gjump}, are lead to the following RHP 
\begin{itemize}
	\item $T(z)$ is analytic for $z\in\mathbb{C}\backslash \big([-1,1]\cup\bigcup_k\Gamma_k\big)$
	\item The jump properties of $T(z)$ are given by the equations
	\begin{equation*}
		T_+(z) = T_-(z)\begin{pmatrix}
		0 & -i\\
		-i & 0\\
		\end{pmatrix},\ \ z\in[-1,1]
	\end{equation*}
	\begin{equation*}	
	 T_+(z)=T_-(z)e^{-s^3g(z)\sigma_3}S_ke^{s^3g(z)\sigma_3},\ z\in\Gamma_k.
	\end{equation*}
	\item In a neighborhood of the endpoints $z=\pm 1$
	\begin{equation}\label{Tsingular}
		T(z)e^{-s^3g(z)\sigma_3} = \check{X}(zs)\begin{pmatrix}
		1 & -\frac{1}{2\pi}\ln\frac{z-1}{z+1}\\
		0 & 1\\
		\end{pmatrix}\left\{
                                   \begin{array}{ll}
                                     I, & \hbox{$\lambda\in\hat{\Omega}_1$,} \\
                                     S_3, & \hbox{$\lambda\in\hat{\Omega}_2$,} \\
                                     S_3S_4, & \hbox{$\lambda\in\hat{\Omega}_3$,} \\
                                     S_3S_4S_6, & \hbox{$\lambda\in\hat{\Omega}_4$,}
                                   \end{array}
                                 \right.
   \end{equation}
   \item As $z\rightarrow\infty$, we have $T(z) = I+O(z^{-1})$
\end{itemize}
%Using the g-function \eqref{gfunction} as well as the indicated scaling, our original $X$-RHP of Figure \ref{fig2} gets transformed to a RHP posed on %$\cup_k\Gamma_k\cup[-1,1]$ with identical jumps
%\begin{equation*}
%	X_+(z)=X_-(z)S_k,\ \ z\in\Gamma_k;\hspace{0.5cm} X_+(z)=X_-(z)\begin{pmatrix}
%                                   0 & -i \\
%                                   -i & 0 \\
%                                 \end{pmatrix},\ \ z\in[-1,1],
%\end{equation*}
%rescaled singular behavior at the endpoints $\pm 1$ in \eqref{Xlocal} and the following asymptotics
%\begin{equation*}
%	X(z)=\Big(I+O\big(z^{-1}\big)\Big)e^{-s^3g(z)\sigma_3},\ \ z\rightarrow\infty.
%\end{equation*}
Let us analyse the behavior of the jumps along the infinite branches $\Gamma_k$ as $s\rightarrow\infty$. To this end consider the sign-diagram of the function $\textnormal{Re}\ g(z)$, depicted in Figure \ref{fig3}, where $x$ is chosen from a compact subset of the real line and $s>0$ is sufficiently large.
\begin{figure}[tbh]
  \begin{center}
  \psfragscanon
  \psfrag{1}{\footnotesize{$-1$}}
  \psfrag{2}{\footnotesize{$1$}}
  \psfrag{3}{\footnotesize{$\textnormal{Re}\ g<0$}}
  \psfrag{4}{\footnotesize{$\textnormal{Re}\ g<0$}}
  \psfrag{5}{\footnotesize{$\textnormal{Re}\ g>0$}}
  \psfrag{6}{\footnotesize{$\textnormal{Re}\ g>0$}}
  \psfrag{7}{\footnotesize{$\textnormal{Re}\ g>0$}}
  \psfrag{8}{\footnotesize{$\textnormal{Re}\ g<0$}}
  \includegraphics[width=7cm,height=6cm]{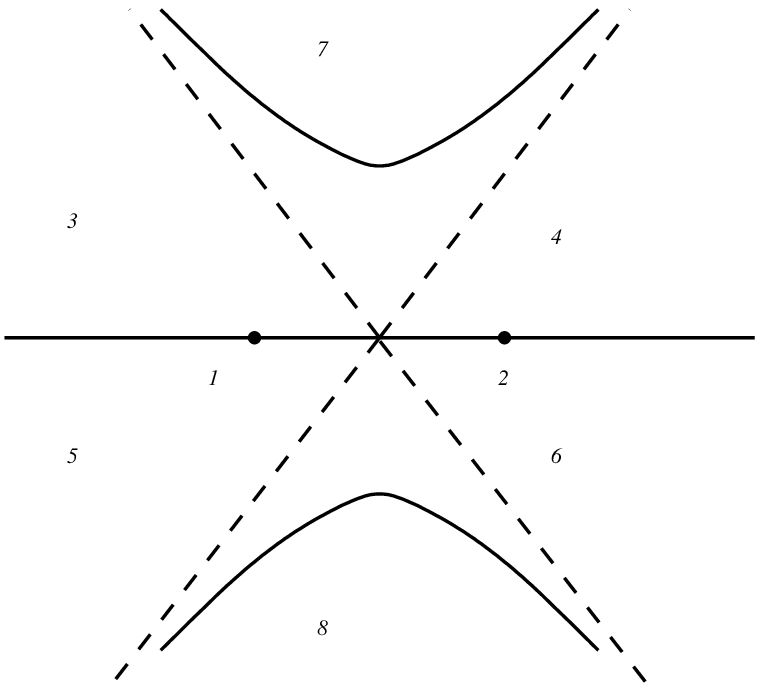}
  \end{center}
  \caption{Sign-diagram for the function $\textnormal{Re}\ g(z)$. Along the solid lines $\textnormal{Re}\ g(z)=0$ and the dashed lines resemble $\textnormal{arg}\ z =\pm\frac{\pi}{3},\pm\frac{2\pi}{3}$}
  \label{fig3}
\end{figure}

%It is important to notice that these anti-Stokes lines will depend on the values of $s$ and $x$ (see \eqref{gfunction}), however if we choose $x$ from a compact %subset of the real line and $s>0$ sufficiently large enough we can, independently of the sign of $x$, always assume that 
Since $\textnormal{Re}\ g(z)$ is negative resp. positive along the rays $\Gamma_1,\Gamma_3$ resp. $\Gamma_4,\Gamma_6$,
\begin{equation}\label{infexpdecay}
	e^{-s^3g(z)\sigma_3}S_ke^{s^3g(z)\sigma_3} \longrightarrow I,\ \ s\rightarrow\infty
\end{equation}
uniformly on any compact subset of the set \eqref{excset1} and the stated convergence is in fact exponentially fast. From \eqref{infexpdecay} we expect, and this will be justified rigorously, that as $s\rightarrow\infty$, $T(z)$ converges to a solution of the model RHP, in which we only have to deal with the constant jump matrix on the line segment $[-1,1]$. Let us now consider this model RHP.
 %Hence this contour deformation leads us to the final jump picture of Figure \ref{fig4}, in which all jump matrices of $X(z)e^{s^3g(z)\sigma_3}$ due to %triangularity and the sign of $\textnormal{Re}\ g(z)$, approach identity exponentially fast as $s\rightarrow\infty$ for fixed $x\in\mathbb{R}$, provided we stay %away from the line segment $[-1,1]$ and its endpoints $\pm 1$. Their vicinities require construction of parametrices and we shall start with the line segment %oriented from left to right as in Figure \ref{fig4}.
%\begin{figure}[tbh]
%  \begin{center}
%  \psfragscanon
%  \psfrag{1}{$\bigr(\begin{smallmatrix}
%  0 & -i\\
%  -i & 0\\
%  \end{smallmatrix}\bigl)$}
%  \psfrag{2}{\footnotesize{$S_1$}}
%  \psfrag{3}{\footnotesize{$S_3$}}
%  \psfrag{4}{\footnotesize{$S_4$}}
%  \psfrag{5}{\footnotesize{$S_6$}}
%  \psfrag{6}{\footnotesize{$+1$}}
%  \psfrag{7}{\footnotesize{$-1$}}
%  \includegraphics[width=7cm,height=4cm]{Xfinaljumpgraph.eps}
%  \end{center}
%  \caption{Final jump graph of the $X$-RHP}
%  \label{fig4}
%\end{figure}
%\bigskip
%------------------------------------------------------------------------------------------------------------------------------------------------------------------------------

\section{The model RHP}\label{sec8}

Find the piecewise analytic $2\times 2$ matrix valued function $M(z)$ such that
\begin{itemize}
	\item $M(z)$ is analytic for $z\in\mathbb{C}\backslash[-1,1]$
	\item Along $[-1,1]$ the following jump condition holds
	\begin{equation*}
		M_+(z)=M_-(z)\bigl(\begin{smallmatrix}
			0 & -i \\
			-i & 0 \\
		\end{smallmatrix}\bigr),\ \ \ z\in[-1,1]
	\end{equation*}
	\item $M(z)$ has at most logarithmic singularities at the endpoints $z=\pm 1$
	\item $M(z)= I+O\big(z^{-1}\big),\ \ z\rightarrow\infty$
\end{itemize}
A solution to this problem can be obtained explicitly 
%by considering $Y^P(z)=X^P(z)e^{s^3g(z)\sigma_3}$. Since by the choice of branches in \eqref{gfunction}
%\begin{equation*}
%	g_+(z)+g_-(z)= 0 ,\ \ \ z\in[-1,1]
%\end{equation*}
%the matrix function $Y^P(z)$ will have the same jump on $[-1,1]$ as $X^P(z)$, however it is normalized at infinity: $Y^P(z)\sim %I+O\big(z^{-1}\big),z\rightarrow\infty$. Hence we can solve the $Y^P$-RHP directly 
via diagonalization (cf. \cite{DIZ})
\begin{equation}\label{modelRHP}
	M(z)= \begin{pmatrix}
		1 & 1 \\
		1 & -1 \\
		\end{pmatrix}\beta(z)^{\sigma_3}\frac{1}{2}\begin{pmatrix}
		1 & 1 \\
		1 & -1 \\
		\end{pmatrix}=\frac{1}{2}\begin{pmatrix}
		\beta+\beta^{-1} & \beta-\beta^{-1} \\
		\beta-\beta^{-1} & \beta+\beta^{-1} \\
		\end{pmatrix},%\ \beta(z)=\bigg(\frac{z+1}{z-1}\bigg)^{1/4}
\end{equation}
with
\begin{equation*}
	\beta(z)=\bigg(\frac{z+1}{z-1}\bigg)^{1/4}
\end{equation*}
and $\big(\frac{z+1}{z-1}\big)^{1/4}$ is defined on $\mathbb{C}\backslash[-1,1]$ with its branch fixed by the condition $\big(\frac{z+1}{z-1}\big)^{1/4}\rightarrow 1$ as $z\rightarrow\infty$.
% This in turn implies
%\begin{equation}\label{XPmodel}
%	X^P(z)=\frac{1}{2}\begin{pmatrix}
%		\beta+\beta^{-1} & \beta-\beta^{-1} \\
%		\beta-\beta^{-1} & \beta+\beta^{-1} \\
%		\end{pmatrix}e^{-s^3g(z)\sigma_3}.
%\end{equation}
%-----------------------------------------------------------------------------------------------------------------------------------------------------------------------------------------------
\section{construction of a parametrix at the edge point $z=+1$}\label{sec9}
In this section we  construct a parametrix for the right endpoint $z=+1$. We start from the local expansion
\begin{equation*}
	g(z)=\frac{4\sqrt{2}}{3}i\bigg(\frac{3}{2}+\frac{3x}{4s^2}\bigg)\sqrt{z-1}\Big(1+O\big((z-1)^{1/2}\big)\Big),\ z\rightarrow 1,\ -\pi<\textnormal{arg}(z-1)\leq\pi
\end{equation*}
and the singular endpoint behavior \eqref{Xlocal}
\begin{equation*}
	X(z)=O\big(\ln(z-1)\big),\ \ z\rightarrow 1.
\end{equation*}
Both observations suggest (cf. \cite{DIZ}) to use the Bessel functions $H_0^{(1)}(\zeta)$ and $H_0^{(2)}(\zeta)$ for our construction. This idea can indeed be justified rigorously as follows. First we recall that the Hankel functions of first and second kind are unique independent solutions to Bessel's equation
\begin{equation*}
	zw''+w'+w=0
\end{equation*}
satisfying the following asymptotic conditions as $\zeta\rightarrow\infty$ and $-\pi<\textnormal{arg}\ \zeta<\pi$ (see \cite{BE})
\begin{eqnarray*}
	H_0^{(1)}(\zeta)&\sim& \sqrt{\frac{2}{\pi\zeta}}e^{i(\zeta-\frac{\pi}{4})}\bigg(1-\frac{i}{8\zeta}-\frac{9}{128\zeta^2}+\frac{75i}{1024\zeta^3}+O\big(\zeta^{-4}\big)\bigg)\\
	H_0^{(2)}(\zeta)&\sim& \sqrt{\frac{2}{\pi\zeta}}e^{-i(\zeta-\frac{\pi}{4})}\bigg(1+\frac{i}{8\zeta}-\frac{9}{128\zeta^2}-\frac{75i}{1024\zeta^3}+O\big(\zeta^{-4}\big)\bigg).
\end{eqnarray*}
Secondly $H_0^{(1)}(\zeta),H_0^{(2)}(\zeta)$ satisfy monodromy relations, valid on the entire universal covering of the punctured plane
\begin{equation}\label{besselmonodromy}
	H_0^{(1)}\big(\zeta e^{\pi i}\big)=-H_0^{(2)}(\zeta),\ H_0^{(2)}\big(\zeta e^{\pi i}\big)=H_0^{(1)}(\zeta)+2H_0^{(2)}(\zeta),\ H_0^{(2)}\big(\zeta e^{-\pi i}\big)=-H_0^{(1)}(\zeta)
\end{equation}
and finally the following expansions at the origin are valid (compare to \eqref{Xlocal})
\begin{equation}\label{Hankelorigin}
	H_0^{(1)}(\zeta) = a_0+a_1\ln\zeta+a_2\zeta^2+a_3\zeta^2\ln\zeta +O\big(\zeta^4\ln\zeta),\ \zeta\rightarrow 0
\end{equation}
with coefficients $a_i$ given as
\begin{equation*}
	a_0=1+\frac{2i\gamma_E}{\pi}-\frac{2i}{\pi}\ln 2,\ \ a_1=\frac{2i}{\pi},\ \ a_2=\frac{i}{2\pi}(1-\gamma_E)-\frac{1}{4}+\frac{i}{2\pi}\ln 2,\ \ a_3=-\frac{i}{2\pi}
\end{equation*}
where $\gamma_E$ is Euler's constant and the expansion for $H_0^{(2)}(\zeta)$ is up to the replacement $a_i\mapsto \bar{a}_i$ identical to \eqref{Hankelorigin}. Remembering all these properties we now introduce the following matrix-valued function on the punctured plane $\zeta\in\mathbb{C}\backslash\{0\}$
\begin{equation}\label{PBE}
	P_{BE}(\zeta)=e^{i\frac{\pi}{4}\sigma_3}\begin{pmatrix}
	H_0^{(2)}(\sqrt{\zeta}) & H_0^{(1)}(\sqrt{\zeta}) \\
	\sqrt{\zeta}\Big(H_0^{(2)}\Big)'(\sqrt{\zeta}) & \sqrt{\zeta}\Big(H_0^{(1)}\Big)'(\sqrt{\zeta})
	\end{pmatrix}e^{-i\frac{\pi}{4}\sigma_3},\ \ \ -\pi<\textnormal{arg}\ \zeta\leq \pi.
\end{equation}
Using the behavior of $H_0^{(1)}(\zeta)$ and $H_0^{(2)}(\zeta)$ at infinity we deduce
\begin{eqnarray*}
	P_{BE}(\zeta)&=&\sqrt{\frac{2}{\pi}}\zeta^{-\sigma_3/4}e^{i\frac{\pi}{4}}\begin{pmatrix}
	1 & 1 \\
	-1 & 1\\
	\end{pmatrix}\bigg[I+\frac{i}{8\sqrt{\zeta}}\begin{pmatrix}
	-1 & -2 \\
	2 & 1 \\
	\end{pmatrix}
	+\frac{3}{128\zeta}\begin{pmatrix}
	1 & -4\\
	-4 & 1 \\
	\end{pmatrix}\\
	&&+\frac{15i}{1024\zeta^{3/2}}\begin{pmatrix}
	1 & 6 \\
	-6 & -1 \\
	\end{pmatrix} + O\big(\zeta^{-2}\big)\bigg]e^{-i\sqrt{\zeta}\sigma_3},
\end{eqnarray*}
as $\zeta\rightarrow\infty$ and $-\pi<\textnormal{arg}\ \zeta\leq \pi$. Let us now assemble the following model function
\begin{equation}\label{PBERH}
	P_{BE}^{RH}(\zeta)= \left\{
                                 \begin{array}{ll}
                                   P_{BE}(\zeta)\begin{pmatrix}
                          1  & 0 \\
                          -i & 1 \\
                        \end{pmatrix}, & \hbox{arg $\zeta\in(\frac{\pi}{6},\pi)$,} \bigskip\\
                                   P_{BE}(\zeta)\begin{pmatrix}
                          1 & i \\
                          0 & 1 \\
                        \end{pmatrix}, & \hbox{arg $\zeta\in(-\pi,-\frac{\pi}{6})$,} \bigskip \\
                                   P_{BE}(\zeta),& \hbox{arg $\zeta\in(-\frac{\pi}{6},\frac{\pi}{6})$.}
                                 \end{array}
                               \right.
\end{equation}
which solves the RHP depicted in Figure \ref{fig5}.
\begin{figure}[tbh]
  \begin{center}
  \psfragscanon
  \psfrag{1}{$\bigr(\begin{smallmatrix}
  0 & -i\\
  -i & 0\\
  \end{smallmatrix}\bigl)$}
  \psfrag{2}{$\bigr(\begin{smallmatrix}
  1 & 0\\
  -i & 1\\
  \end{smallmatrix}\bigl)=S_1$}
  \psfrag{3}{$\bigr(\begin{smallmatrix}
  1 & -i\\
  0 & 1\\
  \end{smallmatrix}\bigl)=S_6$}
  \includegraphics[width=4cm,height=3cm]{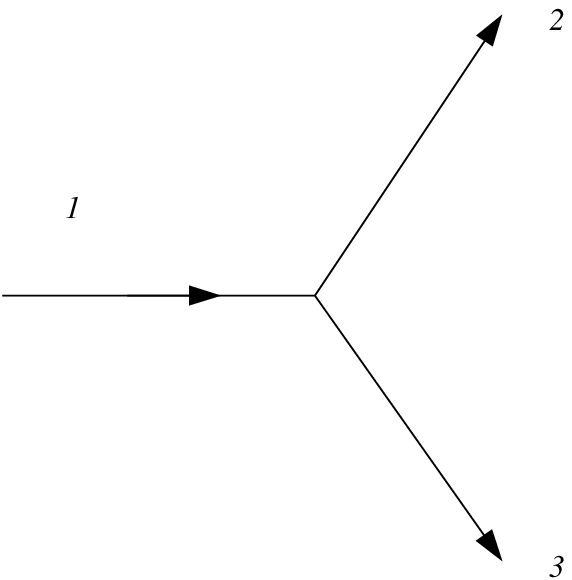}
  \end{center}
  \caption{The model RHP near $z=+1$ which can be solved explicitly using Hankel functions}
  \label{fig5}
\end{figure}

More precisely, the function $P_{BE}^{RH}(\zeta)$ possesses the following analytic properties.
\begin{itemize}
	\item $P_{BE}^{RH}(\zeta)$ is analytic for $\zeta\in\mathbb{C}\backslash\{\textnormal{arg}\ \zeta=-\pi,-\frac{\pi}{6},\frac{\pi}{6}\}$
	\item The following jumps hold
	\begin{eqnarray*}
		\big(P_{BE}^{RH}(\zeta)\big)_+&=&\big(P_{BE}^{RH}(\zeta)\big)_-\begin{pmatrix}
		1 & 0 \\
		-i & 1 \\
		\end{pmatrix},\hspace{1cm} \textnormal{arg}\ \zeta=\frac{\pi}{6}\\
		\big(P_{BE}^{RH}(\zeta)\big)_+&=&\big(P_{BE}^{RH}(\zeta)\big)_-\begin{pmatrix}
		1 & -i \\
		0 & 1 \\
		\end{pmatrix},\hspace{1cm} \textnormal{arg}\ \zeta=-\frac{\pi}{6}\\
	\end{eqnarray*}
	And for the jump on the line $\textnormal{arg}\ \zeta=\pi$ we notice that the monodromy relations imply
	\begin{eqnarray*}
		H_0^{(2)}\big(\sqrt{\zeta}_+\big)&=&H_0^{(2)}\big(\sqrt{\zeta}_-e^{\pi i}\big)=H_0^{(1)}\big(\sqrt{\zeta}_-\big)+2H_0^{(2)}\big(\sqrt{\zeta}_-\big)\\
		\big(H_0^{(2)}\big)'\big(\sqrt{\zeta}_+\big)&=&e^{-i\pi}\big(H_0^{(1)}\big)'\big(\sqrt{\zeta}_-\big)+2e^{-i\pi}\big(H_0^{(2)}\big)'\big(\sqrt{\zeta}_-\big)
	\end{eqnarray*}
	and
	\begin{eqnarray*}
		H_0^{(1)}\big(\sqrt{\zeta}_+\big)&=&H_0^{(1)}\big(\sqrt{\zeta}_-e^{\pi i}\big)=-H_0^{(2)}\big(\sqrt{\zeta}_-\big)\\
		\big(H_0^{(1)}\big)'\big(\sqrt{\zeta}_+\big)&=&\big(H_0^{(2)}\big)'\big(\sqrt{\zeta}_-\big).
	\end{eqnarray*}
	Therefore
	\begin{equation*}
		\big(P_{BE}(\zeta)\big)_+=\big(P_{BE}(\zeta)\big)_-e^{i\frac{\pi}{4}\sigma_3}\begin{pmatrix}
		2 & -1 \\
		1 & 0 \\
		\end{pmatrix}e^{-i\frac{\pi}{4}\sigma_3} = \big(P_{BE}(\zeta)\big)_-\begin{pmatrix}
		2 & -i \\
		-i & 0 \\
		\end{pmatrix}
	\end{equation*}
	and hence
	\begin{equation*}
		\big(P_{BE}^{RH}(\zeta)\big)_+=\big(P_{BE}^{RH}(\zeta)\big)_-\begin{pmatrix}
			0 & -i \\
			-i & 0 \\
		\end{pmatrix},\ \ \ \textnormal{arg}\ \zeta=\pi,
	\end{equation*}
	\item In order to determine the behavior of $P_{BE}^{RH}(\zeta)$ at infinity we make the following observations. First let $\textnormal{arg}\ \zeta\in(\frac{\pi}{6},\pi)$ and consider
	\begin{equation*}
		e^{-i\sqrt{\zeta}\sigma_3}\begin{pmatrix}
			1 & 0 \\
			-i & 1 \\
			\end{pmatrix}e^{i\sqrt{\zeta}\sigma_3} = \begin{pmatrix}
				1 & 0\\
				-ie^{2i\sqrt{\zeta}} & 1
				\end{pmatrix}.
	\end{equation*}
	Observe that $\textnormal{Re}\big(i\sqrt{\zeta}\big)<0$, hence the given product approaches the identity exponentially fast as $\zeta\rightarrow\infty$. Secondly for $\textnormal{arg}\ \zeta\in(-\pi,-\frac{\pi}{6})$
	\begin{equation*}
		e^{-i\sqrt{\zeta}\sigma_3}\begin{pmatrix}
		1 & i \\
		0 & 1 \\
	\end{pmatrix}e^{i\sqrt{\zeta}\sigma_3}=\begin{pmatrix}
	1 & ie^{-2i\sqrt{\zeta}}\\
	0 & 1 \\
	\end{pmatrix}
	\end{equation*}
	and in this situation $\textnormal{Re}\big(-i\sqrt{\zeta}\big)<0$, so again the product approaches the identity exponentially fast as $\zeta\rightarrow\infty$. Both cases together with the previously stated asymptotics for $P_{BE}(\zeta)$ imply therefore
	\begin{eqnarray}\label{PBERHasyinfinity}
	P_{BE}^{RH}(\zeta)&=&\sqrt{\frac{2}{\pi}}\zeta^{-\sigma_3/4}e^{i\frac{\pi}{4}}\begin{pmatrix}
	1 & 1 \\
	-1 & 1\\
	\end{pmatrix}\bigg[I+\frac{i}{8\sqrt{\zeta}}\begin{pmatrix}
	-1 & -2 \\
	2 & 1 \\
	\end{pmatrix}
	+\frac{3}{128\zeta}\begin{pmatrix}
	1 & -4\\
	-4 & 1 \\
	\end{pmatrix}\nonumber\\
	&&+\frac{15i}{1024\zeta^{3/2}}\begin{pmatrix}
	1 & 6 \\
	-6 & -1 \\
	\end{pmatrix} + O\big(\zeta^{-2}\big)\bigg]e^{-i\sqrt{\zeta}\sigma_3},
\end{eqnarray}
as $\zeta\rightarrow\infty$ in a whole neighborhood of infinity.
\end{itemize}
The model function $P_{BE}^{RH}(\zeta)$ will now be used to construct the parametrix to the solution of the original $X$-RHP in a neighborhood of $z=+1$. We proceed in two steps. First define
\begin{equation}\label{BEchangeright}
	\zeta(z)=-s^6g^2(z),\ \ |z-1|<r,\ -\pi<\textnormal{arg}\ \zeta\leq \pi
\end{equation}
or respectively
\begin{equation*}
	\sqrt{\zeta(z)} = -is^3g(z) = \frac{4s^3}{3}\sqrt{z^2-1}\bigg(z^2+\frac{1}{2}+\frac{3x}{4s^2}\bigg).
\end{equation*}
This change of variables is indeed locally conformal, since
\begin{equation*}
	\zeta(z)=\frac{32s^6}{9}\bigg(\frac{3}{2}+\frac{3x}{4s^2}\bigg)^2(z-1)\big(1+O(z-1)\big),\ \ |z-1|<r
\end{equation*}
and it enables us to define the right parametrix $U(z)$ near $z=+1$ by the formula:
\begin{equation}\label{Xrpara}
		U(z)=B_r(z)\frac{\sigma_3}{2}\sqrt{\frac{\pi}{2}}e^{-i\frac{\pi}{4}}P_{BE}^{RH}\big(\zeta(z)\big)e^{s^3g(z)\sigma_3},\ \ \ |z-1|<r
\end{equation}
with $\zeta(z)$ as in \eqref{BEchangeright} and the matrix multiplier
\begin{equation}\label{Brmultiplier}
	B_r(z) = \begin{pmatrix}
	1 & 1 \\
	1 & -1 \\
	\end{pmatrix}\bigg(\zeta(z)\frac{z+1}{z-1}\bigg)^{\sigma_3/4},\ \ B_r(1)=\begin{pmatrix}
	1 & 1 \\
	1 & -1 \\
	\end{pmatrix}\Bigg(\frac{8s^3}{3}\bigg(\frac{3}{2}+\frac{3x}{4s^2}\bigg)\Bigg)^{\sigma_3/2}.
\end{equation}
By construction, in particular since $B_r(z)$ is analytic in a neighborhood of $z=+1$, the parametrix $U(z)$ has jumps along the curves depicted in Figure \ref{fig6}, and we can always locally match the latter curves with the jump curves of the original RHP. Also these jumps are described by the same Stokes matrices as in the original $T$-RHP. Furthermore, and we will elaborate this in full detail very soon, the singular behavior of $U(z)$ at the endpoint $z=+1$ matches the singular behavior of $T(z)$:
\begin{equation}\label{Using}
	U(z)=O\big(\ln(z-1)\big),\ \ |z-1|<r.
\end{equation}
Hence the ratio of $T(z)$ with $U(z)$ is locally analytic, i.e.
\begin{equation}
	T(z)=N_r(z)U(z),\ \ |z-1|<r<\frac{1}{2}.
\end{equation}
\begin{figure}[tbh]
  \begin{center}
  \psfragscanon
  \psfrag{1}{$\bigr(\begin{smallmatrix}
  0 & -i\\
  -i & 0\\
  \end{smallmatrix}\bigl)$}
  \psfrag{2}{\footnotesize{$e^{-s^3g(z)\sigma_3}S_1e^{s^3g(z)\sigma_3}$}}
  \psfrag{3}{\footnotesize{$e^{-s^3g(z)\sigma_3}S_6e^{s^3g(z)\sigma_3}$}}
  \includegraphics[width=4cm,height=3cm]{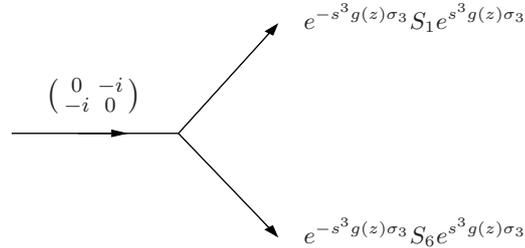}
  \end{center}
  \caption{Transformation of parametrix jumps to original jumps}
  \label{fig6}
\end{figure}

Let us explain the role of the left multiplier $B_r(z)$ in the definition \eqref{Xrpara}. Observe that
\begin{equation*}
	B_r(z)\zeta(z)^{-\sigma_3/4}\frac{1}{2}\begin{pmatrix}
	1 & 1 \\
	1 & -1 \\
	\end{pmatrix} = M(z).
\end{equation*}
This relation together with the asymptotic equation \eqref{PBERHasyinfinity} implies that,
\begin{eqnarray}\label{Xrmatchup}
	U(z) &=& \begin{pmatrix}
	1 & 1 \\
	1 & -1 \\
	\end{pmatrix}\beta(z)^{\sigma_3}\frac{1}{2}\begin{pmatrix}
	1 & 1 \\
	1 & -1 \\
	\end{pmatrix}\bigg[I+\frac{i}{8\sqrt{\zeta}}\begin{pmatrix}
	-1 & -2\\
	2 & 1 \\
	\end{pmatrix}+\frac{3}{128\zeta}\begin{pmatrix}
	1 & -4\\
	-4 & 1\\
	\end{pmatrix}\nonumber\\
	&&+\frac{15i}{1024\zeta^{3/2}}\begin{pmatrix}
	1 & 6 \\
	-6 & -1\\
	\end{pmatrix} +O\big(\zeta^{-2}\big)\bigg]\frac{1}{2}\begin{pmatrix}
	1 & 1 \\
	1 & -1\\
	\end{pmatrix}\beta(z)^{-\sigma_3}\begin{pmatrix}
	1 & 1 \\
	1 & -1\\
	\end{pmatrix}M(z)\nonumber\\
	&=& \bigg[I + \frac{i}{16\sqrt{\zeta}}\begin{pmatrix}
	\beta^2-3\beta^{-2} & -(\beta^2+3\beta^{-2}) \\
	\beta^2+3\beta^{-2} & -(\beta^2-3\beta^{-2})\\
	\end{pmatrix}+\frac{3}{128\zeta}\begin{pmatrix}
	1 & -4\\
	-4 & 1\\
	\end{pmatrix}\nonumber\\
	&&+\frac{15i}{2048\zeta^{3/2}}\begin{pmatrix}
	-(5\beta^2-7\beta^{-2}) & 5\beta^2+7\beta^{-2} \\
	-(5\beta^2+7\beta^{-2}) & 5\beta^2-7\beta^{-2} \\
	\end{pmatrix}+O\big(\zeta^{-2}\big)\bigg]M(z)
\end{eqnarray}
as $s\rightarrow\infty$ and $0<r_1\leq|z-1|\leq r_2<1$ (so $|\zeta|\rightarrow\infty$). Since the function $\zeta(z)$ is of order $O\big(s^6\big)$ on the latter annulus and $\beta(z)$ is bounded, equation \eqref{Xrmatchup} yields the matching relation between the model functions $U(z)$ and $M(z)$,
\begin{equation*}
	U(z) = \big(I+o(1)\big)M(z),\hspace{0.5cm} s\rightarrow\infty,\ \ 0<r_1\leq|z-1|\leq r_2<1
\end{equation*}
which is crucial for the successful implementation of the nonlinear steepest descent method as we shall see after the next subsection. This is the reason for chosing the left multiplier $B_r(z)$ in \eqref{Xrpara} in the form \eqref{Brmultiplier}.
%--------------------------------------------------------------------------------------------------------------------------------------------------------------------------------------------------

\section{Construction of a parametrix at the edge point $z=-1$}\label{sec10}
The model RHP near the other endpoint $z=-1$ can be introduced in a similar way as we did it in the last section. First we consider on the punctured plane $\zeta\in\mathbb{C}\backslash\{0\}$
\begin{equation*}
		\tilde{P}_{BE}(\zeta)=\begin{pmatrix}
		e^{-i\frac{3\pi}{2}}\sqrt{\zeta}\Big(H_0^{(1)}\Big)'\big(e^{-i\frac{\pi}{2}}\sqrt{\zeta}\big) & -\sqrt{\zeta}\Big(H_0^{(2)}\Big)'\big(e^{-i\frac{\pi}{2}}\sqrt{\zeta}\big) \\
		-e^{i\frac{\pi}{2}}H_0^{(1)}\big(e^{-i\frac{\pi}{2}}\sqrt{\zeta}\big) & H_0^{(2)}\big(e^{-i\frac{\pi}{2}}\sqrt{\zeta}\big) \\
		\end{pmatrix},\ \ \ 0<\textnormal{arg}\ \zeta\leq 2\pi
\end{equation*}
satisfying
\begin{eqnarray*}
	\tilde{P}_{BE}(\zeta) &=& \sqrt{\frac{2}{\pi}}\zeta^{\sigma_3/4}\begin{pmatrix}
	-1 & -1 \\
	-i & i \\
	\end{pmatrix}\bigg[I+\frac{1}{8\sqrt{\zeta}}\begin{pmatrix}
	-1 & 2 \\
	-2 & 1 \\
	\end{pmatrix} +\frac{3}{128\zeta}\begin{pmatrix}
	-1 & -4 \\
	-4 & -1 \\
	\end{pmatrix}\\
	&&+\frac{15}{1024\zeta^{3/2}}\begin{pmatrix}
	-1 & 6 \\
	-6 & 1 \\
	\end{pmatrix} +O\big(\zeta^{-2}\big)\bigg]e^{\sqrt{\zeta}\sigma_3}
\end{eqnarray*}
as $\zeta\rightarrow\infty$ and $0<\textnormal{arg}\ \zeta\leq 2\pi$. Next, instead of \eqref{PBERH}, define
\begin{equation}\label{PBERHleft}
	\tilde{P}_{BE}^{RH}(\zeta)=  \left\{
                                 \begin{array}{ll}
                                   \tilde{P}_{BE}(\zeta)\begin{pmatrix}
                          1  & 0 \\
                          -i & 1 \\
                        \end{pmatrix}, & \hbox{arg $\zeta\in(0,\frac{5\pi}{6})$,} \bigskip\\
                                   \tilde{P}_{BE}(\zeta)\begin{pmatrix}
                          1 & i \\
                          0 & 1 \\
                        \end{pmatrix}, & \hbox{arg $\zeta\in(\frac{7\pi}{6},2\pi)$,} \bigskip \\
                                   \tilde{P}_{BE}(\zeta),& \hbox{arg $\zeta\in(\frac{5\pi}{6},\frac{7\pi}{6})$.}
                                 \end{array}
                               \right.
\end{equation}
which solves the model RHP of Figure \ref{fig7}. More precisely, the function $\tilde{P}_{BE}^{RH}(\zeta)$ has the following analytic properties
\begin{figure}[tbh]
  \begin{center}
  \psfragscanon
  \psfrag{1}{$\bigr(\begin{smallmatrix}
  0 & -i\\
  -i & 0\\
  \end{smallmatrix}\bigl)$}
  \psfrag{2}{$S_3=\bigr(\begin{smallmatrix}
  1 & 0\\
  i & 1\\
  \end{smallmatrix}\bigl)$}
  \psfrag{3}{$S_4=\bigr(\begin{smallmatrix}
  1 & i\\
  0 & 1\\
  \end{smallmatrix}\bigl)$}
  \includegraphics[width=4cm,height=3cm]{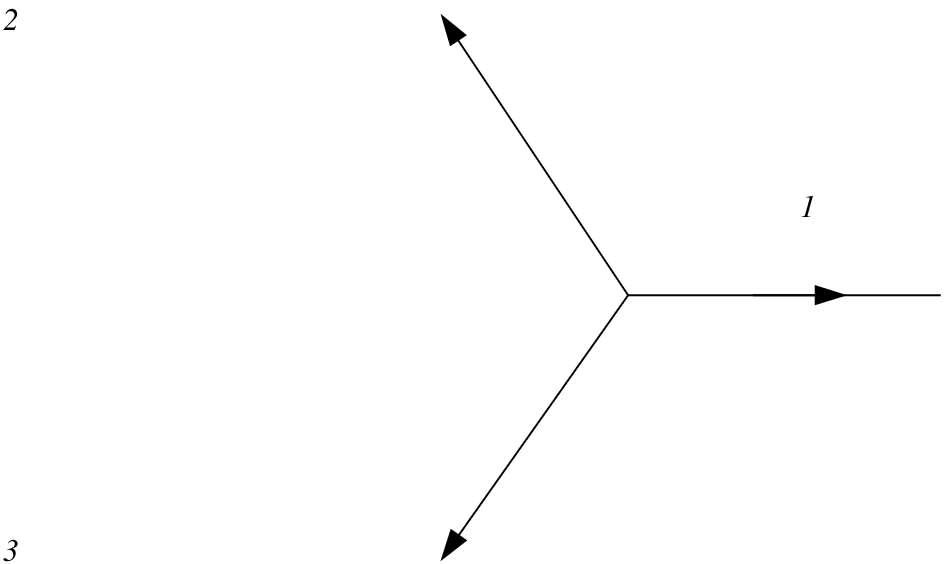}
  \end{center}
  \caption{The model RHP near $z=-1$ which can be solved explicitly using Hankel functions}
  \label{fig7}
\end{figure}
\begin{itemize}
	\item $\tilde{P}_{BE}^{RH}(\zeta)$ is analytic for $\zeta\in\mathbb{C}\backslash\{\textnormal{arg}\ \zeta = \frac{5\pi}{6},\frac{7\pi}{6},2\pi\}$
	\item We have the following jumps on the contour depicted in Figure \ref{fig7}
	\begin{eqnarray*}
		\big(\tilde{P}_{BE}^{RH}(\zeta)\big)_+&=&\big(\tilde{P}_{BE}^{RH}(\zeta)\big)_-\begin{pmatrix}
		1 & 0 \\
		i & 1 \\
		\end{pmatrix},\ \ \ \textnormal{arg}\ \zeta=\frac{5\pi}{6}\\
		\big(\tilde{P}_{BE}^{RH}(\zeta)\big)_+&=&\big(\tilde{P}_{BE}^{RH}(\zeta)\big)_-\begin{pmatrix}
		1 & i \\
		0 & 1 \\
		\end{pmatrix},\ \ \ \textnormal{arg}\ \zeta=\frac{7\pi}{6}
	\end{eqnarray*}
	and on the line segment $\textnormal{arg}\ \zeta=2\pi$
	\begin{eqnarray*}
		H_0^{(1)}\big(e^{-i\frac{\pi}{2}}\sqrt{\zeta}_+\big) &=& H_0^{(1)}\big(e^{-i\frac{\pi}{2}}\sqrt{\zeta}_-e^{-i\pi}\big) = H_0^{(2)}\big(e^{-i\frac{\pi}{2}}\sqrt{\zeta}_-\big)+2H_0^{(1)}\big(e^{-i\frac{\pi}{2}}\sqrt{\zeta}_-\big)\\		\big(H_0^{(1)}\big)'\big(e^{-i\frac{\pi}{2}}\sqrt{\zeta}_+\big)&=&e^{i\pi}\big(H_0^{(2)}\big)'\big(e^{-i\frac{\pi}{2}}\sqrt{\zeta}_-\big)+2e^{i\pi}\big(H_0^{(1)}\big)'\big(e^{-i\frac{\pi}{2}}\sqrt{\zeta}_-\big)\\
	\end{eqnarray*}
	as well as
	\begin{eqnarray*}
		H_0^{(2)}\big(e^{-i\frac{\pi}{2}}\sqrt{\zeta}_+\big)&=& H_0^{(2)}\big(e^{-i\frac{\pi}{2}}\sqrt{\zeta}_-e^{-i\pi}\big)=-H_0^{(1)}\big(e^{-i\frac{\pi}{2}}\sqrt{\zeta}_-\big)\\
		\big(H_0^{(2)}\big)'\big(e^{-i\frac{\pi}{2}}\sqrt{\zeta}_+\big)&=&\big(H_0^{(1)}\big)'\big(e^{-i\frac{\pi}{2}}\sqrt{\zeta}_-\big)
	\end{eqnarray*}
	hence 
	\begin{equation*}
		\big(\tilde{P}_{BE}^{RH}(\zeta)\big)_+ = \big(\tilde{P}_{BE}^{RH}(\zeta)\big)_-\begin{pmatrix}
		0 & -i \\
		-i & 0 \\
		\end{pmatrix},\ \ \ \textnormal{arg}\ \zeta=2\pi.
	\end{equation*}
	\item A similar argument as given in the construction of $P_{BE}^{RH}(\zeta)$ implies 
	\begin{eqnarray}\label{PBERHasyleftinfinity}
	\tilde{P}_{BE}^{RH}(\zeta) &=& \sqrt{\frac{2}{\pi}}\zeta^{\sigma_3/4}\begin{pmatrix}
	-1 & -1 \\
	-i & i \\
	\end{pmatrix}\bigg[I+\frac{1}{8\sqrt{\zeta}}\begin{pmatrix}
	-1 & 2 \\
	-2 & 1 \\
	\end{pmatrix} +\frac{3}{128\zeta}\begin{pmatrix}
	-1 & -4 \\
	-4 & -1 \\
	\end{pmatrix}\nonumber\\
	&&+\frac{15}{1024\zeta^{3/2}}\begin{pmatrix}
	-1 & 6 \\
	-6 & 1 \\
	\end{pmatrix} +O\big(\zeta^{-2}\big)\bigg]e^{\sqrt{\zeta}\sigma_3}
\end{eqnarray}
as $\zeta\rightarrow\infty$, valid in a full neighborhood of infinity.
\end{itemize}
Again we use the model function $\tilde{P}_{BE}^{RH}(\zeta)$ in the construction of the parametrix to the solution of the original $X$-RHP near $z=-1$. Instead of \eqref{BEchangeright}
\begin{equation}\label{BEchangeleft}
	\zeta(z) = s^6g^2(z),\ \ |z+1|<r,\ 0<\textnormal{arg}\ \zeta\leq 2\pi
\end{equation}
or equivalently
\begin{equation*}
	\sqrt{\zeta(z)} = -s^3g(z) = -\frac{4is^3}{3}\sqrt{z^2-1}\bigg(z^2+\frac{1}{2}+\frac{3x}{4s^2}\bigg).
\end{equation*}
This change of the independent variable is locally conformal 
\begin{equation*}
	\zeta(z)=\frac{32s^6}{9}\bigg(\frac{3}{2}+\frac{3x}{4s^2}\bigg)^2(z+1)\big(1+O(z+1)\big),\ \ |z+1|<r
\end{equation*}
and allows us to define the left parametrix $V(z)$ near $z=-1$ by the formula:
\begin{equation}\label{Xlpara}
	V(z) = B_l(z)\frac{1}{2}\begin{pmatrix}
	-1 & 0 \\
	0 & i \\
	\end{pmatrix}\sqrt{\frac{\pi}{2}}\tilde{P}_{BE}^{RH}\big(\zeta(z)\big)e^{s^3g(z)\sigma_3},\ \ |z+1|<r
\end{equation}
with the matrix multiplier
\begin{equation}\label{Blmultiplier}
	B_l(z)=\begin{pmatrix}
	1 & 1 \\
	1 & -1 \\
	\end{pmatrix}\bigg(\zeta(z)\frac{z-1}{z+1}\bigg)^{-\sigma_3/4},\ \ B_l(-1) =\begin{pmatrix}
	1 & 1\\
	1 & -1\\
	\end{pmatrix}\Bigg(\frac{8is^3}{3}\bigg(\frac{3}{2}+\frac{3x}{4s^2}\bigg)\Bigg)^{-\sigma_3/2}.
\end{equation}
Similar to the situation in the last section, $V(z)$ has jumps on the contour depicted in Figure \ref{fig8} which are described by the same Stokes matrices as in the original $T$-RHP. 
\begin{figure}[tbh]
  \begin{center}
  \psfragscanon
  \psfrag{1}{$\bigr(\begin{smallmatrix}
  0 & -i\\
  -i & 0\\
  \end{smallmatrix}\bigl)$}
  \psfrag{2}{\footnotesize{$e^{-s^3g(z)\sigma_3}S_3e^{s^3g(z)\sigma_3}$}}
  \psfrag{3}{\footnotesize{$e^{-s^3g(z)\sigma_3}S_4e^{s^3g(z)\sigma_3}$}}
  \includegraphics[width=6cm,height=3cm]{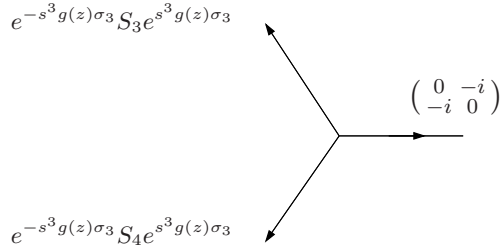}
  \end{center}
  \caption{Transformation of parametrix jumps to original jumps}
  \label{fig8}
\end{figure}

Also here, as we shall see in detail in the section \ref{sec13}, the singular behavior at $z=-1$ matches:
\begin{equation}\label{Vsing}
	V(z)=O\big(\ln(z+1)\big),\ \ |z+1|<r
\end{equation}
Hence the ratio of parametrix $V(z)$ with $T(z)$ is locally analytic
\begin{equation*}
	T(z)=N_l(z)V(z),\ \ |z+1|<r<\frac{1}{2}
\end{equation*}
and the left multiplier \eqref{Blmultiplier} in \eqref{Xlpara} provides us with the following asymptotic matchup between $V(z)$ and $M(z)$:
\begin{eqnarray}\label{Xlmatchup}
	V(z) &=& \begin{pmatrix}
	1 & 1 \\
	1 & -1\\
	\end{pmatrix}\beta(z)^{\sigma_3}\frac{1}{2}\begin{pmatrix}
	1 & 1 \\
	1 & -1\\
	\end{pmatrix}\bigg[I+\frac{1}{8\sqrt{\zeta}}\begin{pmatrix}
	-1 & 2 \\
	-2 & 1 \\
	\end{pmatrix} +\frac{3}{128\zeta}\begin{pmatrix}
	-1 & -4 \\
	-4 & -1 \\
	\end{pmatrix}\nonumber\\
	&&+\frac{15}{1024\zeta^{3/2}}\begin{pmatrix}
	-1 & 6 \\
	-6 & 1 \\
	\end{pmatrix} +O\big(\zeta^{-2}\big)\bigg]\frac{1}{2}\begin{pmatrix}
	1 & 1 \\
	1 & -1 \\
	\end{pmatrix}\beta(z)^{-\sigma_3}\begin{pmatrix}
	1 & 1 \\
	1 & -1 \\
	\end{pmatrix}M(z)\nonumber\\
	&=&\bigg[I+\frac{1}{16\sqrt{\zeta}}\begin{pmatrix}
	\beta^{-2}-3\beta^2 & \beta^{-2}+3\beta^2 \\
	-(\beta^{-2}+3\beta^2) & -(\beta^{-2}-3\beta^2) \\
	\end{pmatrix}+\frac{3}{128\zeta}\begin{pmatrix}
	-1 & -4 \\
	-4 & -1 \\
	\end{pmatrix}\nonumber\\
	&& +\frac{15}{2048\zeta^{3/2}}\begin{pmatrix}
	5\beta^{-2}-7\beta^2 & 5\beta^{-2}+7\beta^2 \\
	-(5\beta^{-2}+7\beta^2) & -(5\beta^{-2}-7\beta^2) \\
	\end{pmatrix}+O\big(\zeta^{-2}\big)\bigg]M(z)
\end{eqnarray}
as $s\rightarrow\infty$ and $0<r_1\leq|z+1|\leq r_2<1$, thus
\begin{equation*}
	V(z) = \big(I+o(1)\big)M(z),\hspace{0.5cm}s\rightarrow\infty,\ \ 0<r_1\leq|z+1|\leq r_2<1.
\end{equation*}
At this point we can use the model functions $M(z),U(z)$ and $V(z)$ to employ the final transformation.
%-----------------------------------------------------------------------------------------------

\section{Third and final transformation of the RHP}\label{sec11}
In this final transformation we put
\begin{equation}\label{errorfunction}
	R(z)=T(z)\left\{
                 \begin{array}{ll}
                   \big(U(z)\big)^{-1}, & \hbox{$|z-1|<r$,} \\
                   \big(V(z)\big)^{-1}, & \hbox{$|z+1|<r$,} \\
                   \big(M(z)\big)^{-1}, & \hbox{$|z\mp 1|>r$} 
                 \end{array}
               \right.
\end{equation}
where $0<r<\frac{1}{4}$ remains fixed. With $C_r$ and $C_l$ denoting the clockwise oriented circles shown in Figure \ref{fig9}, the ratio-function $R(z)$ solves the following RHP
\begin{figure}[tbh]
  \begin{center}
  \psfragscanon
  \psfrag{1}{\footnotesize{$C_r$}}
  \psfrag{2}{\footnotesize{$C_l$}}
  \psfrag{3}{\footnotesize{$\gamma_1$}}
  \psfrag{4}{\footnotesize{$\gamma_3$}}
  \psfrag{5}{\footnotesize{$\gamma_4$}}
  \psfrag{6}{\footnotesize{$\gamma_6$}}
  \includegraphics[width=7cm,height=4cm]{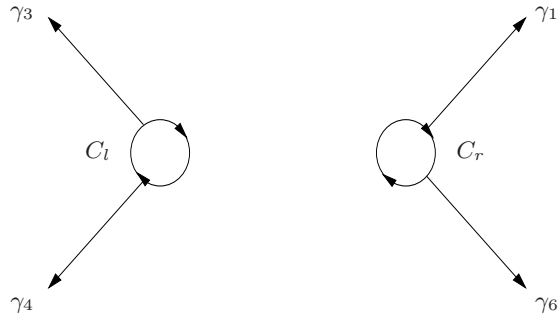}
  \end{center}
  \caption{The jump graph for the ratio-function $R(z)$}
  \label{fig9}
\end{figure}

\begin{itemize}
	\item $R(z)$ is anlytic for $z\in\mathbb{C}\backslash\big\{C_r\cup C_l\cup\bigcup_k\gamma_k\big\}$
	\item Along the infinite branches $\gamma_k$
	\begin{equation*}
		R_+(z)=R_-(z)M(z)e^{-s^3g(z)\sigma_3}S_ke^{s^3g(z)\sigma_3}\big(M(z)\big)^{-1},\ \ z\in\gamma_k,
	\end{equation*}
	and on the clockwise oriented circles $C_{r,l}$
	\begin{eqnarray*}
		R_+(z)&=&R_-(z)U(z)\big(M(z)\big)^{-1},\ z\in C_r\\
		R_+(z)&=&R_-(z)V(z)\big(M(z)\big)^{-1},\ z\in C_l
	\end{eqnarray*}
	\item $R(z)$ is analytic at $z=\pm 1$. This observation follows from \eqref{Using} and \eqref{Vsing}, which will be proved in section \ref{sec13}.
	\item In a neigborhood of infinity, we have $R(z)\rightarrow I$
\end{itemize}
We emphasize that, by construction, the function $R(z)$ has no jumps inside of $C_r$ and $C_l$ and across the line segment in between them.
In order to apply the Deift-Zhou nonlinear steepest descent method for the ratio-RHP, all its jump matrices have to be close to the unit matrix, as $s\rightarrow\infty$, see \cite{DZ1}. Due to the triangularity of all Stokes matrices $S_k$, the jump matrices corresponding to the infinite parts $\bigcup_k\gamma_k$ of the $R$-jump contour are in fact exponentially close to the unit matrix
\begin{equation}\label{esti1}
	\|Me^{-s^3g(\cdot)\sigma_3}S_ke^{s^3g(\cdot)\sigma_3}\big(M\big)^{-1}-I\|_{L^2\cap L^{\infty}(\gamma_k)}\leq c_1e^{-c_2s^3|z\mp 1|}
\end{equation}
emanating from $C_{r,l}$ as $s\rightarrow\infty$ with constants $c_i>0$ whose values are not important. Moreover, by virtue of \eqref{Xrmatchup}, $U(z)\big(M(z)\big)^{-1}$ approaches the unit matrix as $s\rightarrow\infty$
\begin{equation}\label{Uesti}
	\|U\big(M\big)^{-1}-I\|_{L^2\cap L^{\infty}(C_r)}\leq c_3s^{-3}
\end{equation}
and from \eqref{Xlmatchup}, also $V(z)\big(M(z)\big)^{-1}$
\begin{equation}\label{Vesti}
	\|V\big(M\big)^{-1}-I\|_{L^2\cap L^{\infty}(C_l)}\leq c_4s^{-3},\ s\rightarrow\infty.
\end{equation}
All together, with $G_R$ denoting the jump matrix in the latter ratio-RHP and $\Sigma_R$ the underlying contour
\begin{equation}\label{esti3}
	\| G_R-I\|_{L^2\cap L^{\infty}(\Sigma_R)}\leq cs^{-3},\ \ s\rightarrow\infty
\end{equation}
uniformly on any compact subset of the set \eqref{excset1}. The latter estimation enables us to solve the ratio-RHP iteratively.
%-----------------------------------------------------------------------------------------------------------------------------------------------------------------------------------------

\section{Solution of the RHP for $R(z)$ via iteration}\label{sec12}

The RHP for the function $R(z)$,
\begin{itemize}
	\item $R(z)$ is analytic for $z\in\mathbb{C}\backslash\Sigma_R$
	\item Along the contour depicted in Figure \ref{fig9}
	\begin{equation*}
		R_+(z) = R_-(z)G_R(z),\ \ z\in\Sigma_R
	\end{equation*}
	\item At infinity, $R(z)=I+O\big(z^{-1}\big),z\rightarrow\infty$
\end{itemize}
is equivalent to the singular integral equation
\begin{equation}\label{singularintegral}
	R_-(z)=I+\frac{1}{2\pi i}\int\limits_{\Sigma_R}R_-(w)\big(G_R(w)-I\big)\frac{dw}{w-z_-}
\end{equation}
and by standard arguments (see \cite{DZ1}) we know that for sufficiently large $s$ the relevant integral operator is contracting and equation \eqref{singularintegral} can be solved iteratively in $L^2(\Sigma_R)$. Moreover, its unique solution satisfies
\begin{equation}\label{esti4}
		\|R_--I\|_{L^2(\Sigma_R)}\leq cs^{-3},\ \ s\rightarrow\infty.
\end{equation}
It is now time to connect the latter information with \eqref{sidentity} and \eqref{xidentity} to determine the large $s$ asymptotics of $\det(I-K_{\textnormal{PII}})$ up to the constant term.
%---------------------------------------------------------------------------------------------------

\section{Asymptotics of $\ln\det(I-K_{\textnormal{PII}})$ - preliminary steps}\label{sec13}
Let us recall the transformations, which have been used in the asymptotical solution of the original $Y$-RHP
\begin{equation*}
	Y(\lambda)\mapsto \tilde{X}(\lambda)\mapsto X(\lambda)\mapsto T(z)\mapsto R(z).
\end{equation*}
Thus, in order to determine $\ln\det(I-K_{\textnormal{PII}})$ via Proposition \ref{prop1}
we need to connect $\check{X}(\pm s)$ and $\check{X}'(\pm s)$ to the values $R(\pm 1)$ and $R'(\pm 1)$ of the ratio-function. This can be done as follows: From \eqref{errorfunction} and \eqref{Tsingular} for $|z-1|<r$
\begin{equation}\label{rightcompar}
	R(z)U(z)e^{-s^3g(z)\sigma_3} = \check{X}(zs)\begin{pmatrix}
	1 & -\frac{1}{2\pi}\ln\frac{z-1}{z+1}\\
	0 & 1\\
	\end{pmatrix}\left\{
                                   \begin{array}{ll}
                                     I, & \hbox{$\lambda\in\hat{\Omega}_1$,} \\
                                     S_3S_4, & \hbox{$\lambda\in\hat{\Omega}_3$,} \\
                                     S_3S_4S_6, & \hbox{$\lambda\in\hat{\Omega}_4$,}
                                   \end{array}
                                 \right.
\end{equation}
and for $|z+1|<r$
\begin{equation}\label{leftcompar}
	R(z)V(z)e^{-s^3g(z)\sigma_3} = \check{X}(zs)\begin{pmatrix}
	1 & -\frac{1}{2\pi}\ln\frac{z-1}{z+1}\\
	0 & 1\\
	\end{pmatrix}\left\{
                                   \begin{array}{ll}
                                     I, & \hbox{$\lambda\in\hat{\Omega}_1$,} \\
                                     S_3, & \hbox{$\lambda\in\hat{\Omega}_2$,} \\
                                     S_3S_4, & \hbox{$\lambda\in\hat{\Omega}_3$.} 
                                   \end{array}
                                 \right.
\end{equation}
This shows that the required values of $\check{X}(\pm s)$ and $\check{X}'(\pm s)$ can be determined via comparison in either \eqref{rightcompar} or \eqref{leftcompar}, once we know the local expansions of $U(z)$, respectively $V(z)$ at $z=\pm 1$. Our starting point is \eqref{Hankelorigin}
\begin{eqnarray}\label{PBERHlocal}
	P_{BE}^{RH}(\zeta) &=& \begin{pmatrix}
	\bar{a}_0+\frac{\bar{a}_1}{2}\ln\zeta & i(a_0+\frac{a_1}{2}\ln\zeta)\\
	\bar{a}_1 & a_1
	\end{pmatrix}\nonumber\\
	&& + \zeta\begin{pmatrix}
	\bar{a}_2 +\frac{\bar{a}_3}{2}\ln\zeta & i(a_2+\frac{a_3}{2}\ln\zeta) \\
	2\bar{a}_2+\bar{a}_3+\bar{a}_3\ln\zeta & 2 a_2+a_3+a_3\ln\zeta \\
	\end{pmatrix}+O\big(\zeta^2\ln\zeta),
\end{eqnarray}
as $\zeta\rightarrow 0$ and $-\frac{\pi}{6}<\textnormal{arg}\ \zeta <\frac{\pi}{6}$. The latter expansion together with the changes of variables $\zeta=\zeta(z) = -s^6g^2(z)$ and $\lambda=zs$ implies for $-\frac{\pi}{6}<\textnormal{arg}\ (\lambda-s)<\frac{\pi}{6}$
\begin{equation*}
	P_{BE}^{RH}\big(\zeta(z)\big) = \Omega\big(\ln(\lambda-s)\big)+(\lambda-s)\Pi\big(\ln(\lambda-s)\big)+O\big((\lambda-s)^2\ln(\lambda-s)\big),\ \ \lambda\rightarrow s
\end{equation*}
with the matrix functions $\Omega=(\Omega_{ij})$ and $\Pi=(\Pi_{ij})$ being determined from \eqref{PBERHlocal}. Now we combine the latter expansion with \eqref{Xrpara} and \eqref{errorfunction}
\begin{eqnarray*}
	&&R(z)U(z)e^{-s^3g(z)\sigma_3}\Big|_{z=\frac{\lambda}{s}} = R(z)B_r(z)\frac{\sigma_3}{2}\sqrt{\frac{\pi}{2}}e^{-i\frac{\pi}{4}}P_{BE}^{RH}\big(\zeta(z)\big)\Big|_{z=\frac{\lambda}{s}}\\
			 &=& R(1)B_r(1)\frac{\sigma_3}{2}\sqrt{\frac{\pi}{2}}e^{-i\frac{\pi}{4}}\Omega\big(\ln(\lambda-s)\big)+(\lambda-s)\bigg\{R(1)B_r(1)\\
	&&\times\frac{\sigma_3}{2}\sqrt{\frac{\pi}{2}}e^{-i\frac{\pi}{4}}\Pi\big(\ln(\lambda-s)\big)+\big(R'(1)B_r(1)+R(1)B_r'(1)\big)\frac{\sigma_3}{2}\sqrt{\frac{\pi}{2}}e^{-i\frac{\pi}{4}}\frac{\Omega\big(\ln(\lambda-s)\big)}{s}\bigg\}\\
	&&+O\big((\lambda-s)^2\ln(\lambda-s)\big),\ \ \lambda\rightarrow s, \ \ -\frac{\pi}{6}<\textnormal{arg}\ (\lambda-s)<\frac{\pi}{6}
\end{eqnarray*}
and similar identities hold for $-\pi<\textnormal{arg}\ (\lambda-s)<-\frac{\pi}{6}$ and $\frac{\pi}{6}<\textnormal{arg}\ (\lambda-s)<\pi$, they differ from \eqref{PBERHlocal} only by right multiplication with a Stokes matrix (see \eqref{PBERH}). On the other hand the right hand side in \eqref{rightcompar} implies for $-\frac{\pi}{6}<\textnormal{arg}\ (\lambda-s)<\frac{\pi}{6}$
\begin{eqnarray*}
    &&T\Big(\frac{\lambda}{s}\Big)e^{-s^3g(\frac{\lambda}{s})\sigma_3}=
    \Big[\check{X}(s) +(\lambda-s)\check{X}'(s) +O\big((\lambda-s)^2\big)\Big]
    \begin{pmatrix}
                                   1 & -\frac{1}{2\pi}\ln\frac{\lambda-s}{\lambda+s} \\
                                   0 & 1 \\
                                 \end{pmatrix}\begin{pmatrix}
                                                1 & 0 \\
                                                i & 1 \\
                                              \end{pmatrix}\\
&=&\Bigg[\begin{pmatrix}
     \check{X}_{11}(s) & \frac{1}{2\pi}\ln(2s)\check{X}_{11}(s)+\check{X}_{12}(s) \\
     \check{X}_{21}(s) & \frac{1}{2\pi}\ln(2s)\check{X}_{21}(s)+\check{X}_{22}(s) \\
   \end{pmatrix}\\
   &&
+(\lambda-s)\begin{pmatrix}
                \check{X}_{11}'(s) & \frac{1}{2\pi}(\ln(2s)\check{X}_{11}'(s)+\frac{1}{2s}\check{X}_{11}(s))+\check{X}_{12}'(s) \\
                \check{X}_{21}'(s) & \frac{1}{2\pi}(\ln(2s)\check{X}_{21}'(s)+\frac{1}{2s}\check{X}_{21}(s))+\check{X}_{22}'(s) \\
              \end{pmatrix}\Bigg]\\
              &&\times\begin{pmatrix}
                             1-\frac{i}{2\pi}\ln(\lambda-s) & -\frac{1}{2\pi}\ln(\lambda-s) \\
                             i & 1 \\
                           \end{pmatrix}+O\big((\lambda-s)^2\ln(\lambda-s)\big),\ \ \lambda\rightarrow s
\end{eqnarray*}
and thus by comparison of the left hand side and right hand side in \eqref{rightcompar}
\begin{equation}\label{Xcheck11}
	\check{X}_{11}(s) = \sqrt{\frac{\pi}{2}}e^{-i\frac{\pi}{4}}\big(R(1)B_r(1)\big)_{11},\ \ \ \ \check{X}_{21}(s)=\sqrt{\frac{\pi}{2}}e^{-i\frac{\pi}{4}}\big(R(1)B_r(1)\big)_{21}.
\end{equation}
Although \eqref{Xcheck11} was derived from a comparison in the sector $-\frac{\pi}{6}<\textnormal{arg}\ (\lambda-s)<\frac{\pi}{6}$, the multiplication of \eqref{PBERHlocal} with the right Stokes matrix in the other sectors as well as the use of appropriate Stokes matrices in \eqref{Xlocal} show that \eqref{Xcheck11} follows from comparison in a full neighborhood of $\lambda =+s$. Comparing terms of $O\big((\lambda-s)\ln(\lambda-s)\big)$ we also derive
\begin{eqnarray}\label{Xcheck211}
	\check{X}_{11}'(s)&=&-\sqrt{\frac{\pi}{2}}e^{-i\frac{\pi}{4}}\bigg[\frac{8s^5}{9}\bigg(\frac{3}{2}+\frac{3x}{4s^2}\bigg)^2
	\Big(\big(R(1)B_r(1)\big)_{11}+2i\big(R(1)B_r(1)\big)_{12}\Big)\nonumber\\
	&&-\frac{1}{s}\big(R'(1)B_r(1)+R(1)B_r'(1)\big)_{11}\bigg]
\end{eqnarray}
and
\begin{eqnarray}\label{Xcheck221}
	\check{X}_{21}'(s) &=&-\sqrt{\frac{\pi}{2}}e^{-i\frac{\pi}{4}}\bigg[\frac{8s^5}{9}\bigg(\frac{3}{2}+\frac{3x}{4s^2}\bigg)^2
	\Big(\big(R(1)B_r(1)\big)_{21}+2i\big(R(1)B_r(1)\big)_{22}\Big)\nonumber\\
	&&-\frac{1}{s}\big(R'(1)B_r(1)+R(1)B_r'(1)\big)_{21}\bigg].
\end{eqnarray}
All together we have the identities
\begin{equation*}
	F_1(s)=\frac{i}{\sqrt{2\pi}}\check{X}_{11}(s),\  F_2(s)=\frac{i}{\sqrt{2\pi}}\check{X}_{21}(s),\  F_1'(s)=\frac{i}{\sqrt{2\pi}}\check{X}_{11}'(s),\  F_2'(s)=\frac{i}{\sqrt{2\pi}}\check{X}_{21}'(s)
\end{equation*}
related to the solution of the ratio-RHP via \eqref{Xcheck11},\eqref{Xcheck211} and \eqref{Xcheck221}. A completely similar analysis for the left endpoint $\lambda=-s$ provides us with
\begin{equation}\label{F12minus}
	F_1(-s) = \frac{i}{2}\big(R(-1)B_l(-1)\big)_{12},\ \ \ F_2(-s)=\frac{i}{2}\big(R(-1)B_l(-1)\big)_{22}
\end{equation}
and
\begin{eqnarray}\label{Fprime1minus}
	F_1'(-s) &=&\frac{i}{2}\bigg[\frac{8s^5}{9}\bigg(\frac{3}{2}+\frac{3x}{4s^2}\bigg)^2\Big(\big(R(-1)B_l(-1)\big)_{12}+2\big(R(-1)B_l(-1)\big)_{11}\Big)\nonumber\\
	&&+\frac{1}{s}\big(R(-1)B_l'(-1)+R'(-1)B_l(-1)\big)_{12}\bigg]
\end{eqnarray}
as well as
\begin{eqnarray}\label{Fprime2minus}
	F_2'(-s) &=&\frac{i}{2}\bigg[\frac{8s^5}{9}\bigg(\frac{3}{2}+\frac{3x}{4s^2}\bigg)^2\Big(\big(R(-1)B_l(-1)\big)_{22}+2\big(R(-1)B_l(-1)\big)_{21}\Big)\nonumber\\
	&&+\frac{1}{s}\big(R(-1)B_l'(-1)+R'(-1)B_l(-1)\big)_{22}\bigg].
\end{eqnarray}
We can now derive 
\begin{eqnarray*}
	&&R(s,s)= F_1'(s)F_2(s)-F_2'(s)F_1(s)\\
	&=& -\frac{4s^5}{9}\bigg(\frac{3}{2}+\frac{3x}{4s^2}\bigg)^2\Big[\big(R(1)B_r(1)\big)_{11}\big(R(1)B_r(1)\big)_{22}-\big(R(1)B_r(1)\big)_{21}\big(R(1)B_r(1)\big)_{12}\Big]\\
	&&+\frac{i}{4s}\Big[\big(R'(1)B_r(1)+R(1)B_r'(1)\big)_{11}\big(R(1)B_r(1)\big)_{21}-\big(R'(1)B_r(1)+R(1)B_r'(1)\big)_{21}\\
	&&\times\big(R(1)B_r(1)\big)_{11}\Big]
\end{eqnarray*}
as well as
\begin{eqnarray*}
	&&R(-s,-s)=F_1'(-s)F_2(-s)-F_2'(-s)F_1(-s)\\
	&=&-\frac{4s^5}{9}\bigg(\frac{3}{2}+\frac{3x}{4s^2}\bigg)^2\Big[\big(R(-1)B_l(-1)\big)_{11}\big(R(-1)B_l(-1)\big)_{22}-\big(R(-1)B_l(-1)\big)_{21}\\
	&&\times\big(R(-1)B_l(-1)\big)_{12}\Big]-\frac{1}{4s}\Big[\big(R(-1)B_l'(-1)+R'(-1)B_l(-1)\big)_{12}\\
	&&\times\big(R(-1)B_l(-1)\big)_{22}-\big(R(-1)B_l'(-1)+R'(-1)B_l(-1)\big)_{22}\big(R(-1)B_l(-1)\big)_{12}\Big].\\
\end{eqnarray*}
We make the following observation
\begin{proposition} $R(z)$ is unimodular for any $x\in\mathbb{R}$, i.e. $\det R(z)\equiv 1$.
\end{proposition}
\begin{proof} From \eqref{PBERH} we obtain that $\det P_{BE}^{RH}(\zeta)=\frac{4i}{\pi}$, hence $\det U(z)=1$. Similarly $\det\tilde{P}_{BE}^{RH}(\zeta)= -\frac{4i}{\pi}$ leading to $\det V(z)= 1$. Thus the ratio-RHP has a unimodular jump matrix $G_R(z)$ which implies by normalization at infinity $\det R(z)\equiv 1$.
\end{proof}
Applying the latter Proposition, one checks now readily
\begin{equation*}
	\big(R(1)B_r(1)\big)_{11}\big(R(1)B_r(1)\big)_{22}-\big(R(1)
	B_r(1)\big)_{21}\big(R(1)B_r(1)\big)_{12} = -2
\end{equation*}
and
\begin{equation*}
	\big(R(-1)B_l(-1)\big)_{11}\big(R(-1)B_l(-1)\big)_{22}-\big(R(-1)
	B_l(-1)\big)_{21}\big(R(-1)B_l(-1)\big)_{12} = -2.
\end{equation*}
We combine these two identities with the values of $B_r'(1)$ and $B_l'(-1)$ to deduce
\begin{eqnarray}\label{resolventright}
	R(s,s)&=&\frac{8s^5}{9}\bigg(\frac{3}{2}+\frac{3x}{4s^2}\bigg)^2+\frac{2is^2}{3}\bigg(\frac{3}{2}+\frac{3x}{4s^2}\bigg)
	\Big[\big(R'_{11}(1)+R_{12}'(1)\big)\big(R_{21}(1)+R_{22}(1)\big)\nonumber\\
	&&-\big(R_{21}'(1)+R_{22}'(1)\big)\big(R_{11}(1)+R_{12}(1)\big)\Big]
\end{eqnarray}
as well as
\begin{eqnarray}\label{resolventleft}
	R(-s,-s)&=&\frac{8s^5}{9}\bigg(\frac{3}{2}+\frac{3x}{4s^2}\bigg)^2-\frac{2is^2}{3}\bigg(\frac{3}{2}+\frac{3x}{4s^2}\bigg)
	\Big[\big(R_{11}'(-1)-R_{12}'(-1)\big)\\
	&&\times\big(R_{21}(-1)-R_{22}(-1)\big)-\big(R_{21}'(-1)-R_{22}'(-1)\big)\big(R_{11}(-1)-R_{12}(-1)\big)\Big]\nonumber.
\end{eqnarray}

%--------------------------------------------------------------------------------------------------------------------------------------------------------------------------------
\section{Asymptotics of $\ln\det(I-K_{\textnormal{PII}})$ up to the constant term}\label{sec14}

The stated asymptotics \eqref{theo1} without the constant term is a direct consequence of Proposition \ref{prop1} and \ref{prop2}. We first need to compute $R(\pm 1)$. From the integral representation and estimates \eqref{esti3},\eqref{esti4}
\begin{eqnarray*}
	R(\pm 1) &=& I+\frac{1}{2\pi i}\int\limits_{\Sigma_R}R_-(w)\big(G_R(w)-I\big)\frac{dw}{w\mp 1}\\
	&=& I+\frac{1}{2\pi i}\int\limits_{C_{r,l}}\big(G_R(w)-I\big)\frac{dw}{w\mp 1} +O\big(s^{-6}\big) = I+O\big(s^{-3}\big),\ \ s\rightarrow\infty
\end{eqnarray*}
so
\begin{eqnarray*}
	R(s,s) &=& \frac{8s^5}{9}\bigg(\frac{3}{2}+\frac{3x}{4s^2}\bigg)^2+\frac{2is^2}{3}\bigg(\frac{3}{2}+\frac{3x}{4s^2}\bigg)\Big[R'_{11}(1)-R'_{22}(1)\\
	&&+R_{12}'(1)-R'_{21}(1)\Big]+O\big(s^{-4}\big)
\end{eqnarray*}
and
\begin{eqnarray*}
	R(-s,-s)&=&\frac{8s^5}{9}\bigg(\frac{3}{2}+\frac{3x}{4s^2}\bigg)^2+\frac{2is^2}{3}\bigg(\frac{3}{2}+\frac{3x}{4s^2}\bigg)\Big[R'_{11}(-1)-R'_{22}(-1)\\
	&&-R'_{12}(-1)+R'_{21}(-1)\Big]+O\big(s^{-4}\big),\ \ s\rightarrow\infty.
\end{eqnarray*}
In order to compute the values $R'(\pm 1)$ one uses \eqref{Xrmatchup} and \eqref{Xlmatchup}
\begin{eqnarray*}
	R'(\pm 1)&=&\frac{1}{2\pi i}\int\limits_{C_{r,l}}\big(G_R(w)-I\big)\frac{dw}{(w\mp 1)^2} + O\big(s^{-6}\big)\\
	&=&\frac{1}{2\pi i}\int\limits_{C_r}\frac{i}{16\sqrt{\zeta(w)}}\begin{pmatrix}
	\beta^2-3\beta^{-2} & -(\beta^2+3\beta^{-2})\\
	\beta^2+3\beta^{-2} & -(\beta^2-3\beta^{-2})\\
	\end{pmatrix}\frac{dw}{(w\mp 1)^2}\\
	&&+\frac{1}{2\pi i}\int_{C_l}\frac{1}{16\sqrt{\zeta(w)}}\begin{pmatrix}
	\beta^{-2}-3\beta^2 & \beta^{-2}+3\beta^2 \\
	-(\beta^{-2}+3\beta^2) & -(\beta^{-2}-3\beta^2) \\
	\end{pmatrix}\frac{dw}{(w\mp 1)^2}+O\big(s^{-6}\big)
\end{eqnarray*}
with the local variables given in \eqref{BEchangeright},\eqref{BEchangeleft}:
\begin{equation*}
	w\in C_r:\ \frac{\beta^2(w)}{\sqrt{\zeta(w)}}=\frac{3}{4s^3}\bigg(w^2+\frac{1}{2}+\frac{3x}{4s^2}\bigg)^{-1}\frac{1}{w-1},\ \frac{\beta^{-2}(w)}{\sqrt{\zeta(w)}}=\frac{3}{4s^3}\bigg(w^2+\frac{1}{2}+\frac{3x}{4s^2}\bigg)^{-1}\frac{1}{w+1}
\end{equation*}
\begin{equation*}
	w\in C_l:\ \frac{\beta^2(w)}{\sqrt{\zeta(w)}} = \frac{3i}{4s^3}\bigg(w^2+\frac{1}{2}+\frac{3x}{4s^2}\bigg)^{-1}\frac{1}{w-1},\ \frac{\beta^{-2}(w)}{\sqrt{\zeta(w)}} = \frac{3i}{4s^3}\bigg(w^2+\frac{1}{2}+\frac{3x}{4s^2}\bigg)^{-1}\frac{1}{w+1}.
\end{equation*}
In the end residue theorem leads us to
\begin{eqnarray*}
	R'(\pm 1)\bigg(\frac{3}{2}+\frac{3x}{4s^2}\bigg) &=&\frac{3i}{256s^3}\begin{pmatrix}
	-1 & \mp 1\\
	\pm 1 & 1\\
	\end{pmatrix}\\
	&&+\frac{3i}{64s^3}\bigg(\frac{3}{2}+\frac{3x}{4s^2}\bigg)^{-1}\begin{pmatrix}
	-\frac{25}{8}-\frac{9x}{16s^2} & \mp(\frac{41}{8}+\frac{9x}{16s^2})\\
	\pm(\frac{41}{8}+\frac{9x}{16s^2}) & \frac{25}{8}+\frac{9x}{16s^2} \\
	\end{pmatrix}\\
	&&+\frac{3i}{16s^3}\bigg(\frac{3}{2}+\frac{3x}{4s^2}\bigg)^{-2}\begin{pmatrix}
	-1 & \pm 1\\
	\mp 1 & 1\\
	\end{pmatrix} +O\big(s^{-6}\big),
\end{eqnarray*}
and 
\begin{equation}\label{rss}
	R(s,s)=R(-s,-s)=\frac{8s^5}{9}\bigg(\frac{3}{2}+\frac{3x}{4s^2}\bigg)^2+\frac{3}{8s}+O\big(s^{-3}\big),\ \ s\rightarrow\infty.
\end{equation}
Combining \eqref{rss} with \eqref{sidentity} we have thus derived the following asymptotics
\begin{equation}\label{logsasy}
	\frac{d}{ds}\ln\det(I-K_{\textnormal{PII}})=-4s^5-4xs^3-x^2s-\frac{3}{4s}+O\big(s^{-3}\big),\ \ s\rightarrow\infty
\end{equation}
and the error term is uniform on any compact subset of the set \eqref{excset1}.
\bigskip

Opposed to \eqref{logsasy} we are going to determine $\det(I-K_{\textnormal{PII}})$ now via Proposition \ref{prop2}
\begin{equation*}
	\frac{d}{dx}\ln\det(I-K_{\textnormal{PII}})=i\big(X_1^{11}-X_1^{22}\big)-v
\end{equation*}
where
\begin{equation*}
	X_1^{ii} = \lim_{\lambda\rightarrow\infty}\Big(\lambda\big(X(\lambda)e^{i(\frac{4}{3}\lambda^3+x\lambda)\sigma_3}-I\big)_{ii}\Big),\ \ i=1,2.
\end{equation*}
To this end remember the definition of $\beta(z)$ and $g(z)$, hence as $z\rightarrow\infty$
\begin{equation*}
	M(z) = I+\frac{1}{2z}\begin{pmatrix}
	0 & 1 \\
	1 & 0 \\
	\end{pmatrix}+O\big(z^{-2}\big),\  e^{s^3(\vartheta(z)-g(z))\sigma_3} = I+\frac{is^3}{2z}\Big(1+\frac{x}{s^2}\Big)\sigma_3+O\big(z^{-2}\big),
\end{equation*}
which gives
\begin{equation*}
	X_1^{ii}=s\bigg[\frac{is^3}{2}\Big(1+\frac{x}{s^2}\Big)\sigma_3+\frac{1}{2}\begin{pmatrix}
	0 & 1\\
	1 & 0\\
	\end{pmatrix}+\frac{i}{2\pi}\int\limits_{C_{r,l}}\Big(R_-(w)\big(G_R(w)-I\big)\Big)dw\bigg]_{ii}
\end{equation*}
already neglecting exponentially small contributions in the last equality. The latter integral can be evaluated in a similar way as we did it during the computation of \eqref{logsasy}, we end up with
\begin{eqnarray*}
	&&\frac{i}{2\pi}\int\limits_{C_{r,l}}\Big(R_-(w)\big(G_R(w)-I\big)\Big)_{11}dw=\frac{3i}{32s^3}\bigg(\frac{3}{2}+\frac{3x}{4s^2}\bigg)^{-1}+O\big(s^{-6}\big)\\
	&=& -\frac{i}{2\pi}\int\limits_{C_{r,l}}\Big(R_-(w)\big(G_R(w)-I\big)\Big)_{22}dw,
\end{eqnarray*}
i.e. together
\begin{eqnarray}\label{logxasy}
	\frac{d}{dx}\ln\det(I-K_{\textnormal{PII}}) &=& 2is\bigg[\frac{it}{2}\Big(1+\frac{x}{s^2}\Big)+\frac{3i}{32t}\bigg(\frac{3}{2}+\frac{3x}{4s^2}\bigg)^{-1}\bigg]-v+O\big(s^{-5}\big)\nonumber\\
	&=&-s^4-s^2x-v-\frac{1}{8s^2}+O\big(s^{-4}\big),\ \ \ s\rightarrow\infty,
\end{eqnarray}
again uniformly on any compact subset of the set \eqref{excset1}. Given the asymptotic expansions \eqref{logsasy} and \eqref{logxasy} we can now determine the large $s$-asymptotics of $\ln\det(I-K_{\textnormal{PII}})$ via integration
\begin{equation}\label{asywithoutconstant}
	\ln\det(I-K_s)=-\frac{2}{3}s^6-xs^4-x^2s^2-\frac{3}{4}\ln s+\int\limits_x^{\infty}(y-x)u^2(y)dy +\omega +O\big(s^{-1}\big),
\end{equation}
recalling that $u(x)\sim \textnormal{Ai}(x)$ as $x\rightarrow+\infty$. As we see \eqref{asywithoutconstant} matches \eqref{theo1} up to a general, universal constant term $\omega$. We shall now determine this constant term using an approximation argument for the given kernel \eqref{PIIkernel}.
%---------------------------------------------------------------------------------------------------------------------------------------------------------------------------------------------------------

\section{Kernel approximation - $K_{\textnormal{PII}}\mapsto K_{\textnormal{csin}}$}\label{sec15}
Within the asymptotical analysis of the $X$-RHP in the past sections, one of the first transformations of the RHP was the $g$-function transformation: it allowed us to transform the jumps on the infinite branches $\Gamma_k$ to exponentially small contributions and to solve the model problem on the line segment $[-1,1]$. Then and there it was crucial that for $x$ chosen from a compact subset of the real line and $s$ sufficiently large, one always has that $\textnormal{Re}\ g(z)$ is negative on the infinite parts $\Gamma_1,\Gamma_3$ in the upper halfplane and positive on the infinite contours $\Gamma_4,\Gamma_6$ in the lower halfplane. This fact however also holds in the limit $x\rightarrow+\infty$, on the other hand it fails for $x\rightarrow-\infty$: Let
\begin{equation*}
	 z_{\pm} = \pm i\sqrt{\frac{1}{2}+\frac{3x}{4s^2}}
\end{equation*}
denote the two vertices of the curves depicted in Figure \ref{fig33} below. In case $x,s>0$, they are purely imaginary and bounded away from zero, hence the statement on the sign of $\textnormal{Re}\ g(z)$ on $\Gamma_i$ follows.
\begin{figure}[tbh]
  \begin{center}
  \psfragscanon
  \psfrag{1}{\footnotesize{$\textnormal{Re}\ g<0$}}
  \psfrag{2}{\footnotesize{$\textnormal{Re}\ g<0$}}
  \psfrag{3}{\footnotesize{$\textnormal{Re}\ g>0$}}
  \psfrag{4}{\footnotesize{$\textnormal{Re}\ g>0$}}
  \psfrag{5}{\footnotesize{$\textnormal{Re}\ g<0$}}
  \psfrag{6}{\footnotesize{$\textnormal{Re}\ g<0$}}
  \psfrag{7}{\footnotesize{$\textnormal{Re}\ g>0$}}
  \psfrag{8}{\footnotesize{$\textnormal{Re}\ g>0$}}
  \psfrag{9}{\footnotesize{$\textnormal{Re}\ g>0$}}
  \psfrag{10}{\footnotesize{$\textnormal{Re}\ g<0$}}
  \psfrag{11}{\footnotesize{$\textnormal{Re}\ g<0$}}
  \psfrag{12}{\footnotesize{$\textnormal{Re}\ g>0$}}
  \psfrag{13}{\footnotesize{$z_+$}}
  \psfrag{14}{\footnotesize{$z_-$}}
  \psfrag{15}{\footnotesize{$z_-$}}
  \psfrag{16}{\footnotesize{$z_+$}}
  \includegraphics[width=10cm,height=6cm]{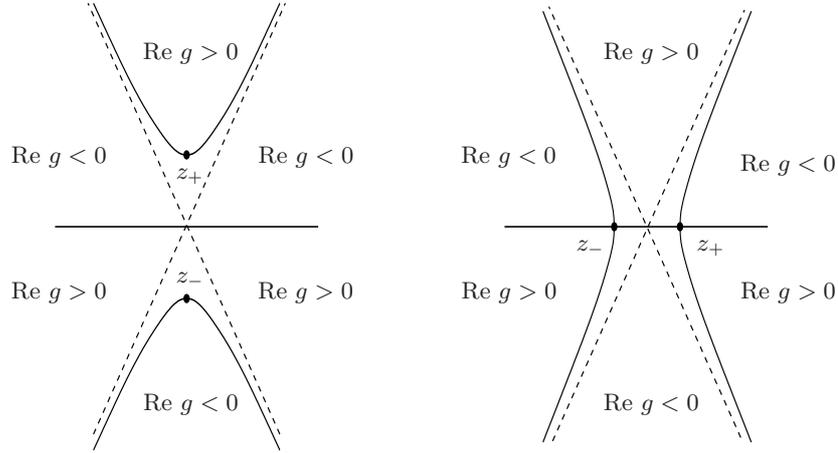}
  \end{center}
  \caption{Sign diagram for the function $\textnormal{Re}\ g(z)$. In the left picture we indicate the location of $z_{\pm}$ as $x>0$ and in the right picture for a particular choice of $x<0$. Along the solid lines $\textnormal{Re}\ g(z)=0$ and the dashed lines resemble $\textnormal{arg}\ z =\pm\frac{\pi}{3},\pm\frac{2\pi}{3}$}
  \label{fig33}
\end{figure}
\bigskip

Our approach henceforth will be to study the large positive $x$-limit of \eqref{PIIkernel}, i.e. the large positive $x$-limit of the associated function $\Psi(\lambda,x)$. We begin with the following Riemann-Hilbert problem depicted in Figure \ref{fig10}, compare \cite{FIKN}
\begin{figure}[tbh]
  \begin{center}
  \psfragscanon
  \psfrag{1}{\footnotesize{$S_1$}}
  \psfrag{2}{\footnotesize{$S_3$}}
  \psfrag{3}{\footnotesize{$S_4$}}
  \psfrag{4}{\footnotesize{$S_6$}}
  \includegraphics[width=7cm,height=3.5cm]{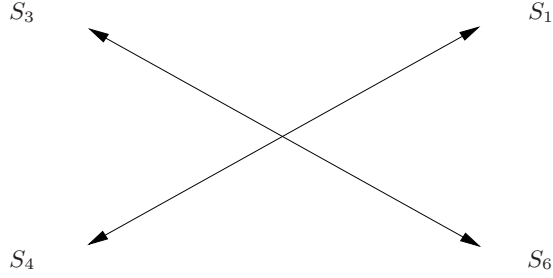}
  \end{center}
  \caption{The RHP jump graph associated with the Hastings-McLeod transcendent}
  \label{fig10}
\end{figure}
\begin{itemize}
	\item $\Psi^{\infty}(\lambda)$ is analytic for $\lambda\in\mathbb{C}\backslash\big(\bigcup_k R_k\big)$ where $R_k$ denote the rays
	\begin{equation*}
		R_k=\{\lambda\in\mathbb{C}|\ \textnormal{arg}\ \lambda=\frac{\pi}{6}+\frac{\pi}{3}(k-1)\},\ k=1,3,4,6
	\end{equation*}
	\item On the rays $R_k$, the boundary values of the the function $\Psi^{\infty}$ satisfy the jump relation
	\begin{equation*}
		\Psi^{\infty}_+(\lambda)=\Psi^{\infty}_-(\lambda)S_k,\ \ \lambda\in R_k,\ k=1,3,4,6
	\end{equation*}
	\item At $\lambda=\infty$ the following asymptotic behavior takes place
	\begin{equation*}
		\Psi^{\infty}(\lambda)e^{i(\frac{4}{3}\lambda^3+x\lambda)\sigma_3}=I+O\big(\lambda^{-1}\big)
	\end{equation*}
\end{itemize} 
which is connected to the given $\Psi$-function of \eqref{PIIkernel} by
\begin{equation*}
	\Psi(\lambda,x)=\Psi^{\infty}(\lambda,x)S_1.
\end{equation*}
As we see, determining the large positive $x$ behavior of $\Psi(\lambda,x)$ therefore reduces to an analysis of the oscillatory $\Psi^{\infty}$-RHP. However the latter RHP is very well known since it is used to determine the large $x$-asymptotics of the Hastings-McLeod solution of the second Painlev\'e transcendent given in the introduction (cf. \cite{FIKN}). We have in fact for $\lambda\in(-s,s)$
\begin{equation*}
	\Psi^{\infty}(\lambda,x)e^{i(\frac{4}{3}\lambda^3+x\lambda)\sigma_3}-I=
	O\bigg(\frac{x^{-1/4}e^{-\frac{2}{3}x^{3/2}}}{\sqrt{4\lambda^2+x}}\bigg),\hspace{0.5cm}x\rightarrow+\infty
\end{equation*}
hence
\begin{eqnarray*}
	\psi_{11}(\lambda,x)&=&e^{-i(\frac{4}{3}\lambda^3+x\lambda)}+O\bigg(\frac{x^{-1/4}e^{-\frac{2}{3}x^{3/2}}}{\sqrt{4\lambda^2+x}}\bigg)\\
	\psi_{21}(\lambda,x)&=&-ie^{i(\frac{4}{3}\lambda^3+x\lambda)}+O\bigg(\frac{x^{-1/4}e^{-\frac{2}{3}x^{3/2}}}{\sqrt{4\lambda^2+x}}\bigg)
\end{eqnarray*}
as $x\rightarrow+\infty$ and $\lambda\in(-s,s)$.
%\begin{equation*}
%	\Psi^{\infty}(\lambda,x) = \bigg[I+\frac{i}{4\sqrt{2\pi}}\begin{pmatrix}
%		0 & \big(\frac{\lambda}{\sqrt{x}}+\frac{i}{2}\big)^{-1} \\
%		\big(\frac{\lambda}{\sqrt{x}}-\frac{i}{2}\big)^{-1} & 0\\
%		\end{pmatrix}\frac{e^{-\frac{2}{3}x^{3/2}}}{x^{3/4}}+O\Big(x^{-3/4}e^{-\frac{4}{3}x^{3/2}}\Big)\bigg]e^{-(\frac{4}{3}i\lambda^3+ix\lambda)\sigma_3}
%end{equation*}
%as $x\rightarrow+\infty$ and $|\frac{\lambda}{\sqrt{x}}\pm\frac{i}{2}|>\varepsilon>0$. Hence for such $\lambda$
%\begin{equation*}
%	\psi_{11}(\lambda,x) = e^{-(\frac{4}{3}i\lambda^3+ix\lambda)}+O_{\varepsilon}\Big(x^{-3/4}e^{-\frac{2}{3}x^{3/2}-x\textnormal{Im}\ \lambda}\Big),\ \ x\rightarrow+\infty
%\end{equation*}
%and
%\begin{equation*}
%	\psi_{21}(\lambda,x) = -ie^{(\frac{4}{3}i\lambda^3+ix\lambda)}+O_{\varepsilon}\Big(x^{-3/4}e^{-\frac{2}{3}x^{3/2}+x\textnormal{Im}\ \lambda}\Big),\ \ x\rightarrow+\infty.
%\end{equation*}
Going back to \eqref{PIIkernel} we obtain
\begin{equation}\label{kernelapprox}
	K_{\textnormal{PII}}(\lambda,\mu) = \check{K}_{\textnormal{csin}}(\lambda,\mu)\Bigg(1+O\bigg(\frac{x^{-1/4}e^{-\frac{2}{3}x^{3/2}}}{\sqrt{(4\lambda^2+x)(4\mu^2+x)}}\bigg)\Bigg),\ x\rightarrow+\infty,\ \ \lambda,\mu\in(-s,s)
\end{equation}
%\begin{eqnarray}\label{kernelapprox}
%	K_{\textnormal{PII}}(\lambda,\mu) &=& \frac{e^{\frac{4}{3}i(\lambda^3-\mu^3)+ix(\lambda-\mu)}-e^{-(\frac{4}{3}i(\lambda^3-\mu^3)+ix(\lambda-\mu))}}{2\pi %i(\lambda-\mu)}\Bigg(1+O\bigg(\frac{x^{-1/4}e^{-\frac{2}{3}x^{3/2}}}{\sqrt{(4\lambda^2+x)(4\mu^2+x)}}\bigg)\Bigg)\nonumber\\
%	%&=&\check{K}_s^{\infty}(\lambda,\mu)\Bigg(1+O\bigg(\frac{x^{-1/4}e^{-\frac{2}{3}x^{3/2}}}{\sqrt{(4\lambda^2+x)(4\mu^2+x)}}\bigg)\Bigg),\hspace{0.5cm}x\rightarrow\%infty,\ \ \lambda,\mu\in(-s,s)
%\end{eqnarray}
%\begin{eqnarray}
%	K_s(\lambda,\mu) &=& \frac{e^{\frac{4}{3}i(\lambda^3-\mu^3)+ix(\lambda-\mu)}-e^{-(\frac{4}{3}i(\lambda^3-\mu^3)+ix(\lambda-\mu))}}{2\pi i(\lambda-\mu)}\bigg(1 + %O_{\varepsilon}\Big(x^{-3/4}e^{-\frac{2}{3}x^{3/2}}\Big)\bigg)\nonumber\\
%	&=&\check{K}_s^{\infty}(\lambda,\mu)\bigg(1+O_{\varepsilon}\Big(x^{-3/4}e^{-\frac{2}{3}x^{3/2}}\Big)\bigg),\ \ x\rightarrow+\infty
%\end{eqnarray}
%uniformly on compact subsets of $\mathbb{C}$ such that $(\lambda,\mu)\in\mathbb{C}^2$ lies outside the set $(\lambda_n,\mu_n)$ satisfying
%\begin{equation*}
%	\frac{4}{3}(\lambda_n^3-\mu_n^3)+x(\lambda_n-\mu_n)=n\pi,\ \ n\in\mathbb{Z},
%\end{equation*}
where
\begin{equation}\label{lambdacube}
	\check{K}_{\textnormal{csin}}(\lambda,\mu)= \frac{\sin\big(\frac{4}{3}(\lambda^3-\mu^3)+x(\lambda-\mu)\big)}{\pi(\lambda-\mu)}.
\end{equation}
The latter integral kernel is a cubic generalization of the well known sine kernel
\begin{equation*}
	\frac{\sin x(\lambda-\mu)}{\pi(\lambda-\mu)}
\end{equation*}
acting on $L^2\big((-s,s);d\lambda\big)$. In order to compute the constant term in \eqref{theo1} we will introduce a parameter $t\in[0,1]$ and pass from \eqref{lambdacube} to
\begin{equation}\label{lambdacubet}
	\check{K}_{\textnormal{csin}}(\lambda,\mu)\mapsto K_{\textnormal{csin}}(\lambda,\mu) = \frac{\sin\big(\frac{4}{3}t(\lambda^3-\mu^3)+x(\lambda-\mu)\big)}{\pi(\lambda-\mu)}
\end{equation}
Our strategy is as follows. First we are going to find the large $s$-asymptotics of
\begin{equation}\label{tlogderivative}
	\frac{d}{dt}\ln\det(I-K_{\textnormal{csin}})
\end{equation}
using the Riemann-Hilbert approach of section $2$. This analysis can be done independently of the original kernel \eqref{PIIkernel}, although in some details it is similar to the presented one.  Afterwards, using uniformity of the asymptotic expansion with respect to $t\in[0,1]$ we shall integrate 
\begin{equation*}
	\int\limits_0^1\frac{d}{dt}\ln\det(I-K_{\textnormal{csin}})\ dt = \ln\det(I-\check{K}_{\textnormal{csin}})-\ln\det(I-K_{\sin});
\end{equation*}
but since the asymptotic expansion of the sine kernel as $s\rightarrow\infty$ is known including the constant term, we know the large $s$-asymptotics of
\begin{equation*}
	\det(I-\check{K}_{\textnormal{csin}})
\end{equation*}
also up to order $O(s^{-1})$, in fact
\begin{equation}\label{approxesti1}
	\ln\det(I-\check{K}_{\textnormal{csin}}) = A(s,x)+\omega_0 +O\big(s^{-1}\big),\ \ s\rightarrow\infty 
\end{equation}
uniformly on any compact subset of the set \eqref{excset1} with 
\begin{equation*}
	A(s,x)=-\frac{2}{3}s^6-s^4x-\frac{1}{2}(sx)^2-\frac{3}{4}\ln s,\hspace{0.5cm} \omega_0=-\frac{1}{6}\ln 2+3\zeta'(-1),
\end{equation*}
which is the statement of Theorem \ref{theo2}. On the other hand from \eqref{asywithoutconstant}
\begin{equation}\label{approxesti2}
	\ln\det(I-K_{\textnormal{PII}})=A(s,x)+\int\limits_x^{\infty}(y-x)u^2(y)dy +\omega +O\big(s^{-1}),\ \ s\rightarrow\infty
\end{equation}
hence considering \eqref{kernelapprox} as well as
\begin{equation}\label{PIIintegralesti}
	\lim_{x\rightarrow\infty}\int\limits_x^{\infty}(y-x)u^2(y)=0
\end{equation}
we might conjecture that $\omega=\omega_0$, and this conclusion can be justified as follows. Notice the following identity for trace class operators $A,B$ (see \cite{S})
\begin{equation*}
	\det(I-A)(I-B) = \det(I-A)\det(I-B)
\end{equation*}
which gives in our situation
\begin{equation*}
	\det(I-K_{\textnormal{PII}})-\det(I-\check{K}_{\textnormal{csin}})=-\det(I-\check{K}_{\textnormal{csin}})
	\Big[1-\det\Big(I-(I-\check{K}_{\textnormal{csin}})^{-1}(K_{\textnormal{PII}}-\check{K}_{\textnormal{csin}})\Big)\Big]
\end{equation*}
provided
\begin{equation}\label{resolventcsin}
	(I-\check{K}_{\textnormal{csin}})^{-1} = I+\check{R}_{\textnormal{csin}}
\end{equation}
exists as a bounded operator. The latter statement will follow from the Riemann-Hilbert analysis given in section \ref{sec23}. Since from \eqref{kernelapprox}
\begin{equation*}
	(K_{\textnormal{PII}}f)(\lambda)=\int\limits_{-s}^sK_{\textnormal{PII}}(\lambda,\mu)f(\mu)d\mu = (\check{K}_{\textnormal{csin}}f)(\lambda)+(Ef)(\lambda)
\end{equation*}
where the trace class operator $E$ has a kernel satisfying
\begin{equation}\label{kernelapprox2}
	E(\lambda,\mu)=O\Big(x^{-1/4}e^{-\frac{2}{3}x^{3/2}}\Big),\ \ x\rightarrow\infty,\ \ (\lambda,\mu)\in[-s,s]\times[-s,s]
\end{equation}
we obtain
\begin{equation*}
	\det(I-K_{\textnormal{PII}})-\det(I-\check{K}_{\textnormal{csin}})=-\det(I-\check{K}_{\textnormal{csin}})
	\Big[1-\det\Big(I-(I+\check{R}_{\textnormal{csin}})E\Big)\Big]
\end{equation*}
and therefore
\begin{equation*}
	\frac{\det(I-K_{\textnormal{PII}})}{\det(I-\check{K}_{\textnormal{csin}})}=\det\Big(I-(I+\check{R}_{\textnormal{csin}})E\Big).
\end{equation*}
Now from the boundedness of $I+\check{R}_{\textnormal{csin}}$ as well as \eqref{kernelapprox2} we see that the convolution kernel of the operator
\begin{equation*}
	(I+\check{R}_{\textnormal{csin}})E
\end{equation*}
approaches zero exponentially fast as $x\rightarrow\infty$, thus via Hadamard's inequality
\begin{eqnarray*}
	&&\det\Big(I-(I+\check{R}_{\textnormal{csin}})E\Big)\\
	&=&1+\sum_{n=1}^{\infty}\frac{(-1)^n}{n!}
	\int\limits_{-s}^s\cdots\int\limits_{-s}^s\det\big[\big(I+\check{R}_{\textnormal{csin}})E\big](x_i,x_j)dx_1\cdots dx_n
	=1+o_s(1),\ x\rightarrow\infty
\end{eqnarray*}
or similarly
\begin{equation}\label{approxesti3}
	\ln\det(I-K_{\textnormal{PII}})=\ln\det(I-\check{K}_{\textnormal{csin}})+o_s(1),\ \ x\rightarrow\infty.
\end{equation}
We combine \eqref{approxesti1},\eqref{approxesti2}, \eqref{approxesti3} and obtain from Schwarz inequality
\begin{equation*}
	|\omega_0 -\omega|\leq \frac{\alpha}{s}+\frac{\beta(s)}{x}+\int\limits_x^{\infty}(y-x)u^2(y)dy
\end{equation*}
for all $x\geq x_0$ and $s\geq s_0$, with a universal constant $\alpha$ and a positive function $\beta=\beta(s)$. Recalling \eqref{PIIintegralesti} we first take the limit $x\rightarrow\infty$ and afterwards $s\rightarrow\infty$ to conclude $\omega_0=\omega$. 

%--------------------------------------------------------------------------------------------------------------------------------------------------------------------------------------

\section{Riemann-Hilbert problem for $\det(I-K_{\textnormal{csin}})$}\label{sec16}

Before we start proving \eqref{approxesti1}, let us locate \eqref{lambdacubet} within the framework of integrable Fredholm operators and derive the connection of \eqref{tlogderivative} with the solution of the underlying Riemann-Hilbert problem.
\smallskip

The given kernel \eqref{lambdacubet} is of integrable type with
\begin{equation*}
	K_{\textnormal{csin}}(\lambda,\mu) = \frac{d^t(\lambda)e(\mu)}{\lambda-\mu},\ d(\lambda)=\frac{1}{\sqrt{2\pi i}}\binom{e^{i(\frac{4}{3}t\lambda^3+x\lambda)}}{e^{-i(\frac{4}{3}t\lambda^3+x\lambda)}},\  e(\lambda)=\frac{1}{\sqrt{2\pi i}}\binom{e^{-i(\frac{4}{3}t\lambda^3+x\lambda)}}{-e^{i(\frac{4}{3}t\lambda^3+x\lambda)}}
\end{equation*}
hence Lemma $2$ implies the following $\Theta$-RHP 
\begin{itemize}
	\item $\Theta(\lambda)$ is analytic for $\lambda\in\mathbb{C}\backslash[-s,s]$
	\item On the line segment $[-s,s]$ oriented from left to right, the following jump holds
	\begin{equation*}
		\Theta_+(\lambda) = \Theta_-(\lambda)\begin{pmatrix}
		0 & e^{2i(\frac{4}{3}t\lambda^3+x\lambda)}\\
		-e^{-2i(\frac{4}{3}t\lambda^3+x\lambda)} & 2 \\
		\end{pmatrix},\hspace{0.5cm} \lambda\in[-s,s]
	\end{equation*}
	\item $\Theta(\lambda)$ has at most logarithmic endpoint singularities at $\lambda=\pm s$
	\begin{equation*}
		\Theta(\lambda)=O\big(\ln(\lambda\mp s)\big),\ \ \lambda\rightarrow\pm s
	\end{equation*}
	\item $\Theta(\lambda)\rightarrow I$ as $\lambda\rightarrow\infty$.
\end{itemize}
Also here we can factorize the jump matrix
\begin{equation*}
	H(\lambda)=\begin{pmatrix}
		0 & e^{2i(\frac{4}{3}t\lambda^3+x\lambda)}\\
		-e^{-2i(\frac{4}{3}t\lambda^3+x\lambda)} & 2 \\
		\end{pmatrix} = e^{i(\frac{4}{3}t\lambda^3+x\lambda)\sigma_3}\begin{pmatrix}
		0 & 1 \\
		-1 & 2 \\
		\end{pmatrix}e^{-i(\frac{4}{3}t\lambda^3+x\lambda)\sigma_3}
\end{equation*}
and employ a first transformation.
%------------------------------------------------------------------------------------------------------------------------------------------------------------------

\section{First transformation of the $\Theta$-RHP}\label{sec17}

Introduce
\begin{equation}\label{Fromthetatophi}
	\Phi(\lambda)=\Theta(\lambda)e^{i(\frac{4}{3}t\lambda^3+x\lambda)\sigma_3},\ \ \lambda\in\mathbb{C}\backslash[-s,s]
\end{equation} 
and obtain the following RHP
\begin{itemize}
	\item $\Phi(\lambda)$ is analytic for $\lambda\in\mathbb{C}\backslash[-s,s]$
	\item On the segment $[-s,s]$
\begin{equation*}
	\Phi_+(\lambda)=\Phi_-(\lambda)\begin{pmatrix}
	0 & 1 \\
	-1 & 2\\
	\end{pmatrix},\hspace{0.5cm}\lambda\in[-s,s]
\end{equation*}
	\item $\Phi(\lambda)$ has at most logarithmic endpoint singularities at $\lambda=\pm s$
	\item We have the following asymptotics at infinity
\begin{equation*}
	\Phi(\lambda)\sim \Big(I+O\big(\lambda^{-1}\big)\Big)e^{i(\frac{4}{3}t\lambda^3+x\lambda)\sigma_3},\hspace{0.5cm}\lambda\rightarrow\infty.
\end{equation*}
\end{itemize}
As we are going to see in the next sections the latter RHP admits direct asymptotical analysis via the nonlinear steepest descent method. This analysis shows similarities to the analysis of the sine kernel determinant analysis (see \cite{DIZ}) and the analysis presented in the past sections of the current paper. However one major difference to \eqref{PIIkernel} is the absence of infinite jump contours in the given $\Phi$-RHP, hence we should not start our analysis from the $X$-RHP in section $2$ and use the previously discussed large $x$-approximation $\Psi^{\infty}(\lambda,x)$.
\smallskip

Before we start the asymptotical analysis, let us first express \eqref{tlogderivative} in terms of $\Theta(\lambda)$.

%----------------------------------------------------------------------------------------------------------------------------------------------------------------------------------------------------------------

\section{The logarithmic $t$-derivative of $\det(I-K_{\textnormal{csin}})$}\label{sec18}

The following derivation is a special situation of the general case given in \cite{KKMST}. We start from \eqref{traceform}
\begin{eqnarray*}				\frac{d}{dt}\ln\det(I-K_{\textnormal{csin}})&=&-\int\limits_{-s}^s\bigg(\big(I-K_{\textnormal{csin}}\big)^{-1}\frac{dK_{\textnormal{csin}}}{dt}\bigg)(\lambda,\lambda)d\lambda\\
&=&-\int\limits_{-s}^s\bigg(\big(I+R_{\textnormal{csin}}\big)\frac{dK_{\textnormal{csin}}}{dt}\bigg)(\lambda,\lambda)d\lambda
\end{eqnarray*}
with $(I-K_{\textnormal{csin}})^{-1}=I+R_{\textnormal{csin}}$ denoting the resolvent. In the given situation
\begin{equation*}
	\frac{dK_{\textnormal{csin}}}{dt}(\lambda,\mu)=\frac{4i}{3}\big(d^t(\lambda)\sigma_3e(\mu)\big)\frac{\lambda^3-\mu^3}{\lambda-\mu}
	=\frac{4i}{3}\big(d^t(\lambda)\sigma_3e(\mu)\big)(\lambda^2+\lambda\mu+\mu^2)
\end{equation*}
hence
\begin{equation}\label{KKMSTtheo1}
	\int\limits_{-s}^{s}\frac{dK_{\textnormal{csin}}}{dt}(\lambda,\lambda)d\lambda=\textnormal{trace}\Big\{4i\sigma_3\int\limits_{-s}^s
	\big(d(\lambda)e^t(\lambda)\big)\lambda^2d\lambda\Big\}.
\end{equation}
On the other hand Lemma $1$ allows us to write the second summand as
\begin{eqnarray}\label{KKMSTtheo2}
	&&\int\limits_{-s}^s\bigg(R_{\textnormal{csin}}\frac{dK_{\textnormal{csin}}}{dt}\bigg)(\lambda,\lambda)d\lambda = \int\limits_{-s}^s\int\limits_{-s}^sR_{\textnormal{csin}}(\lambda,\mu)\frac{dK_{\textnormal{csin}}}{dt}(\mu,\lambda)d\mu d\lambda\nonumber\\
	&=&\int\limits_{-s}^s\int\limits_{-s}^s\frac{D^t(\lambda)E(\mu)}{\lambda-\mu}d^t(\mu)\sigma_3e(\lambda)\frac{4i}{3}(\lambda^2+\lambda\mu+\mu^2)d\mu d\lambda\nonumber\\
	&=&\textnormal{trace}\Bigg[\frac{4}{3}\sigma_3\int\limits_{-s}^s\int\limits_{-s}^s\frac{\big(d(\mu)E^t(\mu)\big)\big(D(\lambda)e^t(\lambda)\big)}{\lambda-\mu}i(\lambda^2+\lambda\mu+\mu^2)d\mu d\lambda\Bigg].
\end{eqnarray}
At this point the following observation is crucial
\begin{equation*}
	i(\lambda^2+\lambda\mu+\mu^2)=\frac{1}{2\pi}\int\limits_{\Sigma}\frac{w^3}{(w-\lambda)(w-\mu)}dw
\end{equation*}
with $\Sigma$ denoting a closed Jordan curve around the line segment $[-s,s]$. This identity combined with \eqref{IIKStheo5} rewrites \eqref{KKMSTtheo2} as
\begin{eqnarray*}		
	&&\int\limits_{-s}^s\bigg(R_{\textnormal{csin}}\frac{dK_{\textnormal{csin}}}{dt}\bigg)(\lambda,\lambda)d\lambda=\textnormal{trace}\Bigg[\frac{\sigma_3}{2\pi}\int\limits_{\Sigma}\int\limits_{-s}^s
	\int\limits_{-s}^s\frac{\big(d(\mu)E^t(\mu)\big)\big(D(\lambda)e^t(\lambda)\big)\frac{4}{3}w^3}{(\lambda-\mu)(w-\lambda)(w-\mu)}d\lambda d\mu dw\Bigg]\nonumber\\
	&=&\textnormal{trace}\Bigg[\frac{\sigma_3}{2\pi}\int\limits_{\Sigma}\int\limits_{-s}^s\frac{d(\mu)E^t(\mu)}{(w-\mu)^2}\frac{4}{3}w^3\bigg[\int\limits_{-s}^s\frac{D(\lambda)e^t(\lambda)}{\lambda-\mu}
	d\lambda-\int\limits_{-s}^s\frac{D(\lambda)e^t(\lambda)}{\lambda-w}d\lambda\bigg]d\mu dw\Bigg]\nonumber\\
\end{eqnarray*}
\begin{eqnarray}\label{KKMSTtheo3}
	&=&\textnormal{trace}\Bigg[\frac{\sigma_3}{2\pi}\int\limits_{\Sigma}\int\limits_{-s}^s\frac{d(\mu)E^t(\mu)}{(w-\mu)^2}\frac{4}{3}w^3\big(\Theta(w)-\Theta(\mu)\big)d\mu dw\Bigg].
\end{eqnarray}
Finally, as a consequence of \eqref{IIKStheo3} and \eqref{IIKStheo4}, we note
\begin{equation*}
	\Theta^{-1}(z)=I+\int\limits_{\Gamma}d(w)E^t(w)\frac{dw}{w-z}
\end{equation*}
thus \eqref{KKMSTtheo3} yields
\begin{eqnarray*}
	\int\limits_{-s}^s\bigg(R_{\textnormal{csin}}\frac{dK_{\textnormal{csin}}}{dt}\bigg)(\lambda,\lambda)d\lambda&=& \textnormal{trace}\Bigg[\frac{\sigma_3}{2\pi}\int\limits_{\Sigma}\big(\Theta^{-1}(w)\big)'\Theta(w)\frac{4}{3}w^3dw\\
	&& -4i\sigma_3\int\limits_{-s}^sd(\mu)e^t(\mu)\mu^2d\mu\Bigg]
\end{eqnarray*}
and adding the latter equality and \eqref{KKMSTtheo1} we derived
\begin{equation*}
	\frac{d}{dt}\ln\det(I-K_{\textnormal{csin}})=\frac{1}{2\pi}\int\limits_{\Sigma}\textnormal{trace}\Big[\Theta'(w)\sigma_3\Theta^{-1}(w)\Big]\frac{4}{3}w^3dw
\end{equation*}
showing explicitly the connection to the $\Theta$-RHP an which infact can be evaluated using residue theorem. We summarize
\begin{proposition}\label{prop4}
The logarithmic $t$-derivative of the Fredholm determinant $\det(I-K_{\textnormal{csin}})$ can be expressed as
\begin{equation}\label{tidentity}
	\frac{d}{dt}\ln\det(I-K_{\textnormal{csin}})=\frac{4i}{3}\textnormal{trace}\Big(-\Theta_1\sigma_3\big(\Theta_1^2-\Theta_2\big)
	+2\Theta_2\sigma_3\Theta_1-3\Theta_3\sigma_3\Big)
\end{equation}
with
\begin{equation*}
	\Theta(\lambda)\sim I+\frac{\Theta_1}{\lambda}+\frac{\Theta_2}{\lambda^2}+\frac{\Theta_3}{\lambda^3}+O\big(\lambda^{-4}\big),\hspace{0.5cm}\lambda\rightarrow\infty.
\end{equation*}
\end{proposition}
Our next task is to determine the large $s$-asymptotics of the stated matrix coefficients $\Theta_1,\Theta_2$ and $\Theta_3$, i.e. we need to asymptotically solve the $\Phi$-RHP. Also here, a $g$-function transformation is essential.
%---------------------------------------------------------------------------------------------------------------------------------------------------------------------------------------------------------

\section{Second transformation of the $\Theta$-RHP - rescaling and $g$-function}\label{sec19}
%Generally speaken the following subsections are very similar to the approach presented in section $3$. We use again local parametrices to approximate the global %solution of the given Riemann-Hilbert problem, however certain techniqual details in the analysis are more involved than in the steps previously taken. Hence, %again, we shall proceed step by step.
%---------------------------------------------------------------------------------------------------
The jump contour of the $\Phi$-RHP consists only of the line segment $[-s,s]$ oriented from left to right
\begin{equation*}
	\Phi_+(\lambda)=\Phi_-(\lambda)\begin{pmatrix}
                                   0 & 1 \\
                                   -1 & 2 \\
                                 \end{pmatrix}
\end{equation*}
and this jump equals the jump one faces during the asymptotical analysis of the sine kernel determinant (cf. \cite{DIZ}). Here and there we use a $g$-function together with the scaling $z=\frac{\lambda}{s}$. Introduce
\begin{equation}\label{gtfunction}
	\hat{g}(z)=\frac{4i}{3}\sqrt{z^2-1}\bigg(z^2t+\frac{t}{2}+\frac{3x}{4s^2}\bigg),\hspace{0.5cm} \sqrt{z^2-1}\sim z,\ \ z\rightarrow\infty
\end{equation}
being analytic outside the segment $[-1,1]$ and as $z\rightarrow\infty$
\begin{equation*}
	\hat{g}(z)=i\bigg(\frac{4}{3}tz^3+\frac{xz}{s^2}\bigg)+O\big(z^{-1}\big).
\end{equation*}
In the situation $t=1$, \eqref{gtfunction} reduces to the previously used $g$-function \eqref{gfunction}, whereas for $t=0$, we obtain the $g$-function used in the analysis of the sine kernel (see \cite{DIZ}). We put
\begin{equation}\label{ghatfunction}
	S(z) = \Phi(zs)e^{-s^3\hat{g}(z)\sigma_3},\ \ z\in\mathbb{C}\backslash[-1,1]
\end{equation}
and are lead to the following RHP
\begin{itemize}
	\item $S(z)$ is analytic for $z\in\mathbb{C}\backslash[-1,1]$
	\item The following jump holds
	\begin{equation*}
		S_+(z) = S_-(z)\begin{pmatrix}
		0 & 1\\
		-1 & 2e^{2s^3\hat{g}_+(z)}\\
		\end{pmatrix},\ \ z\in(-1,1)
	\end{equation*}
	since 
	\begin{equation*}
		\hat{g}_+(z)+\hat{g}_-(z)=0,z\in[-1,1].
	\end{equation*}
	\item $S(z)$ has at most logarithmic singularities at the endpoints $z=\pm 1$
	\item As $z\rightarrow\infty$, $S(z)=I+O\big(z^{-1}\big)$.
\end{itemize}
Since $\textnormal{Im}\ \sqrt{z^2-1}_+>0$ for $z\in(-1,1)$, we have $\textnormal{Re}\ \hat{g}_+(z)<0$ for $z\in(-1,1)$ showing that
\begin{equation*}
	\begin{pmatrix}
		0 & 1 \\
		-1 & 2e^{2s^3\hat{g}_+(z)} \\
		\end{pmatrix}\longrightarrow \begin{pmatrix}
			0 & 1 \\
			-1 & 0 \\
			\end{pmatrix}, \ \ s\rightarrow\infty,\ \ z\in(-1,1)
\end{equation*}
exponentially fast. Thus also here we expect, that as $s\rightarrow\infty$, $S(z)$ converges to a solution of a model RHP posed on the line segment $[-1,1]$.

%------------------------------------------------------------------------------------------------------------------------------------------------------------------------------------------------

\section{The model problem in the $S$-RHP}\label{sec20}
% The given $g$-function and scaling transform the original problem to a RHP with identical jumps
%\begin{equation}\label{gtjump}
%	\Phi_+(z) = \Phi_-(z)\begin{pmatrix}
%	0 & 1 \\
%	-1 & 2\\
%	\end{pmatrix},\hspace{0.5cm}z\in[-1,1]
%\end{equation}
%a rescaled singular endpoint behavior
%\begin{equation}\label{gtendpoint}
%	\Phi(z)=O\big(\ln(z\mp 1)\big),\ \ z\rightarrow\pm 1
%\end{equation}
%and asymptotics
%\begin{equation*}
%	\Phi(z)=\Big(I+O\big(z^{-1}\big)\Big)e^{s^3\hat{g}(z)\sigma_3},\ \ z\rightarrow\infty.
%\end{equation*}
%Opposed to the situation in section $3$ the given $g$-function does not allow us to solve the problem on $[-1,1]$ explicitly. However since
%\begin{equation*}
%	\hat{g}_+(z)+\hat{g}_-(z)=0,\ \ z\in[-1,1]
%\end{equation*}
%we have for $\Pi(z)=\Phi(z)e^{-s^3\hat{g}(z)\sigma_3}$
%\begin{equation*}
%	\Pi_+(z)=\Pi_-(z)\begin{pmatrix}
%		0 & 1\\
%		-1 & 2e^{2s^3\hat{g}_+(z)} \\
%	\end{pmatrix},\hspace{0.5cm}z\in[-1,1]
%\end{equation*}
%and $\textnormal{Im}\ \sqrt{z^2-1}_+>0$ if $z\in(-1,1)$. Thus $\textnormal{Re}\ \hat{g}_+(z) <0$ for $z\in(-1,1)$ showing that 
%\begin{equation*}
%	\begin{pmatrix}
%		0 & 1 \\
%		-1 & 2e^{2s^3\hat{g}_+(z)} \\
%		\end{pmatrix}\longrightarrow \begin{pmatrix}
%			0 & 1 \\
%			-1 & 0 \\
%			\end{pmatrix}, \ \ s\rightarrow\infty,\ \ z\in(-1,1)
%\end{equation*}
Find the piecewise analytic $2\times 2$ matrix valued function $N(z)$ such that
\begin{itemize}
	\item $N(z)$ is analytic for $z\in[-1,1]$
	\item On the line segment $[-1,1]$ the following jump holds
	\begin{equation*}
		N_+(z)=N_-(z)\begin{pmatrix}
		0 & 1 \\
		-1 & 0 \\
		\end{pmatrix},\hspace{0.5cm} z\in[-1,1]
	\end{equation*}
	\item $N(z)$ has at most logarithmic singularities at the endpoints $z=\pm 1$
	\item $N(z)= I+O\big(z^{-1}\big),\ z\rightarrow\infty$
\end{itemize}
This problem has an explicit solution (compare section $8$)
\begin{equation}\label{PhiPmodel}
	N(z)=\begin{pmatrix}
						1 & 1 \\
						i & -i \\
						\end{pmatrix}\beta(z)^{-\sigma_3}\frac{1}{2}\begin{pmatrix}
						1 & -i \\
						1 & i \\
						\end{pmatrix},\hspace{1cm} \beta(z)=\bigg(\frac{z+1}{z-1}\bigg)^{1/4}
\end{equation}
and $\big(\frac{z+1}{z-1}\big)^{1/4}$ is defined on $\mathbb{C}\backslash[-1,1]$ with its branch fixed by the condition $\big(\frac{z+1}{z-1}\big)^{1/4}\rightarrow 1$ as $z\rightarrow\infty$.
%----------------------------------------------------------------------------------------------------

\section{Construction of edge point parametrices in the $S$-RHP}\label{sec21}
The construction of endpoint parametrices is very similar to the constructions given in sections \ref{sec9} and \ref{sec10}. We use again Bessel functions. First for the right endpoint $z=+1$ introduce on the punctured plane $\zeta\in\mathbb{C}\backslash \{0\}$
\begin{equation}\label{QBERHright}
	Q_{BE}^{RH}(\zeta)=\begin{pmatrix}
	\sqrt{\zeta}\big(H_0^{(1)}\big)'(\sqrt{\zeta}) & \sqrt{\zeta}\big(H_0^{(2)}\big)'(\sqrt{\zeta})\\
	H_0^{(1)}(\zeta) & H_0^{(2)}(\zeta) \\
	\end{pmatrix},\hspace{0.5cm} -\pi<\textnormal{arg}\ \zeta\leq \pi.
\end{equation}
Since $Q_{BE}^{RH}(\zeta)=
\sigma_1e^{-i\frac{\pi}{4}\sigma_3}P_{BE}(\zeta)e^{i\frac{\pi}{4}\sigma_3}\sigma_1$ we can use \eqref{PBERHasyinfinity} and deduce
\begin{eqnarray}\label{QBERHasyinfinity}
	Q_{BE}^{RH}(\zeta) &=& \sqrt{\frac{2}{\pi}}\zeta^{\sigma_3/4}e^{i\frac{\pi}{4}}\begin{pmatrix}
	1 & -i \\
	-i & 1 \\
	\end{pmatrix}\bigg[I+\frac{i}{8\sqrt{\zeta}}\begin{pmatrix}
	1 & 2i \\
	2i & -1\\
	\end{pmatrix} +\frac{3}{128\zeta}\begin{pmatrix}
	1 & -4i \\
	4i & 1 \\
	\end{pmatrix}\nonumber\\
	&& +\frac{15i}{1024\zeta^{3/2}}\begin{pmatrix}
	-1 & -6i\\
	-6i & 1\\
	\end{pmatrix} +O\big(\zeta^{-2}\big)\bigg]e^{i\sqrt{\zeta}\sigma_3}
\end{eqnarray}
as $\zeta\rightarrow\infty$, valid in a full neighborhood of infinity. Also on the line $\textnormal{arg}\ \zeta=\pi$ we obtain
\begin{equation*}
	\big(Q_{BE}^{RH}(\zeta)\big)_+=\big(Q_{BE}^{RH}(\zeta)\big)_-\begin{pmatrix}
	0 & 1 \\
	-1 & 2\\
	\end{pmatrix}
\end{equation*}
thus \eqref{QBERHright} solves the RHP depicted in Figure \ref{fig11}.
\begin{figure}[tbh]
  \begin{center}
  \psfragscanon
  \psfrag{1}{$\bigl(\begin{smallmatrix}
  0 & 1\\
  -1 & 2\\
  \end{smallmatrix}\bigr)$}
  \psfrag{2}{\footnotesize{$z=+1$}}
  \includegraphics[width=4cm,height=3cm]{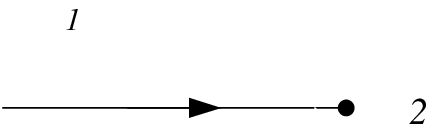}
  \end{center}
  \caption{The model RHP near $z=+1$ which can be solved in terms of Hankel functions}
  \label{fig11}
\end{figure}

We use the model function $Q_{BE}^{RH}(\zeta)$ in the construction of the parametrix to the solution of the original $\Phi$-RHP in a neighborhood of $z=+1$. First (compare \eqref{BEchangeright})
\begin{equation}\label{Qchangeright}
	\zeta(z)=-s^6\hat{g}^2(z),\hspace{0.5cm}|z-1|<r,\ \ -\pi<\textnormal{arg}\ \zeta\leq \pi
\end{equation}
with
\begin{equation*}
	\sqrt{\zeta(z)} = -is^3\hat{g}(z)=\frac{4s^3}{3}\sqrt{z^2-1}\bigg(z^2t+\frac{t}{2}+\frac{3x}{4s^2}\bigg)
\end{equation*}
which gives a locally conformal change of variables
\begin{equation*}
	\zeta(z)=\frac{32s^6}{9}\bigg(\frac{3t}{2}+\frac{3x}{4s^2}\bigg)^2(z-1)\big(1+O(z-1)\big),\ \ |z-1|<r.
\end{equation*}
Secondly define the right parametrix $S^r(z)$ near $z=+1$ by the formula
\begin{equation}\label{Phirparametrix}
	S^r(z)=C_r(z)\frac{1}{2}\begin{pmatrix}
	1 & 0 \\
	0 & i \\
	\end{pmatrix}\sqrt{\frac{\pi}{2}}e^{-i\frac{\pi}{4}}Q_{BE}^{RH}\big(\zeta(z)\big)e^{-s^3\hat{g}(z)\sigma_3},\hspace{0.5cm} |z-1|<r
\end{equation}
with $\zeta(z)$ as in \eqref{Qchangeright} and
\begin{equation*}
	C_r(z) = \begin{pmatrix}
	1 & 1 \\
	i & -i \\
	\end{pmatrix}\bigg(\zeta(z)\frac{z+1}{z-1}\bigg)^{-\sigma_3/4},\ \ C_r(1)=\begin{pmatrix}
	1 & 1 \\
	i & -i \\
	\end{pmatrix}\Bigg(\frac{8s^3}{3}\bigg(\frac{3t}{2}+\frac{3x}{4s^2}\bigg)\Bigg)^{-\sigma_3/2}.
\end{equation*}
As a result of our construction the parametrix has jumps only on the line segment depicted in Figure \ref{fig11}, described by the same jump matrix as in the original $S$-RHP. Also, since $Q_{BE}^{RH}(\zeta)=
\sigma_1e^{-i\frac{\pi}{4}\sigma_3}P_{BE}(\zeta)e^{i\frac{\pi}{4}\sigma_3}\sigma_1$, the singular endpoint behavior matches. Therefore the ratio of $S(z)$ with $S^r(z)$ is locally analytic, i.e.
\begin{equation*}
	S(z)=M_r(z)S^r(z),\hspace{0.5cm} |z-1|<r<\frac{1}{2}
\end{equation*}
and moreover from \eqref{QBERHasyinfinity}
\begin{eqnarray}\label{QBERHrightmatchup}
	S^r(z) &=& \begin{pmatrix}
	1 & 1 \\
	i & -i \\
	\end{pmatrix}\beta(z)^{-\sigma_3}\frac{1}{2}\begin{pmatrix}
	1 & -i\\
	1 & i \\
	\end{pmatrix}\bigg[I+\frac{i}{8\sqrt{\zeta}}\begin{pmatrix}
	1 & 2i \\
	2i & -1\\
	\end{pmatrix} +\frac{3}{128\zeta}\begin{pmatrix}
	1 & -4i \\
	4i & 1 \\
	\end{pmatrix}\nonumber\\
	&& +\frac{15i}{1024\zeta^{3/2}}\begin{pmatrix}
	-1 & -6i\\
	-6i & 1\\
	\end{pmatrix} +O\big(\zeta^{-2}\big)\bigg]\frac{1}{2}\begin{pmatrix}
	1 & 1 \\
	i & -i \\
	\end{pmatrix}\beta(z)^{\sigma_3}\begin{pmatrix}
	1 & -i \\
	1 & i \\
	\end{pmatrix}N(z)\nonumber\\
	&=&\bigg[I+\frac{i}{16\sqrt{\zeta}}\begin{pmatrix}
	3\beta^{-2}-\beta^2 & i(3\beta^{-2}+\beta^2)\\
	i(3\beta^{-2}+\beta^2) & -(3\beta^{-2}-\beta^2) \\
	\end{pmatrix}+\frac{3}{128\zeta}\begin{pmatrix}
	1 & -4i\\
	4i& 1\\
	\end{pmatrix}\nonumber\\
	&&+\frac{15i}{2048\zeta^{3/2}}\begin{pmatrix}
	5\beta^2-7\beta^{-2} & -i(5\beta^2+7\beta^{-2})\\
	-i(5\beta^2+7\beta^{-2}) & -(5\beta^2-7\beta^{-2}) \\
	\end{pmatrix}+O\big(\zeta^{-2}\big)\bigg]N(z)
\end{eqnarray}
as $s\rightarrow\infty$ and $0<r_1\leq|z-1|\leq r_2<1$ (so $|\zeta|\rightarrow\infty$). It is very important that the function $\zeta(z)$ is or order $O\big(s^2\big)$ on the latter annulus for all $t\in[0,1]$. Hence since $\beta(z)$ is bounded, equation \eqref{QBERHrightmatchup} yields the desired matching relation between the model functions $S^r(z)$ and $N(z)$,
\begin{equation*}
	S^r(z)=\big(I+o(1)\big)N(z),\hspace{0.5cm} s\rightarrow\infty,\ 0<r_1\leq|z-1|\leq r_2<1
\end{equation*}
uniformly on any compact subset of the set
\begin{equation}\label{exceptset2}
	\big\{ (t,x)\in\mathbb{R}^2:\ 0\leq t\leq 1, -\infty<x<\infty\big\}.
\end{equation}
\bigskip

We move on to the left endpoint $z=-s$. For $\zeta\in\mathbb{C}\backslash\{0\}$ define
\begin{equation}\label{QBERHleft}
	\tilde{Q}_{BE}^{RH}(\zeta)=\begin{pmatrix}
	H_0^{(2)}(e^{-i\frac{\pi}{2}}\sqrt{\zeta}) & H_0^{(1)}(e^{-i\frac{\pi}{2}}\sqrt{\zeta}) \\
	-e^{-i\frac{\pi}{2}}\sqrt{\zeta}\big(H_0^{(2)}\big)'(e^{-i\frac{\pi}{2}}\sqrt{\zeta}) & e^{-i\frac{3\pi}{2}}\sqrt{\zeta}\big(H_0^{(1)}\big)'(e^{-i\frac{\pi}{2}}\sqrt{\zeta}) \\
	\end{pmatrix},\ 0<\textnormal{arg}\ \zeta\leq 2\pi
\end{equation}
hence, since $\tilde{Q}_{BE}^{RH}(\zeta)=\sigma_1e^{-i\frac{\pi}{4}\sigma_3}\tilde{P}_{BE}(\zeta)e^{i\frac{\pi}{4}\sigma_3}\sigma_1$ we obtain from \eqref{PBERHasyleftinfinity}
\begin{eqnarray}\label{QBERHasyleftinfinity}
	\tilde{Q}_{BE}^{RH}(\zeta) &=& \sqrt{\frac{2}{\pi}}\zeta^{-\sigma_3/4}\begin{pmatrix}
	i & 1 \\
	i & -1\\
	\end{pmatrix}\bigg[I+\frac{1}{8\sqrt{\zeta}}\begin{pmatrix}
	1 & -2i \\
	-2i & -1\\
	\end{pmatrix}+\frac{3}{128\zeta}\begin{pmatrix}
	-1 & -4i\\
	4i & -1\\
	\end{pmatrix}\nonumber\\
	&&+\frac{15}{1024\zeta^{3/2}}\begin{pmatrix}
	1 & -6i\\
	-6i & -1\\
	\end{pmatrix} +O\big(\zeta^{-2}\big)\bigg]e^{-\sqrt{\zeta}\sigma_3}
\end{eqnarray}
as $\zeta\rightarrow\infty$ and on the line $\textnormal{arg}\ \zeta=2\pi$
\begin{equation*}
	\big(\tilde{Q}_{BE}^{RH}(\zeta)\big)_+=\big(\tilde{Q}_{BE}^{RH}(\zeta)\big)_-\begin{pmatrix}
	0 & 1\\
	-1 & 2\\
	\end{pmatrix}
\end{equation*}
which shows that \eqref{QBERHleft} solves the model problem of Figure \ref{fig12}.
\begin{figure}[tbh]
  \begin{center}
  \psfragscanon
  \psfrag{1}{$\bigl(\begin{smallmatrix}
  0 & 1\\
  -1 & 2\\
  \end{smallmatrix}\bigr)$}
  \psfrag{2}{\footnotesize{$z=-1$}}
  \includegraphics[width=4cm,height=3cm]{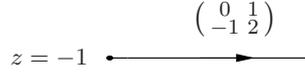}
  \end{center}
  \caption{The model RHP near $z=-1$ which can be solved in terms of Hankel functions}
  \label{fig12}
\end{figure}
This model problem enables us to introduce the parametrix $S^l(z)$ in a neighborhood of $z=-1$. Define
\begin{equation}\label{Qchangeleft}
	\zeta(z) = s^6\hat{g}^2(z),\hspace{0.5cm}|z+1|<r,\ \ 0<\textnormal{arg}\ \zeta\leq 2\pi
\end{equation}
with
\begin{equation*}
	\sqrt{\zeta(z)} = -s^3\hat{g}(z)=-\frac{4is^3}{3}\sqrt{z^2-1}\bigg(z^2t+\frac{t}{2}+\frac{3x}{4s^2}\bigg),
\end{equation*}
a locally conformal change of variables
\begin{equation*}
	\zeta(z) = \frac{32s^6}{9}\bigg(\frac{3t}{2}+\frac{3x}{4s^2}\bigg)^2(z+1)\big(1+O(z+1)\big),\hspace{0.5cm}|z+1|<r.
\end{equation*}
Given the left parametrix $S^l(z)$ near $z=-1$ by the formula
\begin{equation}\label{Philparametrix}
	S^l(z) = C_l(z)\bigg(-\frac{i}{2}\bigg)\sqrt{\frac{\pi}{2}}\tilde{Q}_{BE}^{RH}\big(\zeta(z)\big)e^{-s^3\hat{g}(z)\sigma_3},\hspace{0.5cm}|z+1|<r
\end{equation}
where
\begin{equation*}
	C_l(z) = \begin{pmatrix}
	1 & 1 \\
	i & -i\\
	\end{pmatrix}\bigg(\zeta(z)\frac{z-1}{z+1}\bigg)^{\sigma_3/4},\hspace{0.5cm} C_l(-1)=\begin{pmatrix}
	1 & 1\\
	i & -i\\
	\end{pmatrix}\Bigg(\frac{8is^3}{3}\bigg(\frac{3t}{2}+\frac{3x}{4s^2}\bigg)\Bigg)^{\sigma_3/2}
\end{equation*}
and $\zeta=\zeta(z)$ as in \eqref{Qchangeleft}, we see that the model jump matches the jump in the original $S$-RHP and by the symmetry relation $\tilde{Q}_{BE}^{RH}(\zeta)=\sigma_1e^{-i\frac{\pi}{4}\sigma_3}\tilde{P}_{BE}(\zeta)e^{i\frac{\pi}{4}\sigma_3}\sigma_1$ also 
\begin{equation*}
	S^l(z)=O\big(\ln(z+1)\big),\hspace{0.5cm}z\rightarrow -1.
\end{equation*}
Hence the ratio of $S(z)$ with $S^l(z)$ is locally analytic
\begin{equation*}
	S(z) = M_l(z)S^l(z),\hspace{0.5cm}|z+1|<r<\frac{1}{2}
\end{equation*}
and via \eqref{QBERHasyleftinfinity}
\begin{eqnarray}\label{QBERHleftmatchup}
	S^l(z) &=& \begin{pmatrix}
	1 & 1\\
	i & -i\\
	\end{pmatrix}\beta(z)^{-\sigma_3}\frac{1}{2}\begin{pmatrix}
	1 & -i\\
	1 & i\\
	\end{pmatrix}\bigg[I+\frac{1}{8\sqrt{\zeta}}\begin{pmatrix}
	1 & -2i \\
	-2i & -1\\
	\end{pmatrix}+\frac{3}{128\zeta}\begin{pmatrix}
	-1 & -4i\\
	4i & -1\\
	\end{pmatrix}\nonumber\\
	&&+\frac{15}{1024\zeta^{3/2}}\begin{pmatrix}
	1 & -6i\\
	-6i & -1\\
	\end{pmatrix} +O\big(\zeta^{-2}\big)\bigg]\frac{1}{2}\begin{pmatrix}
	1 & 1\\
	i & -i\\
	\end{pmatrix}\beta(z)^{\sigma_3}\begin{pmatrix}
	1 & -i\\
	1 & i\\
	\end{pmatrix}N(z)\nonumber\\
	&=&\bigg[I+\frac{1}{16\sqrt{\zeta}}\begin{pmatrix}
	3\beta^2-\beta^{-2} & -i(3\beta^2+\beta^{-2})\\
	-i(3\beta^2+\beta^{-2}) &-(3\beta^2-\beta^{-2})\\
	\end{pmatrix}+\frac{3}{128\zeta}\begin{pmatrix}
	-1 & -4i\\
	4i & -1\\
	\end{pmatrix}\nonumber\\
	&&+\frac{15}{2048\zeta^{3/2}}\begin{pmatrix}
	7\beta^2-5\beta^{-2}&-i(7\beta^2+5\beta^{-2})\\
	-i(7\beta^2+5\beta^{-2}) & -(7\beta^2-5\beta^{-2})\\
	\end{pmatrix}+O\big(\zeta^{-2}\big)\bigg]N(z)
\end{eqnarray}
as $s\rightarrow\infty$ and $0<r_1\leq|z+1|\leq r_2<1$ (hence $|\zeta|\rightarrow\infty$). Also here the function $\zeta(z)$ in \eqref{Qchangeleft} is of order $O\big(s^2\big)$ on the latter annulus for all $t\in[0,1]$. Therefore \eqref{QBERHleftmatchup} yields
\begin{equation*}
	S^l(z)=\big(I+o(1)\big)N(z),\hspace{0.5cm}s\rightarrow\infty,\ 0<r_1\leq|z+1|\leq r_2<1
\end{equation*}
again uniformly on any compact subset of the set \eqref{exceptset2}.
%-------------------------------------------------------------------------------------------------

\section{Third and final transformation of the $S$-RHP}\label{sec22}
Similar to \eqref{errorfunction} we define
\begin{equation}\label{errorfunctiont}
	\Lambda(z)=S(z)\left\{
                 \begin{array}{ll}
                   \big(S^r(z)\big)^{-1}, & \hbox{$|z-1|<\varepsilon$,} \\
                   \big(S^l(z)\big)^{-1}, & \hbox{$|z+1|<\varepsilon$,} \\
                   \big(N(z)\big)^{-1}, & \hbox{$|z\mp 1|>\varepsilon$} 
                 \end{array}
               \right.
\end{equation}
with $0<\varepsilon<\frac{1}{4}$ fixed and are lead to the ratio-RHP depicted in Figure \ref{fig13}\vspace{0.5cm}
\begin{figure}[tbh]
  \begin{center}
  \psfragscanon
  \psfrag{1}{\footnotesize{$\hat{C}_r$}}
  \psfrag{2}{\footnotesize{$\hat{C}_l$}}
  \includegraphics[width=5.5cm,height=1cm]{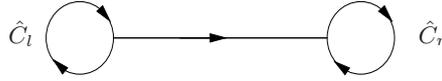}
  \end{center}\vspace{0.5cm}
  \caption{The jump graph for the ratio function $\Lambda(z)$}
  \label{fig13}
\end{figure}

More precisely, the function $\Lambda(z)$ has the following analytic properties.
\begin{itemize}
	\item $\Lambda(z)$ is analytic for $z\in\mathbb{C}\backslash\{(-1+\varepsilon,1-\varepsilon)\cup \hat{C}_r\cup \hat{C}_l\}$
	\item The following jumps are valid on the clockwise oriented circles
	\begin{equation*}
		\Lambda_+(z)=\Lambda_-(z)S^{r,l}(z)\big(N(z)\big)^{-1},\hspace{0.5cm}|z\mp 1|=\varepsilon
	\end{equation*} 
	whereas on the line segment $(-1+\varepsilon,1-\varepsilon)$
	\begin{equation*}
		\Lambda_+(z)=\Lambda_-(z)N_+(z)\begin{pmatrix}
		1 & -2\\
		0 & 1\\
		\end{pmatrix}\big(N_+(z)\big)^{-1}.
	\end{equation*}
	\item $\Lambda(z)$ is analytic at $z=\pm 1$
	\item As $z\rightarrow\infty$ we have $\Lambda(z)\rightarrow I$.
\end{itemize}
As a result of our construction \eqref{errorfunctiont}, $\Lambda(z)$ has no jumps in the parts of the original jump contour which lie inside the circles $\hat{C}_{r,l}$ and as we shall see now, the latter $\Lambda$-RHP admits direct asymptotical analysis. To this end recall the matching relations \eqref{QBERHrightmatchup}, \eqref{QBERHleftmatchup} and deduce
\begin{equation}\label{DZesti1}
	\|S^{r,l}\big(N\big)^{-1}-I\|_{L^2\cap L^{\infty}(\hat{C}_{r,l})}\leq c_4s^{-1},\hspace{0.5cm}s\rightarrow\infty
\end{equation} 
which holds uniformly on any compact subset of the set \eqref{exceptset2} with some constant $c_4>0$. Secondly recall \eqref{PhiPmodel}
\begin{eqnarray*}
	N_+(z)\begin{pmatrix}
		1 & -2\\
		0 & 1\\
		\end{pmatrix}\big(N_+(z)\big)^{-1}&=&\begin{pmatrix}
		\beta_++\beta_+^{-1} & i(\beta_+-\beta_+^{-1})\\
		-i(\beta_+-\beta_+^{-1}) & \beta_++\beta_+^{-1}\\
		\end{pmatrix}
		\begin{pmatrix}
		1 & -2e^{2s^3\hat{g}_+(z)}\\
		0 & 1\\
		\end{pmatrix}\\
		&&\times\begin{pmatrix}
		\beta_++\beta_+^{-1} & i(\beta_+-\beta_+^{-1})\\
		-i(\beta_+-\beta_+^{-1}) & \beta_++\beta_+^{-1}\\
		\end{pmatrix}^{-1}
\end{eqnarray*}
however, as we mentioned previously, for $z\in(-1+\varepsilon,1-\varepsilon)$
\begin{equation*}
	\textnormal{Re}\ \hat{g}_+(z)<0,
\end{equation*}
i.e.
\begin{equation}\label{DZesti2}
	\|N_+\bigl(\begin{smallmatrix}
	1 & -2 \\
	0 & 1\\
	\end{smallmatrix}\bigr)\big(N_+\big)^{-1}-I\|_{L^2\cap L^{\infty}(-1+\varepsilon,1-\varepsilon)}\leq c_5e^{-c_6\varepsilon s},\hspace{0.5cm} s\rightarrow\infty
\end{equation}
also here uniformly on any compact subset of the set \eqref{exceptset2}. Thus together in the limit $s\rightarrow\infty$, with $G_{\Lambda}$ denoting the jump matrix in the $\Lambda$-RHP and $\Sigma_{\Lambda}$ the underlying contour,
\begin{equation}\label{DZesti3}
	\|G_{\Lambda}-I\|_{L^2\cap L^{\infty}(\Sigma_{\Lambda})}\leq c_7s^{-1},\hspace{0.5cm}s\rightarrow\infty
\end{equation}
uniformly on any compact subset of \eqref{exceptset2}. The last estimate provides us with the unique solvability of the $\Lambda$-RHP.

%---------------------------------------------------------------------------------------------------------------------------------------------------------------------------

\section{Solution of the RHP for $\Lambda(z)$ via iteration}\label{sec23}

The given $\Lambda$-RHP is equivalent to the singular integral equation
\begin{equation}\label{integraleq}
	\Lambda_-(z)=I+\frac{1}{2\pi i}\int\limits_{\Sigma_{\Lambda}}\Lambda_-(w)\big(G_{\Lambda}(w)-I\big)\frac{dw}{w-z_-}
\end{equation}
which, by \eqref{DZesti3}, can be solved iteratively in $L^2(\Sigma_{\Lambda})$. Its unique solution satisfies
\begin{equation}\label{DZesti4}
	\|\Lambda_--I\|_{L^2(\Sigma_{\Lambda})}\leq \hat{c}s^{-1}.
\end{equation}
uniformly on any compact subset of \eqref{exceptset2}. Tracing back the transformations
\begin{equation*}
	\Theta(\lambda)\mapsto \Phi(\lambda)\mapsto S(z)\mapsto \Lambda(z)
\end{equation*}
we also obtain existence and boundedness of $\Theta(\lambda),\lambda\in[-s,s]$ and hence existence and boundedness of the resolvent $I+\check{R}_{\textnormal{csin}}$ for sufficiently large $s$ which is needed in \eqref{resolventcsin}. The given information enables us now to determine the large $s$-asymptotics of $\det(I-K_{\textnormal{csin}})$ via Proposition \ref{prop4}, however the computations will be more involved than in section \ref{sec14}.
%---------------------------------------------------------------------------------------------------------------------------------------------------------------------------------------------------------

\section{Asymptotics of $\det(I-K_{\textnormal{csin}})$ - proof of theorem \ref{theo2}}\label{sec24}
Let us start to connect the required coefficients $\Theta_i$ in \eqref{tidentity} to the solution of the $\Lambda$-RHP. Since
\begin{equation*}
	\Theta(\lambda)\sim I+\frac{\Theta_1}{\lambda}+\frac{\Theta_2}{\lambda^2}+\frac{\Theta_3}{\lambda^3}+O\big(\lambda^{-4}\big),\hspace{0.5cm}\lambda\rightarrow\infty
\end{equation*}
we have
\begin{eqnarray*}
	\Theta_1&=&\lim_{\lambda\rightarrow\infty}\Big(\lambda\big(\Theta(\lambda)-I\big)\Big),\hspace{1cm} \Theta_2=\lim_{\lambda\rightarrow\infty}\big(\lambda^2\Big(\Theta(\lambda)-I-\frac{\Theta_1}{\lambda}\Big)\big),\\  \Theta_3&=&\lim_{\lambda\rightarrow\infty}\big(\lambda^3\Big(\Theta(\lambda)-I-\frac{\Theta_1}{\lambda}-\frac{\Theta_2}{\lambda^2}\Big)\big).
\end{eqnarray*}
At this point the following expansions for $z\rightarrow\infty$ are needed.
\begin{equation*}
	N(z) = I+\frac{i}{2z}\begin{pmatrix}
	0 & 1\\
	-1 & 0\\
	\end{pmatrix}+\frac{1}{8z^2}\begin{pmatrix}
	1 & 0\\
	0 & 1\\
	\end{pmatrix}+\frac{3i}{16z^3}\begin{pmatrix}
	0 & 1\\
	-1 & 0\\
	\end{pmatrix}+O\big(z^{-4}\big),
\end{equation*}
as well as
\begin{equation*}
	\hat{g}(z)=i\bigg(\frac{4}{3}tz^3+\frac{xz}{s^2}\bigg)-\frac{i}{2z}\Big(t+\frac{x}{s^2}\Big)-\frac{i}{2z^3}\bigg(\frac{t}{3}+\frac{x}{4s^2}\bigg)+O\big(z^{-5}\big).
\end{equation*}
and also
\begin{eqnarray*}
	\Lambda(z)&=&I+\frac{i}{2\pi z}\int\limits_{\Sigma_{\Lambda}}\Lambda_-(w)\big(G_{\Lambda}(w)-I\big)dw+\frac{i}{2\pi z^2}\int\limits_{\Sigma_{\Lambda}}\Lambda_-(w)\big(G_{\Lambda}(w)-I\big)wdw\\
	&&+\frac{i}{2\pi z^3}\int\limits_{\Sigma_{\Lambda}}\Lambda_-(w)\big(G_{\Lambda}(w)-I\big)w^2dw+O\big(z^{-4}\big).
\end{eqnarray*}
Since
\begin{eqnarray*}
	e^{(s^3\hat{g}(z)-is^3(\frac{4}{3}tz^3+\frac{xz}{s^2}))\sigma_3}&=&I-\frac{i}{2z}(ts^3+xs)\sigma_3-\frac{1}{8z^2}(ts^3+xs)^2I\\
	&&-\frac{i}{2z^3}\bigg(\frac{ts^3}{3}+\frac{xs}{4}\bigg)\sigma_3+\frac{i}{48z^3}(ts^3+xs)^3\sigma_3+O\big(z^{-4}\big)
\end{eqnarray*}
we obtain
\begin{eqnarray*}
	\Theta_1&=&%s\lim_{z\rightarrow\infty}\Big(z\big(\Phi(z)e^{-is^3(\frac{4}{3}tz^3+\frac{xz}{s^2})\sigma_3}-I\big)\Big) = %s\lim_{z\rightarrow\infty}\Big(z\big(\Lambda(z)\Phi^P(z)e^{-is^3(\frac{4}{3}tz^3+\frac{xz}{s^2})\sigma_3}-I\big)\Big)\\
	s\lim_{z\rightarrow\infty}\Big(z\big(\Lambda(z)N(z)e^{(s^3\hat{g}(z)-is^3(\frac{4}{3}tz^3+\frac{xz}{s^2}))\sigma_3}-I\big)\Big)\\
	&=&s\bigg[-\frac{i}{2}(ts^3+xs)\sigma_3-\frac{\sigma_2}{2}+\frac{i}{2\pi}\int\limits_{\Sigma_{\Lambda}}\Lambda_-(w)\big(G_{\Lambda}(w)-I\big)dw\bigg].
\end{eqnarray*}
Similarly
\begin{eqnarray*}
	\Theta_2&=&s^2\bigg[-\frac{1}{8}\big(ts^3+xs\big)^2I-\frac{1}{4}(ts^3+xs)\sigma_1+\frac{I}{8}\\
	&&+\frac{1}{4\pi}\int\limits_{\Sigma_{\Lambda}}\Lambda_-(w)\big(G_{\Lambda}(w)-I\big)dw(ts^3+xs)\sigma_3\\
	&&-\frac{i}{4\pi}\int\limits_{\Sigma_{\Lambda}}\Lambda_-(w)\big(G_{\Lambda}(w)-I\big)dw\sigma_2
	+\frac{i}{2\pi}\int\limits_{\Sigma_{\Lambda}}\Lambda_-(w)\big(G_{\Lambda}(w)-I\big)wdw\bigg]
\end{eqnarray*}
and furthermore
\begin{eqnarray*}
	&&\Theta_3=s^3\bigg[\frac{i}{48}\big(ts^3+xs\big)^3\sigma_3+\frac{1}{16}\big(ts^3+xs\big)^2\sigma_2-\frac{i}{16}(ts^3+xs)\sigma_3\\
	&&-\frac{i}{2}\bigg(\frac{ts^3}{3}+\frac{xs}{4}\bigg)\sigma_3-\frac{3i}{16}\sigma_2
	-\frac{i}{16\pi}\int\limits_{\Sigma_{\Lambda}}\Lambda_-(w)\big(G_{\Lambda}(w)-I\big)dw\big(ts^3+xs\big)^2\\
	&&-\frac{i}{8\pi}\int\limits_{\Sigma_{\Lambda}}\Lambda_-(w)\big(G_{\Lambda}(w)-I\big)dw(ts^3+xs)\sigma_1
	+\frac{i}{16\pi}\int\limits_{\Sigma_{\Lambda}}\Lambda_-(w)\big(G_{\Lambda}(w)-I\big)dw\\
	&&+\frac{1}{4\pi}\int\limits_{\Sigma_{\Lambda}}\Lambda_-(w)\big(G_{\Lambda}(w)-I\big)wdw(ts^3+xs)\sigma_3\\
	&&-\frac{i}{4\pi}\int\limits_{\Sigma_{\Lambda}}\Lambda_-(w)\big(G_{\Lambda}(w)-I\big)wdw\sigma_2
	+\frac{i}{2\pi}\int\limits_{\Sigma_{\Lambda}}\Lambda_-(w)\big(G_{\Lambda}(w)-I\big)w^2dw\bigg].
\end{eqnarray*}
Our next move focuses on the computation of 
\begin{equation*}
	I_n=\int\limits_{\Sigma_{\Lambda}}\Lambda_-(w)\big(G_{\Lambda}(w)-I\big)w^ndw,\ \ n=0,1,2.
\end{equation*}
As we see from \eqref{DZesti3} and \eqref{DZesti4}
\begin{equation*}
	I_n=O\big(s^{-3}\big),\ \ s\rightarrow\infty,\ t>0\hspace{1cm} I_n=O\big(s^{-1}\big),\ \ s\rightarrow\infty,\ t=0
\end{equation*}
on the other hand \eqref{tidentity} has to be evaluated up to $O\big(s^{-1}\big)$ in order to determine $\omega$ in \eqref{theo1}. Hence we need to start iterating \eqref{integraleq}. First in case $t>0$ for $z\in\Sigma_{\Lambda}$
\begin{equation*}
	\Lambda_-(z)-I=\frac{1}{2\pi i}\int\limits_{\hat{C}_r}\big(G_{\Lambda}(w)-I\big)\frac{dw}{w-z_-}+\frac{1}{2\pi i}\int\limits_{\hat{C}_l}\big(G_{\Lambda}(w)-I\big)\frac{dw}{w-z_-}+O\big(s^{-6}\big)
\end{equation*}
and if $t=0$ the latter error term is of order $O\big(s^{-2}\big)$. Thus as $s\rightarrow\infty$
\begin{eqnarray*}
	\Lambda_-(z)-I&=&\frac{1}{2\pi i}\int\limits_{\hat{C}_r}\frac{i}{16\sqrt{\zeta}}\begin{pmatrix} 
	3\beta^{-2}-\beta^2 & i(3\beta^{-2}+\beta^2)\\
	i(3\beta^{-2}+\beta^2) & -(3\beta^{-2}-\beta^2)\\
	\end{pmatrix}\frac{dw}{w-z_-}\\
	&&+\frac{1}{2\pi i}\int\limits_{\hat{C}_l}\frac{1}{16\sqrt{\zeta}}\begin{pmatrix}
	3\beta^2-\beta^{-2} & -i(3\beta^2+\beta^{-2})\\
	-i(3\beta^2+\beta^{-2}) & -(3\beta^2-\beta^{-2})\\
	\end{pmatrix}\frac{dw}{w-z_-}
\end{eqnarray*}
modulo a correction term. Since
\begin{eqnarray*}
	\int\limits_{\hat{C}_r}\frac{\beta^{-2}(w)}{\sqrt{\zeta(w)}(w-z_-)}dw &=& -\frac{3}{4s^3}\bigg(z^2t+\frac{t}{2}+\frac{3x}{4s^2}\bigg)^{-1}\frac{2\pi i}{z+1}\\
	\int\limits_{\hat{C}_r}\frac{\beta^2(w)}{\sqrt{\zeta(w)}(w-z_-)}dw &=&\frac{3}{4s^3}\Bigg(\bigg(\frac{3t}{2}+\frac{3x}{4s^2}\bigg)^{-1}-\bigg(z^2t+\frac{t}{2}+\frac{3x}{4s^2}\bigg)^{-1}\Bigg)\frac{2\pi i}{z-1}\\
	\int\limits_{\hat{C}_l}\frac{\beta^2(w)}{\sqrt{\zeta(w)}(w-z_-)}dw &=&-\frac{3i}{4s^3}\bigg(z^2t+\frac{t}{2}+\frac{3x}{4s^2}\bigg)^{-1}\frac{2\pi i}{z-1}\\
	\int\limits_{\hat{C}_l}\frac{\beta^{-2}(w)}{\sqrt{\zeta(w)}(w-z_-)}dw &=& \frac{3i}{4s^3}\Bigg(\bigg(\frac{3t}{2}+\frac{3x}{4s^2}\bigg)^{-1}-\bigg(z^2t+\frac{t}{2}+\frac{3x}{4s^2}\bigg)^{-1}\Bigg)\frac{2\pi i}{z+1}
\end{eqnarray*}
we obtain
\begin{eqnarray}\label{step1}
	\Lambda_-(z)-I &=& \frac{3i}{64s^3}\Bigg(-2\bigg(z^2t+\frac{t}{2}+\frac{3x}{4s^2}\bigg)^{-1}-\bigg(\frac{3t}{2}+\frac{3x}{4s^2}\bigg)^{-1}\Bigg)\nonumber\\
	&&\times \bigg[\frac{1}{z-1}\begin{pmatrix}
	1 & -i\\
	-i & -1\\
	\end{pmatrix}+\frac{1}{z+1}\begin{pmatrix}
	1 & i \\
	i & -1\\
	\end{pmatrix}\bigg]+O\big(s^{-6}\big)\\
	&\equiv&\frac{3i}{64s^3}f(z,t)\bigg[\frac{1}{z-1}\begin{pmatrix}
	1 & -i\\
	-i & -1\\
	\end{pmatrix}+\frac{1}{z+1}\begin{pmatrix}
	1 & i \\
	i & -1\\
	\end{pmatrix}\bigg]+O\big(s^{-6}\big),\ \ s\rightarrow\infty\nonumber
\end{eqnarray}
for $t>0$ respectively with a correction term of order $O\big(s^{-2}\big)$ in case $t=0$. We are going to improve the latter estimation via iteration
\begin{equation*}
	\Lambda_-(z)-I=\int\limits_{\Sigma_{\Lambda}}\big(\Lambda_-(w)-I\big)\big(G_{\Lambda}(w)-I\big)\frac{dw}{w-z_-}
	+\int\limits_{\Sigma_{\Lambda}}\big(G_{\Lambda}(w)-I\big)\frac{dw}{w-z_-}
\end{equation*}
and the first integral $\Lambda_--I$ is given by \eqref{step1}, thus
\begin{eqnarray*}
	&&\Lambda_-(z)-I =-\frac{1}{2\pi i}\int\limits_{\hat{C}_r}\frac{3i}{32s^3}\frac{f(w,t)}{\sqrt{\zeta(w)}}\begin{pmatrix}
	\frac{\beta^2}{w+1}-\frac{3\beta^{-2}}{w-1}& -i\big(\frac{\beta^2}{w+1}+\frac{3\beta^{-2}}{w-1}\big)\\
	i\big(\frac{\beta^2}{w+1}+\frac{3\beta^{-2}}{w-1}\big) & \frac{\beta^2}{w+1}-\frac{3\beta^{-2}}{w-1}\\
	\end{pmatrix}\frac{dw}{w-z_-}\\
	&&+\frac{1}{2\pi i}\int\limits_{\hat{C}_l}\frac{3}{32s^3}\frac{f(w,t)}{\sqrt{\zeta(w)}}\begin{pmatrix}
	\frac{3\beta^2}{w+1}-\frac{\beta^{-2}}{w-1}&-i\big(\frac{3\beta^2}{w+1}+\frac{\beta^{-2}}{w-1}\big)\\
	i\big(\frac{3\beta^2}{w+1}+\frac{\beta^2}{w-1}\big)&\frac{3\beta^2}{w+1}-\frac{\beta^{-2}}{w-1}\\
	\end{pmatrix}\frac{dw}{w-z_-}\\
	&&+\frac{1}{2\pi i}\int\limits_{\hat{C}_r}\bigg[\frac{i}{16\sqrt{\zeta(w)}}\begin{pmatrix}
	3\beta^{-2}-\beta^2& i(3\beta^{-2}+\beta^2)\\
	i(3\beta^{-2}+\beta^2) & -(3\beta^{-2}-\beta^2)\\
	\end{pmatrix}+\frac{3}{128\zeta(w)}\begin{pmatrix}
	1 & -4i\\
	4i & 1\\
	\end{pmatrix}\bigg]\frac{dw}{w-z_-}\\
	&&+\frac{1}{2\pi i}\int\limits_{\hat{C}_l}\bigg[\frac{1}{16\sqrt{\zeta(w)}}\begin{pmatrix}
	3\beta^2-\beta^{-2}&-i(3\beta^2+\beta^{-2})\\
	-i(3\beta^2+\beta^{-2})&-(3\beta^2-\beta^{-2})\\
	\end{pmatrix}+\frac{3}{128\zeta(w)}\begin{pmatrix}
	-1 & -4i\\
	4i & -1\\
	\end{pmatrix}\bigg]\\
	&&\times\frac{dw}{w-z_-}+O\big(s^{-9}\big)
\end{eqnarray*}
as $s\rightarrow\infty$ for $t>0$ or with a correction term of order $O\big(s^{-3}\big)$ in case $t=0$. Applying residue theorem we compute
\begin{eqnarray*}
	&&\Lambda_-(z)-I = \frac{3i}{64s^3}f(z,t)
	 \bigg[\frac{1}{z-1}\begin{pmatrix}
	1 & -i\\
	-i & -1\\
	\end{pmatrix}+\frac{1}{z+1}\begin{pmatrix}
	1 & i \\
	i & -1\\
	\end{pmatrix}\bigg]\nonumber\\
	&&-\frac{27i}{128s^6}\bigg(\frac{3t}{2}+\frac{3x}{4s^2}\bigg)^{-2}\bigg[\frac{1}{z-1}\begin{pmatrix}
	1+\frac{i}{32}& 2i+\frac{1}{8}\\
	-2i-\frac{1}{8}&1+\frac{i}{32}\\
	\end{pmatrix}-\frac{1}{z+1}\begin{pmatrix}
	1+\frac{i}{32}&-2i-\frac{1}{8}\\
	2i+\frac{1}{8}&1+\frac{i}{32}\\
	\end{pmatrix}\bigg]\\
	&&+\frac{9i}{16s^6}h(z,t)\begin{pmatrix}
	1 & 0\\
	0 & 1\\
	\end{pmatrix}+O\big(s^{-9}\big),\ \ s\rightarrow\infty
\end{eqnarray*}
with
\begin{equation*}
	h(z,t) =\frac{\big(z^2t+\frac{t}{2}+\frac{3x}{4s^2}\big)^{-1}}{(z+1)(z-1)}\Bigg[\bigg(z^2t+\frac{t}{2}+\frac{3x}{4s^2}\bigg)^{-1}\bigg(1+\frac{3i}{64}\bigg)+\bigg(\frac{3t}{2}+\frac{3x}{4s^2}\bigg)^{-1}\Bigg]
\end{equation*}
again for $t>0$ and with an error term of order $O\big(s^{-3}\big)$ in case $t=0$. Having the latter information we will now start to compute $I_0$
\begin{equation*}
	I_0=\int\limits_{\Sigma_{\Lambda}}\big(\Lambda_-(w)-I\big)\big(G_{\Lambda}(w)-I\big)dw +\int\limits_{\Sigma_{\Lambda}}\big(G_{\Lambda}(w)-I\big)dw
\end{equation*}
First
\begin{eqnarray*}
	&&\int\limits_{\Sigma_{\Lambda}}\big(\Lambda_-(w)-I\big)\big(G_{\Lambda}(w)-I\big)dw\\
	&=&\frac{3}{512s^3}\int\limits_{\hat{C}_r}\frac{f(z,t)}{\sqrt{\zeta(z)}}\begin{pmatrix}
	\frac{\beta^2}{z+1}-\frac{3\beta^{-2}}{z-1} & -i\big(\frac{\beta^2}{z+1}+\frac{3\beta^{-2}}{z-1}\big)\\
	i\big(\frac{\beta^2}{z+1}+\frac{3\beta^{-2}}{z-1}\big) & \frac{\beta^2}{z+1}-\frac{3\beta^{-2}}{z-1}\\
	\end{pmatrix}dz\\
	&&+\frac{9i}{8192s^3}\int\limits_{\hat{C}_r}\frac{f(z,t)}{\zeta(z)}\bigg[\frac{1}{z-1}\begin{pmatrix}
	5 & -5i\\
	-5i & -5\\
	\end{pmatrix}+\frac{1}{z+1}\begin{pmatrix}
	-3 & -3i\\
	-3i & 3\\
	\end{pmatrix}\bigg]dz\\
	&&+\frac{27}{2048s^6}\bigg(\frac{3t}{2}+\frac{3x}{4s^2}\bigg)^{-2}\\
	&&\times\Bigg[\begin{pmatrix}
	1 +\frac{i}{32} & 2i+\frac{1}{8}\\
	-2i-\frac{1}{8} & 1+\frac{i}{32}\\
	\end{pmatrix}\int\limits_{\hat{C}_r}\begin{pmatrix}
	3\beta^{-2}-\beta^2 & i(3\beta^{-2}+\beta^2)\\
	i(3\beta^{-2}+\beta^2) & -(3\beta^{-2}-\beta^2)\\
	\end{pmatrix}\frac{dz}{\sqrt{\zeta(z)}(z-1)}\\
	&&-\begin{pmatrix}
	1 +\frac{i}{32} & -2i-\frac{1}{8}\\
	2i+\frac{1}{8} & 1+\frac{i}{32}\\
	\end{pmatrix}\int\limits_{\hat{C}_r}\begin{pmatrix}
	3\beta^{-2}-\beta^2 & i(3\beta^{-2}+\beta^2)\\
	i(3\beta^{-2}+\beta^2) & -(3\beta^{-2}-\beta^2)\\
	\end{pmatrix}\frac{dz}{\sqrt{\zeta(z)}(z+1)}\Bigg]\\
	&&-\frac{9}{256s^6}\int\limits_{\hat{C}_r}\frac{h(z,t)}{\sqrt{\zeta(z)}}\begin{pmatrix}
	3\beta^{-2}-\beta^2 & i(3\beta^{-2}+\beta^2)\\
	i(3\beta^{-2}+\beta^2) &-(3\beta^{-2}-\beta^2)\\
	\end{pmatrix}dz\\
	&&+\frac{3i}{512s^3}\int\limits_{\hat{C}_l}\frac{f(z,t)}{\sqrt{\zeta(z)}}\begin{pmatrix}
		\frac{3\beta^2}{z+1}-\frac{\beta^{-2}}{z-1}& -i\big(\frac{3\beta^2}{z+1}+\frac{\beta^{-2}}{z-1}\big)\\
		i\big(\frac{3\beta^2}{z+1}+\frac{\beta^{-2}}{z-1}\big)& \frac{3\beta^2}{z+1}-\frac{\beta^{-2}}{z-1}\\
		\end{pmatrix}dz\\
		&&+\frac{9i}{8192s^3}\int\limits_{\hat{C}_l}\frac{f(z,t)}{\zeta(z)}\bigg[\frac{1}{z-1}\begin{pmatrix}
		3 & -3i\\
		-3i & -3\\
		\end{pmatrix}+\frac{1}{z+1}\begin{pmatrix}
		-5 & -5i\\
		-5i & 5\\
		\end{pmatrix}\bigg]dz\\
		&&-\frac{27i}{2048s^6}\bigg(\frac{3t}{2}+\frac{3x}{4s^2}\bigg)^{-2}\\
		&&\times \Bigg[\begin{pmatrix}
		1 +\frac{i}{32} & 2i+\frac{1}{8}\\
		-2i-\frac{1}{8} & 1+\frac{i}{32}\\
		\end{pmatrix}\int\limits_{\hat{C}_l}\begin{pmatrix}
		3\beta^2-\beta^{-2} & -i(3\beta^2+\beta^{-2})\\
		-i(3\beta^2+\beta^{-2}) & -(3\beta^2-\beta^{-2})\\
		\end{pmatrix}\frac{dz}{\sqrt{\zeta(z)}(z-1)}\\
		&&-\begin{pmatrix}
		1+\frac{i}{32} & -2i-\frac{1}{8}\\
		2i+\frac{1}{8} & 1+\frac{i}{32}\\
		\end{pmatrix}\int\limits_{\hat{C}_l}\begin{pmatrix}
		3\beta^2-\beta^{-2} & -i(3\beta^2+\beta^{-2})\\
		-i(3\beta^2+\beta^{-2}) & -(3\beta^2-\beta^{-2})\\
		\end{pmatrix}\frac{dz}{\sqrt{\zeta(z)}(z+1)}\Bigg]\\
		&&+\frac{9i}{256s^6}\int\limits_{\hat{C}_l}\frac{h(z,t)}{\sqrt{\zeta(z)}}\begin{pmatrix}
		3\beta^2-\beta^{-2} & -i(3\beta^2+\beta^{-2})\\
		-i(3\beta^2+\beta^{-2}) & -(3\beta^2-\beta^{-2})\\
		\end{pmatrix}dz+O\big(s^{-12}\big),\ s\rightarrow\infty
\end{eqnarray*}
and secondly
\begin{eqnarray*}
	&&\int\limits_{\Sigma_{\Lambda}}\big(G_{\Lambda}(w)-I\big)dw = \int\limits_{\hat{C}_r}\bigg[\frac{i}{16\sqrt{\zeta(z)}}\begin{pmatrix}
	3\beta^{-2} -\beta^2 & i(3\beta^{-2}+\beta^2)\\
	i(3\beta^{-2}+\beta^2) & -(3\beta^{-2}-\beta^2)\\
	\end{pmatrix}\\
	&&+\frac{3}{128\zeta(z)}\begin{pmatrix}
	1 & -4i\\
	4i & 1\\
	\end{pmatrix}\\
	&&+\frac{15i}{2048\zeta(z)^{3/2}}\begin{pmatrix}
	5\beta^2-7\beta^{-2} & -i(5\beta^2+7\beta^{-2})\\
	-i(5\beta^2+7\beta^{-2}) & -(5\beta^2-7\beta^{-2})\\
	\end{pmatrix}\bigg]dz\\
	&&+\int\limits_{\hat{C}_l}\bigg[\frac{1}{16\sqrt{\zeta(z)}}\begin{pmatrix}
	3\beta^{2}-\beta^{-2} & -i(3\beta^2+\beta^{-2})\\
	-i(3\beta^2+\beta^{-2})& -(3\beta^2-\beta^{-2})\\
	\end{pmatrix}+\frac{3}{128\zeta}\begin{pmatrix}
	-1 & -4i\\
	4i & -1\\
	\end{pmatrix}\\
	&&+\frac{15}{2048\zeta(z)^{3/2}}\begin{pmatrix}
		7\beta^2-5\beta^{-2} & -i(7\beta^2+5\beta^{-2})\\
		-i(7\beta^2+5\beta^{-2})&-(7\beta^2-5\beta^{-2})\\
		\end{pmatrix}\bigg]dz+O\big(s^{-12}\big),\ s\rightarrow\infty.
\end{eqnarray*}
All integrals can be evaluated via residue theorem, we summarize the results
\begin{eqnarray*}
	&&\frac{1}{2\pi i}\int\limits_{\Sigma_{\Lambda}}\big(G_{\Lambda}(w)-I\big)dw = \frac{3i}{32s^3}\bigg(\frac{3t}{2}+\frac{3x}{4s^2}\bigg)^{-1}\sigma_3-\frac{27}{512s^6}\bigg(\frac{3t}{2}+\frac{3x}{4s^2}\bigg)^{-2}\sigma_2\\	&&+\frac{405i}{65536s^9}\bigg(\frac{3t}{2}+\frac{3x}{4s^2}\bigg)^{-3}\Bigg[\frac{5}{4}\frac{\frac{27t}{2}+\frac{3x}{4s^2}}{\frac{3t}{2}+\frac{3x}{4s^2}}+\frac{7}{4}\Bigg]\sigma_3+O\big(s^{-12}\big),
\end{eqnarray*}
as $s\rightarrow\infty$ for $t>0$ and with an error term of order $O\big(s^{-4}\big)$ in case $t=0$. Similarly
\begin{eqnarray*}
	&&\frac{1}{2\pi i}\int\limits_{\hat{C}_{r,l}}\big(\Lambda_-(w)-I\big)\big(G_{\Lambda}(w)-I\big)dw=\frac{27}{512 s^6}\bigg(\frac{3t}{2}+\frac{3x}{4s^2}\bigg)^{-2}\sigma_2\\	&&+\frac{81}{1024s^9}\bigg(\frac{3t}{2}+\frac{3x}{4s^2}\bigg)^{-3}\bigg(1-\frac{3i}{256}\bigg)\sigma_3+\frac{81}{8192s^9}\bigg(\frac{3t}{2}+\frac{3x}{4s^2}\bigg)^{-3}\bigg(1-\frac{5i}{32}\bigg)\sigma_3\\	&&+\frac{27}{32768s^9}\bigg(\frac{3t}{2}+\frac{3x}{4s^2}\bigg)^{-4}\bigg[49t+3it+\Big(\frac{211}{2}+3i\Big)\bigg(\frac{19t}{2}+\frac{3x}{4s^2}\bigg)\bigg]\sigma_3\\
	&&-\frac{81}{8192s^9}\bigg(\frac{3t}{2}+\frac{3x}{4s^2}\bigg)^{-3}\bigg[\bigg(1+\frac{i}{32}\bigg)\Big(3i-4it\bigg(\frac{3t}{2}+\frac{3x}{4s^2}\bigg)^{-1}\Big)\\
	&&-\bigg(2i+\frac{1}{8}\bigg)\Big(3+4t\bigg(\frac{3t}{2}+\frac{3x}{4s^2}\bigg)^{-1}\Big)\bigg]\sigma_1+O\big(s^{-12}\big),\ \ \ s\rightarrow\infty
\end{eqnarray*}
and we summarize
\begin{equation}\label{evaleq1}
	\frac{I_0}{2\pi i} = \frac{3i}{32s^3}\bigg(\frac{3t}{2}+\frac{3x}{4s^2}\bigg)^{-1}\sigma_3+\frac{27}{64s^9}\bigg(\frac{3t}{2}+\frac{3x}{4s^2}\bigg)^{-3}\big(a(s,t)\sigma_3+b(s,t)\sigma_1\big)+O\big(s^{-12}\big)
\end{equation}
as $s\rightarrow\infty$ for $t>0$ respectively with an error term of order $O\big(s^{-4}\big)$ for $t=0$. Here the functions $a=a(s,t)$ and $b=b(s,t)$ can be read of from the previous lines, we state $I_0$ in this form since as we will see, only the structure of the term of order $O\big(s^{-9}\big)$ matters. Moving on to $I_1$ and $I_2$ similar computations imply
\begin{equation}\label{evaleq2}
	\frac{I_1}{2\pi i} =\frac{3}{32s^3}\bigg(\frac{3t}{2}+\frac{3x}{4s^2}\bigg)^{-1}\sigma_1-\frac{81}{2048s^6}\bigg(\frac{3t}{2}+\frac{3x}{4s^2}\bigg)^{-2}I+O\big(s^{-9}\big)
\end{equation}
and
\begin{equation}\label{evaleq3}
	\frac{I_2}{2\pi i}=\frac{3i}{32s^3}\bigg(\frac{3t}{2}+\frac{3x}{4s^2}\bigg)^{-1}\sigma_3 +O\big(s^{-9}\big)
\end{equation}
in the limit $s\rightarrow\infty$ for $t>0$ or in case $t=0$ with adjusted error terms. It is now straight forward to use the given information \eqref{evaleq1},\eqref{evaleq2} and \eqref{evaleq3} to obtain the large $s$-asymptotics for $\Theta_1,\Theta_2$ and $\Theta_3$. Once we have the latter expansions we go back to \eqref{tidentity}
\begin{eqnarray*}
	&&\frac{4i}{3}\textnormal{trace}\ \big(-3\Theta_3\sigma_3\big)=\frac{s^3}{6}\big(ts^3+xs\big)^3-\frac{s^3}{2}\big(ts^3+xs\big)-4s^3\bigg(\frac{ts^3}{3}+\frac{xs}{4}\bigg)\\
	&&+\frac{3}{32}\bigg(\frac{3t}{2}+\frac{3x}{4s^2}\bigg)^{-1}\big(ts^3+xs\big)^2-\frac{27i}{64s^6}\bigg(\frac{3t}{2}+\frac{3x}{4s^2}\bigg)^{-3}\big(ts^3+xs\big)^2a(s,t)\\
	&&-\frac{15}{32}\bigg(\frac{3t}{2}+\frac{3x}{4s^2}\bigg)^{-1}-\frac{81}{512s^3}\bigg(\frac{3t}{2}+\frac{3x}{4s^2}\bigg)^{-2}\big(ts^3+xs\big)+O\big(s^{-3}\big),\ \ s\rightarrow\infty
\end{eqnarray*}
and
\begin{eqnarray*}
	&&\frac{4i}{3}\textnormal{trace}\ \big(2\Theta_2\sigma_3\Theta_1\big)=-\frac{3}{16}\bigg(\frac{3t}{2}+\frac{3x}{4s^2}\bigg)^{-1}\big(ts^3+xs\big)^2\\	&&+\frac{27i}{32s^6}\bigg(\frac{3t}{2}+\frac{3x}{4s^2}\bigg)^{-3}\big(ts^3+xs\big)^2a(s,t)-\frac{s^3}{3}\big(ts^3+xs\big)^3+s^3\big(ts^3+xs\big)\\
	&&+\frac{3}{16}\bigg(\frac{3t}{2}+\frac{3x}{4s^2}\bigg)^{-1}+\frac{21}{256s^3}\bigg(\frac{3t}{2}+\frac{3x}{4s^2}\bigg)^{-2}\big(ts^3+xs\big)+O\big(s^{-3}\big)
\end{eqnarray*}
as well as
\begin{eqnarray*}
	&&\frac{4i}{3}\textnormal{trace}\ \big(-\Theta_1\sigma_3(\Theta_1^2-\Theta_2)\big)=\frac{39}{512s^3}\bigg(\frac{3t}{2}+\frac{3x}{4s^2}\bigg)^{-2}\big(ts^3+xs\big)-\frac{3}{32}\bigg(\frac{3t}{2}+\frac{3x}{4s^2}\bigg)^{-1}\\
	&&+\frac{3}{32}\bigg(\frac{3t}{2}+\frac{3x}{4s^2}\bigg)^{-1}\big(ts^3+xs\big)^2-\frac{27i}{64s^6}\bigg(\frac{3t}{2}
	+\frac{3x}{4s^2}\bigg)^{-3}\big(ts^3+xs\big)^2a(s,t)\\
	&&-\frac{s^3}{2}\big(ts^3+xs\big)+\frac{s^6}{6}\big(ts^3+xs\big)^3+O\big(s^{-3}\big)
\end{eqnarray*}
in all cases as $s\rightarrow\infty$ uniformly on any compact subset of the set \eqref{exceptset2}. Now use Proposition \ref{prop4} and add up the latter three identities 
\begin{equation*}
	\frac{d}{dt}\ln\det(I-K_{\textnormal{csin}}) = -\frac{4}{3}ts^6-xs^4-\frac{3}{8}\bigg(\frac{3t}{2}+\frac{3x}{4s^2}\bigg)^{-1}+O\big(s^{-3}\big),\ \ s\rightarrow\infty
\end{equation*}
uniformly on any compact subset of the set \eqref{exceptset2}. We integrate and obtain
\begin{eqnarray*}
	\int\limits_0^1\frac{d}{dt}\ln\det(I-K_{\textnormal{csin}})\ dt &=& -\frac{2}{3}s^6-s^4x-\frac{1}{4}\ln\bigg(\frac{3}{2}+\frac{3x}{4s^2}\bigg)+\frac{1}{4}\ln\bigg(\frac{3x}{4s^2}\bigg)+O\big(s^{-3}\big)\\
	&=&-\frac{2}{3}s^6-s^4x-\frac{1}{2}\ln s+\frac{1}{4}\ln x-\frac{1}{4}\ln 2+O\big(s^{-2}\big),\ s\rightarrow\infty.
\end{eqnarray*}
On the other hand
\begin{equation*}
	\int\limits_0^1\frac{d}{dt}\ln\det(I-K_{\textnormal{csin}}) = \ln\det(I-\check{K}_{\textnormal{csin}})-\ln\det(I-K_{\sin})
\end{equation*}
and we know (see \cite{dyson})
\begin{equation*}
	\ln\det(I-K_{\sin}) = -\frac{(xs)^2}{2}-\frac{1}{4}\ln (sx)  +\frac{1}{12}\ln 2+3\zeta'(-1)+O\big(s^{-1}\big),\ s\rightarrow\infty
\end{equation*}
hence together as $s\rightarrow\infty$
\begin{equation*}
	\ln\det(I-\check{K}_{\textnormal{csin}}) = -\frac{2}{3}s^6-xs^4-\frac{(xs)^2}{2}-\frac{3}{4}\ln s-\frac{1}{6}\ln 2+3\zeta'(-1)+O\big(s^{-1}\big)
\end{equation*}
and the error term is uniform on any compact subset of the set \eqref{excset1}. This proves Theorem \ref{theo2}.
%----------------------------------------------------------------------------------------------------------------------------------------------------------------------------------

\section{Asymptotics of $\det(I-K_{\textnormal{PII}})$ - proof of Theorem \ref{t1}}\label{sec25}

Recalling the discussion of section \ref{sec15} and the result of Theorem \ref{theo2} we obtain immediately
\begin{equation*}
	\ln\det(I-\check{K}_{\textnormal{PII}})=-\frac{2}{3}s^6-s^4x-\frac{3}{4}\ln s+\int\limits_x^{\infty}(y-x)u^2(y)dy + \omega+O\big(s^{-1}\big)
\end{equation*}
uniformly on any compact subset of the set \eqref{excset1} with $\omega=-\frac{1}{6}\ln 2+3\zeta'(-1)$. This proves Theorem \ref{t1}.

%--------------------------------------------------------------------------------------------------------------------------------------------------


\bigskip

\begin{thebibliography}{299}

\bibitem{BE}
H. Bateman, A. Erdelyi, {\it Higher Transcendental Functions,}
McGraw-Hill, NY, 1953.

\bibitem{BH} E. Br\'ezin, S. Hikami, Level spacing of random mtrices in an 
external source, {\it Phys. Rev. E} {\bf 58} (1998), 7176 - 7185

\bibitem{BI1} P. Bleher and A. Its, Semiclassical asymptotics of orthogonal polynomials,
Riemann-Hilbert problem, and universality in the matrix model,
{\it Annals Math.}{\bf  150} (1999), 185-266.

\bibitem{BI2} P. Bleher and A. Its, Double scaling limit in the random matrix model:
the Riemann-Hilbert approach, {\it Comm. Pure Appl. Math.} {\bf 56} (2003),
433-516.

\bibitem{CET} Y. Chen, K. J. Eriksen, C. A. Tracy, Largest eigenvalue distribution in the double
scaling limit of matrix models: a Coulomb fluid approach, {\it J. Phys A}
{\bf 28} (1995), L207 - L211.

\bibitem{CK} T. Claeys and A. Kuijlaars, Universality of the double
scaling limit in random matrix models, {\it Comm. Pure Appl. Math.},
{\bf 59}, no. 11 (2006), 1573-1603. 

\bibitem{CIK} T. Claeys, A. Its, ad I. Krasovsky,
 Higher order analogues of the Tracy-Widom distribution 
 and the Painlev\'e II hierarchy, {\it Comm. Pure. Appl. Math.}
 {\bf 63} no. 3 (2010) Pages 362 - 412


\bibitem{D} P. Deift, {\it Orthogonal Polynomials and Random Matrices: A Riemann-
Hilbert Approach,} Courant Lecture Notes in Mathematics 3, Amer.
Math. Soc., Providence, RI, 1999.

\bibitem{DIKZ} P. Deift, A. Its, I. Krasovsky, and X. Zhou, The Widom-Dyson constant
for the gap probability in random matrix theory, {\it J. Comput. Appl. Math.} {\bf 202} (2007), no. 1, 26-47

\bibitem{DIKa} P. Deift, A. Its, I. Krasovsky, Asymptotics of the Airy-kernel determinant,
{\it Communications in Mathematical Physics,} {\bf 278} (2008), no. 3, 643-678 

\bibitem{DIZ} P. Deift, A. Its, and X. Zhou, A Riemann-Hilbert approach
to asymptotic problems arising in the theory of random matrix models, and also in the
theory of integrable statistical mechanics, {\it Ann. Math.} {\bf 146} (1997), 149-235

\bibitem{DKM} P. Deift, T. Kriecherbauer, and K.T-R McLaughlin, New results on the
equilibrium measure for logarithmic potentials in the presence of an
external field, {\it J. Approx. Theory} {\bf 95} (1998),  388-475

\bibitem{DKMVZ} P. Deift, T. Kriecherbauer, K.T-R McLaughlin, S. Venakides, and X.
Zhou, Uniform asymptotics for polynomials orthogonal with respect to
varying exponential weights and applications to universality questions in
random matrix theory, {\it Comm. Pure Appl. Math.} {\bf 52} (1999), 1335-1425

\bibitem{DZ1}
P.A. Deift and X. Zhou, A steepest descent method for oscillatory
Riemann-Hilbert problems. Asymptotics for the MKdV equation, {\it
Ann.\ of Math.}, {\bf 137} (1993) 295--368.

\bibitem{dyson} F. Dyson, Fredholm determinants and inverse scattering problems,{\it Commun. Math. Phys.} {\bf 47} (1976) 171-183

\bibitem{E} T. Ehrhardt, Dyson's constant in the asymptotics of the
Fredholm determinant of the sine kernel, {\it Comm. Math. Phys.} {\bf 262} (2006), 317-341 

\bibitem{FIKN} A. Fokas, A. Its, A. Kapaev, V. Novokshenov, {\it Painlev\'e transcendents. The Riemann-Hilbert approach,}
AMS series: Mathematical surveys and monographs, {\bf 128} (2006)
553p.

\bibitem{FN}
H.~Flaschka and A.C.~Newell,  Monodromy- and spectrum-preserving
deformations I, {\it Comm.\ Math.\ Phys.} {\bf 76} (1980) 65--116.

\bibitem{For} P. J. Forrester, {\it Nucl. Phys.} B {\bf 402} (1993) 709

\bibitem{HM} S.P. Hastings, and J.B. McLeod, A boundary value problem associated
with the second Painlev\'e transcendent and the Korteweg-de Vries equation,
{\it Arch. Rational Mech. Anal.} {\bf 73} (1980), 31-51.

\bibitem{IIKS} A. Its, A. Izergin, V. Korepin, and N. Slavnov,
Differential equations for quantum correlation functions, {\it Int. J. Mod. Physics} {\bf B4} (1990), 1003-1037

\bibitem{JMU} M. Jimbo, T. Miwa, and K. Ueno, Monodromy preserving deformation of linear
ordinary differential equations with rational coefficients. I. General theory and $\tau$-function,
{\it Phys. D 2}  {\bf no. 2} (1981), 306-352.

\bibitem{KKMST} N. Kitanine, K. Kozlowski, J. Maillet, N. Slavnov, V. Terras, Riemann-Hilbert approach
to a generalised sine kernel and applications, {\it Comm. Math. Phys.} {\bf 291} (2009) 691-761

\bibitem{K} I. Krasovsky, Gap probability in the spectrum of random matrices and asymptotics
of polynomials orthogonal on an arc of the unit circle, {\it Int. Math. Res. Not.} {\bf 2004} (2004), 1249-1272

\bibitem{M} M.L. Mehta, {\it Random Matrices,} 2nd edition, Academic Press, Boston,
1991.

\bibitem{PS} L. Pastur, M. Shcherbina, Universality of the local eigenvalue statistics for a class of  unitary
invariant random matrix ensembles. 1996
 
\bibitem{S} B. Simon, {\it Trace ideals and their applications,} London Mathematical Society Lecture Note Series,
{\bf 35}. Cambridge University Press, Cambridge-New York, 1979



\end{thebibliography}
\end{document}